\documentclass[12pt]{article}
\pdfoutput=1

\usepackage{jjones}
\usepackage[section]{placeins}

\setlist[enumerate]{label=\upshape(\roman*)}

\newlist{theoremlist}{enumerate}{1}
\setlist[theoremlist]{label=\upshape(\roman*), ref=\thetheorem (\roman*), before=\leavevmode, parsep=0\parsep, itemsep=0.5\itemsep}
\crefname{theoremlisti}{theorem}{theorems}

\newlist{asulist}{enumerate}{1}
\setlist[asulist]{label=\upshape(\roman*), ref=\theassumption (\roman*), before=\leavevmode, parsep=0\parsep, itemsep=1\itemsep}
\crefname{asulisti}{assumption}{assumptions}

\def\mathspaceand{\hspace*{5pt}\text{and}\hspace*{5pt}}

\newcommand{\subj}[1][j]{ (\{#1\}) }

\def\mytitle{Causal Mediation Analysis: Selection with Asymptotically Valid Inference}

\title{ \mytitle}
\author
{Jeremiah Jones, Ashkan Ertefaie and Robert L. Strawderman}

\def\work{paper }

\def\correctThmDisplayMath{\vspace{-1.9em}}

\begin{document}

\maketitle

\begin{abstract}
Researchers are often interested in learning not only the effect of treatments on outcomes, but also the pathways through which these effects operate. A mediator is a variable that is affected by treatment and subsequently affects outcome. Existing methods for penalized mediation analyses may lead to ignoring important mediators and either assume that finite-dimensional linear models are sufficient to remove confounding bias, or perform no confounding control at all. In practice, these assumptions may not hold. We propose a method that considers the confounding functions as nuisance parameters to be estimated using data-adaptive methods. We then use a novel regularization method applied to this objective function to identify a set of important mediators. We derive the asymptotic properties of our estimator and establish the oracle property under certain assumptions. Asymptotic results are also presented in a local setting which contrast the proposal with the standard adaptive lasso. 
We also propose a perturbation bootstrap technique to provide asymptotically valid post-selection inference for the mediated effects of interest. The performance of these methods will be discussed and demonstrated through simulation studies.

\end{abstract}


\section{Introduction}

Researchers are often interested in learning about the mechanisms by which interventions or exposures affect outcomes. Mediation analysis seeks to characterize the effect of treatment through different pathways -- the effects directly resulting from the treatment, and those which are \textit{mediated} through other variables. 
A mediator is a variable that is affected by the treatment and subsequently affects the outcome. 
Due in part to contributions from the field of causal inference, much progress has been made in understanding sufficient assumptions for identifying various mediated effects, and statistical methods for estimating these quantities. These contributions have focused almost entirely on the case of a single mediator 
\citep{robinsIdentifiabilityExchangeabilityDirect1992,vanderlaanDirectEffectModels2008,vanderweeleConceptualIssuesConcerning2009,vanderweeleOddsRatiosMediation2010}, 
with some attention given to the case of multiple mediators. \cite{vanderweeleMediationAnalysisMultiple2014}
proposed a regression-based approach to assess the effect of an exposure on an outcome through several mediators considered jointly. 

More recently, methodology has been developed for the study of mediation in high-dimensional settings.
However, as illustrated in our simulation studies, standard regularization methods which do not account for the mediator-exposure relationship may lead to ignoring important mediators. Some examples of these regularization-based selection techniques are the adaptive lasso \citep{zouAdaptiveLassoIts2006a} and the minimax concave penalty \citep{zhangNearlyUnbiasedVariable2010}.
A proposal by \cite{zhangEstimatingTestingHighdimensional2016} used the minimax concave penalty along with pre-screening in high-dimensional epigenetic studies, basing the selection solely upon the outcome model.
More recent methods attempt to improve upon standard selection techniques by developing methods specific to the problem of mediator selection. 
\cite{zhaoPathwayLassoEstimate2016} proposed a method to perform sparse estimation of mediation pathway coefficients using a modified lasso penalty that directly penalized the estimated contribution to a causal estimand. \cite{schaidPenalizedModelsAnalysis2020} critiqued this approach for penalizing true mediators more heavily than non-mediators, and instead proposed a penalty to the full-data likelihood function based on the sparse group lasso. 
Importantly, neither of these methods explicitly handle measured confounders, which can invalidate causal inferences in mediation analysis. 
Further, the theoretical properties of these selection procedures have not yet, to the authors' knowledge, been established.
Instead of a selection-based methodology, \cite{songBayesianShrinkageEstimation2020} formulated a Bayesian prior that performs ridge-like shrinkage in high-dimensional settings.  
Other alternatives to penalized estimators include dimensionality reduction-based methods. These require transformations of the mediators and hence may negatively impact interpretability of the results \citep{chenHighdimensionalMultivariateMediation2017, huangHypothesisTestMediation2016}.   
All the methods discussed above assume that either no confounding is present, or that a postulated finite-dimensional linear model for confounders is sufficient to control for the residual confounding in the mediator and outcome models.

In the high dimensional settings of some of this previous work, techniques are broadly proposed from the perspective of combatting the curse of dimensionality and allowing the estimation of mediation effects. On the other hand, there may be scientific interest in identifying and interpreting parsimonious pathways through which causal effects are mediated, even when the dimensionality of the problem is not severe. From this perspective, methods for performing mediator selection are valued not based on an improvement in estimation, but based on an improvement in scientific understanding or interpretability. In this paper, we propose a framework for mediator selection focused on this scientific goal.

Our proposed framework for mediation analysis incorporates mediator selection and nonparametric confounding control. Specifically, we propose an estimating procedure where the functional association form of the confounder-treatment, confounder-mediators and confounder-outcome associations are considered infinite-dimensional nuisance parameters which are estimated using data-adaptive techniques. 
The parametrization we leverage de-couples the estimation of the focal parameters from confounding control. 
This allows a wide variety of techniques such as machine learning tools to be employed as confounding control, as long as the corresponding nuisance parameters are consistently estimated with certain rates.
We show that our estimators are asymptotically linear under certain assumptions. 
We demonstrate through simulation studies that the commonly-employed parametric linear confounding control suffers from severe confounding bias; we further show that such bias can be mitigated by the proposed method.
In addition, we develop regularization tailored to mediation, which can be used to identify the true mediators among a set of candidates. In \Cref{sec:theory} we establish the variable selection consistency of this approach---to our knowledge, the first result of this kind in the domain.
A perturbation bootstrap approach for inference after selection is presented which extends the results of \cite{minnierPerturbationMethodInference2011} to this setting.
A closer study of asymptotic behavior follows in \Cref{sec:local-asymp} using a local asymptotic framework that better approximates the situation where some parameters are close to zero. Our proposed estimator is shown to maintain mediator selection consistency while a standard regularization technique, the adaptive lasso, is shown to fail under the same scenario when the penalty term only involves coefficients in the outcome model.
We conduct extensive simulation studies to highlight the importance of the proposal and study the small sample performance of proposed estimators. In particular, we show: (1) failing to properly adjust for confounders can lead to a biased estimate of the mediated effects; (2) using standard regularization methods (e.g., adaptive lasso) that focus solely on the outcome model may lead to ignoring important mediators, and; (3) the perturbation bootstrap leads to confidence intervals close to the nominal level when combined with our novel estimator.

Our method is applied to a data set in education. The STAR study was a randomized controlled trial that sought to identify if small class sizes were responsible for measurably improving student outcomes in Grades K-3. A seminal analysis by \cite{kruegerExperimentalEstimatesEducation1999c} has shown that small class sizes lead to demonstrable improvements in student outcomes. 
A mediation analysis may answer the question: through which pathways does the causal effect of small class sizes travel to improve these test scores?
We report the results of our method on simultaneous selection of mediators and estimation of corresponding effects when applied to two mathematics standardized test outcomes in Grade 8: Computation, and Concepts \& Applications.

\section{Defining the Multiple-Mediator Problem}\label{sec:notation}

\subsection{Model Statement}
\label{sec:notation-model-statement}

In this section, we will formally define the mediation model and the notation to be used throughout the rest of the paper. Let $\Vert \bg \Vert_{q}$ represent the $\ell_q$ norm of the vector $\bg$ for $q=1,2,\ldots,\infty$  and let $\Vert C \Vert_{\infty}$ represent the maximal element of the matrix $C$. Let $\bg \cdot \bh$ represent the dot product of vectors $\bg$ and $\bh.$ For a distribution $P$ and set $\Xcal$, we say $g \in L_2(P)(\Xcal)$ if $g: \Xcal \mapsto \R$ and $\int_{\Xcal} g^2(x) dP(x) < \infty$. 
For any indexing set $\model$, we will use the conventions of \cite{kuchibhotlaValidPostselectionInference2020c} to refer to subsetted vectors and matrices: $\bH(\model)$ will refer to the principal submatrix of the square matrix $\bH$ constructed using the indices in $\model$, and $\bg(\model)$ will refer to a similarly constructed subvector of the vector $\bg \in \R^p$. As a special case, $\bH\subj[j]$ will refer to the $j^{th}$ diagonal element of $\bH$ and $\bg\subj[j]$ will similarly represent the $j^{th}$ component of $\bg$. 

Suppose we observe data $\bO_i\equiv(D_i, \bX_i, \bM_i, Y_i)$ for $i=1,\ldots,n$, which are i.i.d. according to some unknown distribution $P_0$. For $\bO\equiv(D, \bX, \bM, Y) \sim P_0,$ $D \in \lbrace 0, 1 \rbrace$ is a binary treatment, $\bX \in \R^{q}$ represents a set of confounding variables, $Y \in \R$ represents the outcome, and $\bM \in \R^p$ represents a set of post-treatment covariates which are candidates for mediating the treatment-outcome relationship. We assume the observed-data distribution $P_0$ satisfies:
\begin{equation}
	\begin{aligned}
    Y &= \bZ^\top \btheta_0 + \psi_Y(\bX) + \epsilon & \\
    \bM &= \balpha_{0} D + \bm{\psi}_{M}(\bX) + \bm{\eta}, &&
  \end{aligned}
  \label{eq:med-model}
\end{equation}
where we have partitioned $\bZ = (D, \bM^\top)^\top$ and $\btheta_0 = (\gamma_0, \bbeta_0^\top)^\top$.
In this display, $\bm{\eta}$ and $\epsilon$ satisfy 
$\Eop(\bm{\eta} \vert \bX, D)=\bzero$, $\Eop(\epsilon \vert \bX, D, \bM) = 0$,
$\Vert \Var(\bm{\eta} \vert \bX, D) \Vert_{\infty}<\infty$, and $\Var(\epsilon \vert \bX, D, \bM) < \infty$. These requirements are discussed further in \Cref{sec:ch2-notation-estimand}. 
The unknown functions $\psi_Y(\bX) \in L_2(P_0)(\Xcal)$ and $\bm{\psi}_{M}(\bX) \in [L_2(P_0)(\Xcal)]^p$ represent the confounding effects of $\bX$ on each of the respective variables. 
In the representation of the previous model, we allow that elements of the Euclidean parameters $\balpha_0$ and $\btheta_0$ may be exactly zero; this will motivate the development of our selection methodology.

The model \eqref{eq:med-model} expresses each of $Y$ and $\bM$ in terms of additive confounding functions of $\bX$. These functions can be difficult to adequately handle, as they are not estimable from the observed random variables without control of the linear terms. One approach is to re-express the system \eqref{eq:med-model} using Robinson's transformation \citep{robinsonRootnconsistentSemiparametricRegression1988}. This transformation introduces $p + 2$ infinite-dimensional nuisance parameters in exchange for removing all of the explicit dependence on the $\psi$ functions in \eqref{eq:med-model}. The introduced parameters are identified as conditional expectations which are directly estimable, separately from the Euclidean parameters $\balpha_0$ and $\btheta_0$. In this sense, the Robinson-transformed system is a re-parametrization of the original additive system.
In Section \ref{sec:conf-control}, estimation based on this re-parametrization will be discussed and in Section \ref{sec:theory} we will demonstrate that this framework allows $\sqrt{n}$-consistent, asymptotically normal estimation of the model parameters.

We apply Robinson's transformation by taking conditional expectations and then subtracting each from both sides of the equality in the respective lines of \eqref{eq:med-model}. This yields the system
\begin{equation}
  \begin{aligned}
    Y - \mu_{Y0}(\bX) =& \{\bZ - \bmu_{Z0}(\bX)\}^\top \btheta_{0} + \epsilon \\
    \bM - \bmu_{M0}(\bX) =& \{D - \mu_{D0}(\bX)\} \balpha_{0} + \bm{\eta},
  \end{aligned}
  \label{eq:one-stage-true}
\end{equation}
where $\mu_{D0}(\bX) = \Eop(D~|~\bX)$ is the treatment propensity, $\mu_{Y0}(\bX) = \Eop(Y~|~\bX)$, $\bmu_{Z0}(\bX)=\Eop(\bZ ~|~ \bX)$, and $\bmu_{M0}(\bX)=\Eop(\bM ~|~ \bX)$, and the boldface font represents that the conditional expectation is vector-valued. We collect all of these functions into $\bm{\mu}_0 = (\mu_{Y0}, \bmu_{Z0}^\top)^\top$, where $\bmu_{Z0} = (\mu_{D0}, \bmu_{M0}^\top)^\top$.

Our interest lies in the setting where it is known that there exists some submodel $\model^* \in \modelspace$ that corresponds to a ``good'' set of mediator variables, where $\modelspace$ is the set of all subsets of $\model^F := \{1,\ldots,p\}$. In this sense, $\modelspace$ may be considered a set of possible mediator models, with $\model^* \in \modelspace$ a particularly desirable one. Consequently, this lends $\model^{*c}$ the interpretation of a set of spurious or unimportant post-treatment variables for the mediation problem, where $\model^c \defined \model^F \setminus \model$ for any $\model \in \modelspace$.
In terms of \eqref{eq:med-model} and \eqref{eq:one-stage-true}, we define $\model^* \defined \{ j : \balpha_0\subj[j] \bbeta_0 \subj[j] \neq 0 \}$ and defer the motivation to \Cref{sec:ch2-notation-full-model}. 
For any $\model \in \modelspace$, we will use the previously-defined subsetting notation for quantities which are subsetted.
Submodel-defined objects will be subscripted, such as parameters $\btheta_{0\model}$ which may not always be defined as a subvector of the full $\btheta_0$ (i.e. $\btheta_{0\model^*}\neq\btheta_0(\model^*)$, in general).

Since we are only interested in selection on $\bM$, but the model is presented in terms of $\bZ=(D, \bM^\top)^\top$, we use the subscript $Z$ to map any index $\model$ for $\bM$ into an index for $\bZ$ which always includes the element $1$, or equivalently always includes the treatment $D$. Specifically, for any $\model \in \modelspace$, $\model_Z$ subsets $\bZ$ as $\bZ(\model_Z) = (D, \bM(\model)^\top)^\top$.
The two main cases which we encounter are $\model_Z^*$ and $\model_Z^F$, which are the indices in $\bZ$ born from $\model^*$ and $\model^F$, respectively. 

\subsection{Causal Estimand}
\label{sec:ch2-notation-estimand}

The assumption $\model^* \subseteq \model^F$ motivates identifying a lower-dimensional set of variables which mediate the effect of treatment on the outcome. The relevant mediation effects are defined for this subset $\model^*$ through the counterfactual or potential outcomes language of causal inference.
Let $Y_{d,\bmm_{\model^*}}$ be the potential outcome of $Y$ had the treatment been set to $D=d$ and the $\model^*$-related variables $\bM({\model^*})$ been set to $\bM({\model^*}) = \bmm_{\model^*}$.
We stress that these counterfactuals represent direct interventions only upon $D$ and $\bM(\model^*)$; the remaining $\bM(\model^{*c})$ are not controlled and may take whichever values naturally occur as a result of the intervention.
The subscript $\bmm_{\model^*}$ was chosen to reflect that these counterfactuals corresponding to $\model^*$ are fundamentally different than those related to $\model^F$, although one might also use the selection notation to reflect that only a subset of the full variables in $\bM$ are directly controlled.
Lastly, we will make use of the notation $\bM_{d}(\model^*)$ for the counterfactual value of $\bM({\model^*})$ had the treatment been set to $D=d$. 

Two causal quantities that we consider throughout this \work are the \textit{natural direct effect} and \textit{natural indirect effect}. 
Following \cite{vanderweeleMediationAnalysisMultiple2014}, these two quantities are respectively defined relative to $\model^*$ as
\begin{align}
  \label{eq:nde-causal}
  \nde_{\model^*} =& \E{ 
    Y_{1,\bM_{0}(\model^*) } - Y_{0,\bM_{0}(\model^*) }
  } \\
  \label{eq:nie-causal}
  \nie_{\model^*} =& \E{
    Y_{1,\bM_{1}(\model^*) }
    - Y_{1,\bM_{0}(\model^*) }
  }.
\end{align}
The $\nde_{\model^*}$ compares the expected difference in the counterfactual outcomes between treatments $D=1$ and $D=0$, while intervening to fix the value of the candidate mediators to $\bM(\model^*) = \bM_{0}(\model^*)$, or that which would have been observed under treatment $0$. The $\nie_{\model^*}$ captures the expected difference in the counterfactual outcomes obtained by holding the treatment constant at $D=1$ and intervening to set the candidate mediators to levels $\bM(\model^*) = \bM_{1}(\model^*)$ and $\bM(\model^*) = \bM_{0}(\model^*)$. The causal effects \eqref{eq:nde-causal} and \eqref{eq:nie-causal}, defined for the set of variables in $\model^*$, can be identified and estimated from the observed data when the following assumptions hold:
\begin{assumption}[Causal assumptions for mediation analysis]
  \begin{asulist}
    \item\label{asu:ch2-positivity} $\varepsilon < \mu_{D0}(\bX) < 1-\varepsilon$ for some $\varepsilon>0$
    \item\label{asu:1} $Y_{d,\bmm_{\model^*}} \independent D ~\vert~ \bX$~
  for any $d$ and $\bmm_{\model^*}$
    \item \label{asu:3}$ \bM_{d}(\model^*) \independent D ~\vert~ \bX $~ for any $d$
    \item\label{asu:2} $ Y_{d,\bmm_{\model^*}} \independent \bM(\model^*) ~\vert~ \lbrace D, \bX \rbrace $~ for any $d$ and $\bmm_{\model^*}$
    \item \label{asu:4} $ Y_{d,\bmm_{\model^*}} \independent \bM_{d^*}(\model^*) ~\vert~ \bX $~ for any $d,~d^*$ and $\bmm_{\model^*}$
  \end{asulist}
  \label{asu:1to4-mstar}
  \vspace{-0.5em}
\end{assumption}

\Cref{asu:ch2-positivity} is commonly called the positivity assumption in causal inference.
\Cref{asu:1,asu:2,asu:3} are necessary to ensure that the measured confounders $\bX$ are sufficiently large to control for outcome and mediator confounding. \Cref{asu:ch2-positivity,asu:1,asu:3} are met if the treatment is randomly assigned to the individual, as in a randomized controlled trial. Otherwise, these assumptions state that these variables behave as if randomized based on baseline covariates. \Cref{asu:2,asu:4} are crucial for the study of multiple mediators, even when treatment is randomized.
\Cref{asu:2} requires that no variables outside of $D$ and $\bX$ confound the relationship between the outcome and the variables in $\model^*$.
\Cref{asu:4} states that there are no mediators outside of $\model^*$ which are related to any variables in $\model^*$. 

The model \eqref{eq:med-model} is not generally sufficient to study the $\nde_{\model^*}$ and $\nie_{\model^*}$ under \Cref{asu:1,asu:2,asu:3,asu:4}. 
Consequently, we focus on submodels which satisfy:
\begin{equation}
  \begin{aligned}
    Y - \mu_{Y0}(\bX) =& \{\bZ(\model^*_Z) - \bmu_{Z0}(\bX)(\model^*_Z)\}^\top \btheta_{0\model^*} + \epsilon_{\model^*} \\
    \bM({\model^*}) - \bmu_{M0}(\bX)(\model^*) =& \{D - \mu_{D0}(\bX)\} \balpha_{0}(\model^*) + \bm{\eta}(\model^*),
  \end{aligned}
  \label{eq:med-submodel}
\end{equation}
where $\model_Z^*$, as defined in \Cref{sec:notation-model-statement}, satisfies $\bZ(\model_Z^*) = (D, \bM(\model^*)^\top)^\top$. We use a similar model restriction on the error as in \eqref{eq:med-model}, captured in the following assumption.
\begin{assumption}
  \Cref{eq:med-submodel} holds with $\Eop\{ \epsilon_{\model^*} ~|~ \bX, D, \bM(\model^*)\}=0$ and $\Eop\{ \bm{\eta}(\model^*) ~|~ \bX, D\} = \bzero.$
  \label{asu:med-submodel-correct}
\end{assumption}

Under \Cref{asu:med-submodel-correct,asu:1to4-mstar},
\Cref{prop:ch2-natural-effect-submodel} 
establishes a representation for $\nde_{\model^*}$ and $\nie_{\model^*}$ in terms of previously-defined parameters:

\begin{proposition}
  Under \Cref{asu:med-submodel-correct,asu:1to4-mstar},
  \begin{align}
    \nde_{\model^*} &= \gamma_{0 \model^*} \label{eq:nde} \\
    \nie_{\model^*} &= \balpha_{0}(\model^*) \cdot \bbeta_{0 \model^*} .\label{eq:nie}
  \end{align}
  \label{prop:ch2-natural-effect-submodel}
  \vspace{-2em}
\end{proposition}
\begin{proof}
  This proof is presented in \Cref{sec:app-nde-nie}.
\end{proof}

\begin{remark}
  \label{rem:ch2-interventional-effects}
  For the purposes of targeting the sparse causal estimands on the basis of observed data, the \Crefrange{asu:1}{asu:4} are appropriate in multiple-mediator settings \citep{vanderweeleMediationAnalysisMultiple2014}. As these authors
  reported, such assumptions may be violated in settings when true mediating variables are not included in the mediating set. This issue may potentially arise during selection, as the assumptions effectively rule out certain relationships between post-treatment variables in the sparse set $\model^*$ and those in $\model^{*c}$ (see \Cref{sec:app-med-assumps} for one such example). Partially to remedy this problem, \cite{vansteelandtInterventionalEffectsMediation2017} defined interventional effects of mediation. In the context of our model, our expressions \eqref{eq:nde} and \eqref{eq:nie} agree with the expressions of the interventional effects defined therein. Consequently, an interventional effect interpretation could be appropriate for our estimands, thereby obviating \Cref{asu:4}. However, \Cref{asu:2} still requires none of the variables indexed by $\model^{*c}$ to confound the mediator-outcome relationship. We will continue to focus this \work on the previously-defined natural effects with the understanding that an interventional interpretation may also be made under slightly more general assumptions. Additional remarks on our assumptions are presented in \Cref{sec:app-med-assumps}.
\end{remark}

\begin{remark}
  \label{rem:ch2-selection-on-alpha}
  It may be apparent from \Cref{eq:nie,eq:nde,eq:med-submodel} that selection of $\model^*$ has the potential to fundamentally re-define the parameters $\btheta_{0 \model^*}$ corresponding to the submodel; see \Cref{rem:natural-effect-submodel-equal} for further discussion.
  On the other hand, selection of $\model^*$ in the second line of \eqref{eq:med-submodel} only serves to create a subvector of the parameter $\balpha_{0}$. Hence, it suffices to estimate $\balpha_{0}$ once, under $\model^F$.
\end{remark}

\subsection{Our Contribution}
\label{sec:ch2-notation-full-model}

Although we are focused on the natural mediation effects defined by interventions on $\model^*$, this sparse mediation model is unknown. In general we must use the full-data $\model^F$ to identify $\model^*$. 
To this end, we may use similar logic as in the development of \Cref{prop:ch2-natural-effect-submodel} to derive a representation for the natural effects relative to $\model^F$: $\nde_{\model^F}$ and $\nie_{\model^F}$. These quantities are defined similarly to \eqref{eq:nde-causal} and \eqref{eq:nie-causal}, respectively, but with instances of $\model^*$ being replaced by $\model^F$ to reflect interventions on both $D$ and the full vector $\bM$.
As seen in the previous subsection, additional conditions are required to ensure that the $\nde_{\model^F}$ and $\nie_{\model^F}$ are representable based on the observed data; these are summarized in the following assumption:
\begin{assumption}
  \Cref{asu:1to4-mstar,asu:med-submodel-correct}
  hold with respect to $\model^F$. Specifically, the errors satisfy $\Eop( \epsilon ~|~ \bX, D, \bM)=0$ and $\Eop( \bm{\eta} ~|~ \bX, D) = \bzero.$
  \label{asu:1to4-full}
\end{assumption}
\noindent Relevant to our development is the expression for $\nie_{\model^F}$:
\begin{equation}
  \label{eq:nie-full-model}
  \nie_{\model^F} = \balpha_{0} \cdot \bbeta_{0},
\end{equation}
From this representation, the previous definition of $\model^*$ given in \Cref{sec:notation-model-statement} is transparently motivated by $\model^F$. This set of post-treatment variables contains all of the averaged indirect effects flowing through $\model^F$. Heuristically, the set $\model^*$ is the important subset of $\model^F$, in the sense that $\nie_{\model^F}$ is fully determined by these variables. We therefore may view the variables in $\model^{*c}$ as unimportant to the fundamental mediation problem.

\begin{remark}
  \label{rem:natural-effect-submodel-equal}
  By the definition of $\model^*$, $\balpha_0 \cdot \bbeta_0 = \balpha_0(\model^*) \cdot \bbeta_0(\model^*)$. , Recalling the notation of $\bv\subj$ for the $j^{th}$ element of the vector $\bv$, $\nde_{\model^F}=\btheta_0\subj[1]$ and $\nde_{\model^*}=\btheta_{0\model^*}\subj[1]$. Therefore, whenever $\btheta_0(\model^*)=\btheta_{0\model^*}$ the causal effects from each model are numerically equivalent in the sense that $\nde_{\model^*}=\nde_{\model^F}$ and $\nie_{\model^*}=\nie_{\model^F}$. Standard calculations for linear models may establish situations when the subvector of the full parameter equals the submodel parameter. 
  As shown in \Cref{sec:natural-effect-submodel-equal}, the condition $\btheta_0(\model^*)=\btheta_{0\model^*}$ holds when either $\Eop\left[ \bm{\eta}\subj[\ell] \bm{\eta}\subj[j] \right] =0$ for all $j\in\model^*,~\ell\in\model^{*c}$ or $\bbeta_0(\model^{*c})=\bzero$. 
  The simulations in \Cref{sec:simulation} are designed to satisfy this property, so that the natural effects are numerically equivalent between both $\model^*$ and $\model^F$.
\end{remark}

Using the development so far, we are ready to formally state our goals. First, we wish to select the model $\model^*$ using the full data in $\bO_1,\ldots,\bO_n$. Second, we wish to estimate the natural effects corresponding to $\model^*$ by plugging parameter estimates into \eqref{eq:nde}-\eqref{eq:nie}. Finally, we wish to provide inference for these causal estimands after selection. 

The choice to focus on inference for $\model^*$ is motivated from a policy perspective. 
Inclusion of many spurious variables as mediators may hamper researchers' understanding of the causal structure of certain treatment-outcome relationships. For example, when studying the effect of small class sizes in grades K-3 upon 8th grade maths test scores, it may be of interest to discover which particular scores transmit the effect. Identifying a small set of scores could be interesting, as it could assist in determining which specific parts of the curriculum benefit from small class sizes, and how these benefits evolve over time. This can have policy implications, as more targeted efforts to improve the identified mediating variables may allow more efficient use of resources while maximizing student outcomes \citep{ertefaieDiscoveringTreatmentEffect2018}. We therefore do not view $\model^F$ as the mediation model of direct interest, as natural effects under this model would correspond to interventions on a potentially large set of variables. Instead, $\model^F$ is a universe of possible mediation models which might be selected, with $\model^*$ the most interesting and relevant choice in the sense of minimizing the number of mediators while containing all of the causal effect in $\nie_{\model^*}.$

\section{Methodology}
\subsection{Nonparametric Confounding Control}\label{sec:conf-control}

Parametric models have many desirable and well-studied properties when the class of models considered is sufficiently large to characterize the true model. However in many applications, there may be little prior evidence to suggest the parametric form of confounding effects on the outcome, $\psi_Y(\cdot)$, and those effects on $\bM$, $\bm{\psi}_{M}(\cdot)$. These effects represent a nuisance to the researcher, since the primary motivation in mediation analysis lies in characterizing the effect of treatment through direct and indirect pathways, yet control for them must be made in order to obtain unbiased estimates. 
In the previous section, we detailed a re-expression of Model \eqref{eq:med-model} in terms of conditional expectations. 
Whereas the original functions $\psi_Y(\bX)$ and $\bm{\psi}_M(\bX)$ were only estimable in the context of the other covariates in each model, the conditional expectations are directly estimable on the basis of the subscripted variables and $\bX$. 
In other words, the conditional expectations are the same between $\model^*$ and $\model^F$ and need only be estimated once.
The first step in our framework is then to provide acceptable estimates for $\bmu_0$.

To avoid over-fitting, we use $K$-fold cross-fitting to estimate the nuisance parameters \citep{klaassenConsistentEstimationInfluence1987}.
We first randomly split the original sample into 
disjoint (hence independent) samples $(\crossfitData)_{k=1}^K$ such that the size of each sample is roughly $\left \lfloor n/K \right \rfloor$ and $\Ik \subset \{1,\ldots,n\}, k = 1,\ldots,K$ partition the indices $\{1,\ldots,n\}.$  
Analogously, define $\Ikc$ as the set of sample indices that are not included in $\Ik;$ that is, $\Ikc = \{1,2,\cdots,n\} \setminus \Ik, k=1,\ldots,K$. Then, for each $k=1,2,\cdots,K$, estimate the nuisance functions  $\mu_{Y0}(\cdot),$ $ \mu_{D0}(\cdot),$ and $ \bmu_{M0}(\cdot)$ using the data in $\crossfitDataC;$
we respectively denote these estimators 
$\hat \mu_{Y}(\cdot;\crossfitDataC),$ $\hat \mu_{D}(\cdot;\crossfitDataC),$ and $\hat \bmu_{M}(\cdot;\crossfitDataC),$ $k=1,\ldots K.$. As a convenience, we denote these parameters evaluated at $\bX_i$ as $\hat\mu_{Yi}$, $\hat\mu_{Di}$, and $\hat\bmu_{Mi}$, respectively, with $\hat\bmu_{Zi} = (\hat\mu_{Di},~ \hat\bmu_{Mi}^\top)^\top$. 
Similar notation will be used for the true function values, e.g. $\bmu_{Z0i}=(\mu_{D0i}, \bmu_{M0i}^\top)^\top$ with $\mu_{D0i}=\mu_{D0}(\bX_i)$, et cetera.

Instead of relying on any individual machine-learning or parametric method to control confounding, we combine multiple techniques into an ensemble-based predictor. The ``stacked ensemble'' approach was introduced by \cite{breimanStackedRegressions1996} and later generalized under the name SuperLearner \citep{vanderlaanSuperLearner2007}. This method defines a library of candidate estimators and uses cross-validation to estimate an optimal weighted combination of these predictors, which then defines the SuperLearner. Due to previous asymptotic characterizations of cross-validation procedures, this procedure adapts to the optimal model at a fast rate. In addition, the size of the grid can grow at a polynomial rate compared with the sample size without detrimental effects to its oracle performance 
\citep{dudoitAsymptoticsCrossvalidatedRisk2005,vaartOracleInequalitiesMultifold2006}.  For these reasons, it is recommended that the library consist of a 
large and diverse set of regression modeling procedures (e.g., data-adaptive, semiparametric, and parametric). We also recommend including  different tuning parameter specifications of a data-adaptive method.

\subsection{Mediator Selection and Estimation}
\label{sec:proposed-loss}

Identification of $\model^*$ requires finding those elements of $\bM$ which are associated with both the treatment and the outcome. 
Standard regularized estimators that only consider one of two associations may perform poorly in finite samples. For example, penalizing the $\bbeta$ parameters in the outcome model may lead to ignoring mediators that are strongly associated with treatment but weakly associated with the outcome. Removing such variables may bias the estimated $\nie_{\model^*}$ due to the violation of one of the core assumptions in mediation analysis \citep{vanderweeleMediationAnalysisMultiple2014}.

To overcome these limitations, we propose an extension of the adaptive lasso \citep{zouAdaptiveLassoIts2006a} uniquely tailored to the mediation setting. 
Following the representation \eqref{eq:nie-full-model} for $\nie_{\model^F}$ and using the expression for the full partial linear model \eqref{eq:one-stage-true}, we define a class of objective functions
\begin{equation}
    \hat \Lcal_n(\btheta; \hat\bw, \lambda_n) := \frac{1}{n} \sum_{i=1}^n{\Big\{Y_i - \hat{\mu}_{Yi} - \big( \bZ_i - {\hat \bmu}_{Zi} \big)^\top \btheta \Big\}^2} + \frac{\lambda_n}{n} \sum_{j=1}^{p+1}{\hat\bw\subj[j] \big\vert \btheta\subj[j] \big\vert},
    \label{eq:est-med-lasso-obj}
\end{equation}
where $\hat\bw$ represents a vector of (possibly estimated) penalty weights. For example, the standard adaptive lasso weights would use $\hat\bw\subj[j] = |\breve\btheta\subj[j]|^{-\kappa}$ for a preliminary estimator $\breve\btheta$. 
The penalized estimator minimizes the display above, $\hat\btheta^L = \argmin_{\btheta} \hat\Lcal_n(\btheta; \hat\bw, \lambda_n)$,
which also contains the unpenalized estimator $\hat\btheta := \argmin_{\btheta} \hat\Lcal_n(\btheta; \hat\bw, 0)$ as a special case.
In light of \eqref{eq:med-model}, we partition $\hat\btheta^L=(\hat\gamma^L, \hat\bbeta^L)^\top$. Finally, we estimate $\balpha_0$ using least-squares:
\begin{equation}
  \hat { \balpha } = \bigg\{ \sum_{i=1}^n (D_i - \hat\mu_{Di})^2 \bigg\}^{-1} \bigg\{ \sum_{i=1}^n (D_i - \hat\mu_{Di}) (\bM_i - \hat\bmu_{Mi}) \bigg\}.
  \label{eq:alpha}
\end{equation}
Equation \eqref{eq:alpha} is the result of fitting $p$ single-parameter least-squares models. As explained in \Cref{rem:ch2-selection-on-alpha}, selection merely subsets the full $\hat{\balpha}$ vector as used in \eqref{eq:nie}. 
In other words, the individual estimates $\hat\balpha\subj[j]$ for $j=1,\ldots,p$ do not change through selection.
After obtaining a selected model $\hat\model$ as well as estimators $\hat\balpha$ and $\hat\btheta^L$, one may plug in these quantities to create estimators $\ndehat=\hat\gamma^L$ and $\niehat=\hat\balpha \cdot \hat\bbeta^L$. Depending on the penalty used, these estimates may or may not target $\nde_{\model^*}$ and $\nie_{\model^*}$, respectively (c.f. \Cref{sec:theory}).

To identify $\model^*$, we should not penalize $\hat \Lcal_n$ for including a nonzero treatment coefficient; we accomplish this by setting $\hat\bw\subj[1] = 0$. The remaining weights $\hat\bw\subj[1+j]$ for $j=1,\ldots,p$ use pilot estimators $\breve{\balpha}$ and $\breve{\btheta}$, which we assume are $\sqrt{n}-$consistent for $\balpha_0$ and $\btheta_0$, respectively. As verified in Theorem \ref{thm:stage1-can}, a straightforward choice under the assumption $p<n$ would be to define $\breve{\balpha}$ as given in \eqref{eq:alpha} and use the unpenalized estimator $\breve\btheta \equiv \hat\btheta$. 
Adaptations could be made to accommodate the $p>n$ setting; see \Cref{sec:discussion} for discussion.

Our extension proceeds by choosing the adaptive penalty based on the contribution to the $\nie_{\model^F}$. Assuming we have pilot estimators $\breve{\balpha}$ and $\breve{\bbeta}$, let $\kappa > 0$. We propose the novel product weight (PRD):
\begin{equation}
  \hat\bw\subj[1+j] = \vert \breve{\balpha}\subj[j] \breve{\bbeta}\subj[j] \vert^{-\kappa}.
  \label{eq:prd-weight}
\end{equation}
This differs from the standard adaptive lasso (ADP) weight $ \vert \breve{\bbeta}\subj[j] \vert^{-\kappa}$ by including information on the relationship between mediator and treatment.
The penalty function resulting from \eqref{eq:prd-weight} directly penalizes candidate mediators that have little estimated contribution to the $\nie_{\model^F}$. This offers an improvement upon the ADP weight by incorporating information from the complete mediation model, protecting against removal of mediators with small $|\bbeta_0 \subj[j]|$ but large $|\balpha_0\subj[j]|$. Consequently, the PRD weight may include mediators that the ADP does not.

\subsection{Performing Inference Post-Selection}\label{sec:meth-pert-boot}

Our estimators rely on the performance of the adaptive lasso estimator for its simultaneous selection and nearly-unbiased estimation capabilities. Inferential techniques based on the adaptive lasso have been explored by others in the context of outcome-only models. \cite{zouAdaptiveLassoIts2006a} suggested appealing to the oracle normal distribution with a variance estimate informed by the Local Quadratic Approximation (LQA) technique of \cite{fanVariableSelectionNonconcave2001}. More recent work has explored the use of certain bootstrap techniques with the adaptive lasso as an improvement upon the oracle normal approximation. 
\cite{minnierPerturbationMethodInference2011} established first-order correctness of a perturbation bootstrap distribution, and claimed an improvement over the LQA approach supported by simulation studies. Later, \cite{dasPerturbationBootstrapAdaptive2019d} provided a second-order correction to the perturbation bootstrap distribution and claimed further improvement over the \cite{minnierPerturbationMethodInference2011} approach.

Despite investigations into the performance of these approaches in a single-outcome model, it is unclear how various methods perform in the system of equations defined by the mediation model. Traditionally, inference for the $\nie_{\model^*}$ using a plug-in estimator based on \eqref{eq:nie} is possible by appealing to the Delta Method. Confidence intervals and tests could be constructed by plugging in estimates of the asymptotic covariance matrix by appealing to an oracle property similar to LQA. Due to the simultaneous selection and estimation properties of our estimator, we expect such a method to perform undesirably in certain situations \citep{leebSparseEstimatorsOracle2008,leebCanOneEstimate2006}.

Instead, we propose a  generalization of the perturbation bootstrap approach of \cite{minnierPerturbationMethodInference2011}. Let $G_i$, $i=1,\ldots,n$ represent i.i.d. data-independent draws from a user-specified distribution satisfying $\Eop G_1 = 1$ and $\Eop G_1^2 < \infty$. These random variables represent observation-specific weights used to multiply the observation-specific contribution to the relevant loss functions. 
Specifically, the perturbation bootstrap estimates are defined by
\begin{equation}
  \begin{aligned}
    \hat\balpha^{b} =& \bigg\{ \sum_{i=1}^n G_i (D_i - \hat\mu_{Di})^2 \bigg\}^{-1} \bigg\{\sum_{i=1}^n G_i (D_i - \hat\mu_{Di}) (\bM_i - \hat\bmu_{Mi}) \bigg\}  \\
    \hat\btheta^{Lb} :=& \argmin_{\btheta} \sum_{i=1}^n G_i \big\{Y_i - \hat\mu_{Yi} - \btheta^\top (\bZ_i - \hat\bmu_{Zi}) \big\}^2 + \lambda_n^b \sum_{j=1}^{p+1} \hat \bw^b\subj[j] \big\vert \bbeta\subj[j] \big\vert
  \end{aligned}
  \label{eq:pert-boot-est}
\end{equation}
where the weights $\hat\bw^b$ make use of perturbed pilot estimators $(\breve{\balpha}^b, \breve{\bbeta}^b)$ and we once again partition $\hat\btheta^{Lb}=(\hat\gamma^{Lb}, \hat\bbeta^{Lb\top})^\top$. As in \cite{minnierPerturbationMethodInference2011}, we propose to use the distributions of $\sqrt{n}(\hat\gamma^{Lb} - \hat\gamma)$ and $\sqrt{n}(\hat\balpha^{b} \cdot \hat\bbeta^{Lb} - \hat\balpha \cdot \hat\bbeta^L)$ to approximate those of 
$\sqrt{n}(\hat\gamma^L - \gamma_0)$ and $\sqrt{n}(\hat\balpha \cdot \hat\bbeta^L - \balpha_0 \cdot \bbeta_0)$, respectively. This is the basis for forming confidence intervals for the $\nde_{\model^*}$ and $\nie_{\model^*}$.

\section{Theoretical results} \label{sec:theory}

\subsection{Oracle Selection and Normality}
\label{sec:theory-oracle}

Let $\mu_{M0j}(\cdot)$ and $\hat\mu_{Mj}(\cdot;\crossfitDataC)$ represent the $j^{th}$ elements of $\bmu_{M0}(\cdot)$ and $\hat\bmu_{M}(\cdot;\crossfitDataC)$, respectively. Assume the cross-fitting approach of \Cref{sec:conf-control} is used with $K$ separate folds. For $k=1,\ldots,K$ let $n_k = | \bI_k |$ represent the number of observations in each fold, with $n_k / n \rightarrow K^{-1}$ and $| \bI_k^c | = n - n_k$.
Recall our definition of $\| \bm{x} \|_q$ as the $\ell_q$ norm of a vector $\bm{x}$  for $q=1,2,\infty$, and that
$\crossfitDataC$ represents the $\bO_i$ outside of fold $k=1,\ldots,K$; then for a real-valued funtion $f(\bO; \crossfitDataC)$ of $\bO$ and $\crossfitDataC$, where $\bO \independent \crossfitDataC$, let $\Ltwo[P]{f}$ be as defined as the $\crossfitDataC$-dependent $L_2(P)$ norm for some measure $P$:
\begin{equation}
  \Ltwo[P]{f} = \Eop_{P} \left\{ 
    f(\bO; \crossfitDataC)^2
    ~|~ \crossfitDataC
  \right\}^{1/2},
  \label{eq:ltwo-crossfit}
\end{equation}
where $\Eop_{P}$ indicates integration over the distribution of $\bO \sim P$ drawn independently of $\crossfitDataC$. Consequently, the norm itself depends on the random variables $\crossfitDataC$ as well as the partition $(\Ik)_{k=1}^K$.
Let $\lambda_{\min}(\bV)$ and $\lambda_{\max}(\bV)$ represent the minimum and maximum eigenvalues, respectively, of the real-valued square matrix $\bV$. 

Define the matrix $\bH_{0} = \Eop\{\bZ - \bmu_{Z0}(\bX)\}^{\otimes 2}$. 
We make use of the submodel projection coefficients $\btheta_{0\model} = [H_0(\model_Z)]^{-1} \Eop \left[ \{\bZ(\model_Z) - \bmu_{Z0}(\bX)(\model_Z)\} \{ Y - \mu_{Y0}(\bX) \} \right]$, which aligns with previous notation under the special cases $\model=\model^*$ and $\model=\model^F$ for $\btheta_0 \equiv \btheta_{0\model^F}$. These projections induce error terms $\epsilon_{\model} = Y - \mu_{Y0}(\bX) - \{\bZ(\model_Z) - \bmu_{Z0}(\bX)(\model_Z)\}^\top \btheta_{0\model}$, which are defined for any $\model\in\modelspace$. Such quantities have been explored in the post-selection inference literature \citep[e.g., ][]{berkValidPostselectionInference2013c,kuchibhotlaValidPostselectionInference2020c}, although the $\btheta_{0\model}$ are merely a convenience in our formulation; they permit examination of asymptotic behavior for different submodels.
Finally, let $\bV_{1\model} = \Eop[\epsilon_{\model}^2 \{\bZ(\model_Z) - \bmu_{Z0}(\bX)(\model_Z)\}^{\otimes 2}]$
and $\bV_{2} = \Eop[\{D - \mu_{D0}(\bX)\}^{2} \bm{\eta}^{\otimes 2} ]$.
We make the following assumptions.

\begin{assumption}[Bounded Parameters]
  \label{asu:bounded-parameters}
  The vectors $\balpha_0,\btheta_0$ satisfy
  $\Vert\balpha_0\Vert_{\infty} < \infty $ and $ \Vert\btheta_0\Vert_{\infty} < \infty$.
  \vspace{-0.25em}
\end{assumption}

\begin{assumption}[Bounded Variance Matrices]
  \label{asu:ch2-bounded-variance}
  The matrices $\bH_{0}$, $\bV_{1\model}$, and $\bV_{2}$ each have eigenvalues bounded away from $0$ and $\infty$.
  \vspace{-0.25em}
\end{assumption}

\begin{assumption}[Cross product  rates]
  For $j=1,\ldots,p$ and $k=1,\ldots,K$: 
  \begin{align*}
    &\Ltwo{\mu_{D0} - \hat\mu_D} . 
    \Ltwo{\mu_{M0j} - \hat\mu_{Mj} } = o_p((n - n_k)^{-1/2}) \\
    &\Ltwo{\mu_{Y0} - \hat\mu_Y } 
  . \Ltwo{\mu_{M0j} - \hat\mu_{Mj} } = o_p((n - n_k)^{-1/2}) \\
  &\Ltwo{ \mu_{Y0} - \hat\mu_Y } 
  . \Ltwo{\mu_{D0} - \hat\mu_D }  = o_p((n - n_k)^{-1/2})
  \end{align*} 
  \correctThmDisplayMath
  \label{asu:cross-rate}
  
\end{assumption}
\begin{assumption}[Accuracy of treatment and mediator models] \hphantom{sdf}
  \newline
  For $j=1,\ldots,p$ and $k=1,\ldots,K$:
  \[
    \Ltwo{\mu_{D0} - \hat\mu_D} = \Ltwo{\mu_{M0j} - \hat\mu_{Mj} } = o_p((n - n_k)^{-1/4}) 
  \]
  \label{asu:rate-nonpar}
  \correctThmDisplayMath
\end{assumption}
\begin{assumption}[Convergence of the outcome model] For $k=1,\ldots,K$,
  \[
    \Ltwo{ \mu_{Y0} - \hat\mu_Y } = o_p(1).
  \]
  \label{asu:y-converge}
  \correctThmDisplayMath
\end{assumption}

Under these rate assumptions, we first establish the asymptotic behaviour of our proposed estimators under no penalization (i.e., $\lambda_n=0$). Let 
$\hat{\bm{\theta}}_\model\equiv \hat{\bm{\theta}}_\model(\lambda_n=0)$ and $\hat{\bm{\alpha}}_\model$ be the least-squares estimates resulting from \eqref{eq:est-med-lasso-obj} and \eqref{eq:alpha} for a particular submodel restriction $\model\in\modelspace$.

\begin{theorem}
Suppose \Cref{asu:med-submodel-correct,asu:1to4-full,asu:1to4-mstar,asu:bounded-parameters,asu:ch2-bounded-variance,asu:cross-rate,asu:rate-nonpar,asu:y-converge}
hold. 
Let $\model \in \modelspace$ be fixed.
Then,
\begin{align*}
    \sqrt{n}\Big( \hat\btheta_{\model} - \btheta_{0\model} \Big) 
    \rightarrow^d \normal\left( \bzero, \Jcal_{1\model} \right), 
    \mathspaceand~
    \sqrt{n} \Big\{ \hat\balpha(\model) - \balpha_{0}(\model) \Big\} 
    \rightarrow^d \normal\left( \bzero, \Jcal_{2\model} \right),
\end{align*}
where the matrices $\Jcal_{1\model}$ and $\Jcal_{2\model}$ are defined as
\begin{align*}
  \Jcal_{1\model} = [\bH_{0}(\model_Z)]^{-1} \bV_{1\model} [\bH_{0}(\model_Z)]^{-1} ,
  \mathspaceand~
  \Jcal_{2\model} = \Eop\left[  \{D - \mu_{D0}(\bX)\}^2  \right]^{-2} \bV_{2}(\model).
\end{align*}
\label{thm:stage1-can}
\correctThmDisplayMath
\end{theorem}
\begin{proof}
  This proof
  is presented in \Cref{sec:prf-thm-stage1-can}.   
\end{proof}

\Cref{asu:med-submodel-correct,asu:1to4-full,asu:1to4-mstar} are not necessary for the asymptotic distribution results of this theorem; instead, they are presented since they allow causal interpretations under this framework (\Cref{sec:ch2-notation-estimand,sec:ch2-notation-full-model}).
Theorem \ref{thm:stage1-can} importantly requires cross-fitted estimates of the nuisance parameters. This has previously been studied as a step to de-bias parameter estimates \citep{klaassenConsistentEstimationInfluence1987,zhengCrossValidatedTargetedMinimumLossBased2011,chernozhukovDoubleDebiasedMachine2018f}. Combined with 
\Cref{asu:1to4-full,asu:bounded-parameters,asu:ch2-bounded-variance,asu:cross-rate,asu:rate-nonpar,asu:y-converge}, this allows the bias due to nuisance parameter estimation to vanish faster than $n^{-1/2}$. When using correctly-specified parametric models, \Crefrange{asu:cross-rate}{asu:y-converge} hold by a Cauchy-Schwarz argument. When using more adaptive methods like machine learning techniques, the assumptions may still hold--see \cite{biauAnalysisRandomForests2012c, chenLargeSampleSieve2007, xiaohongchenImprovedRatesAsymptotic1999}.

Now we turn to the asymptotic behavior of $\hat\btheta^L$. Let
$\model^\dagger = \{ j : \balpha_0\subj[j] \neq 0 \vee \bbeta_0 \subj[j] \neq 0 \}$ and 
$\model^\ddagger = \{j : \bbeta_0 \subj[j] \neq 0 \}$, for $\vee$ the logical OR operator. Let $\model \in \modelspace$ represent some target set of mediators. We wish to establish oracle asymptotic results like the following for suitably defined $\model$:
\begin{equation}
  \label{eq:oracle-property}
  \begin{aligned}
    \Prob \big\{\hat\btheta^L\subj[j] \neq 0 ~\forall j \in \model^{c} \big\}\rightarrow 0, \mathspaceand
    \sqrt{n}\big\{\hat\btheta^L(\model) - \btheta_{0\model}\big\} \CiD \normal(\bzero, \Jcal_{1 \model}).
  \end{aligned}
\end{equation}

\begin{theorem}
  \label{thm:oracle-selection-normality}
  Let $\kappa > 0$, and
  adopt the setup of Theorem \ref{thm:stage1-can}. Suppose $\breve{\balpha}$, $\breve{\bbeta}$ are $\sqrt{n}-$consistent estimators of $\balpha_0$, $\bbeta_0$, respectively.
  \begin{theoremlist}
    \item\label{thm:oracle-selection-normality-1} Suppose the PRD weights \eqref{eq:prd-weight} are chosen. If $n^{(\kappa-1)/2} \lambda_n \rightarrow \infty$ and $n^{-1/2}\lambda_n \rightarrow 0$, then \eqref{eq:oracle-property} holds with $\model=\model^*$ .
    \item\label{thm:oracle-selection-normality-2} Suppose the PRD weights \eqref{eq:prd-weight} are chosen. If $n^{(2\kappa-1)/2} \lambda_n \rightarrow \infty$ and $n^{(\kappa-1)/2}\lambda_n \rightarrow 0$, then \eqref{eq:oracle-property} holds with $\model=\model^\dagger$.
    \item\label{thm:oracle-selection-normality-3} Suppose the ADP weights are chosen. If $n^{(\kappa-1)/2} \lambda_n \rightarrow \infty$ and $n^{-1/2}\lambda_n \rightarrow 0$, then \eqref{eq:oracle-property} holds with $\model=\model^\ddagger$.
  \end{theoremlist}
\end{theorem}

Theorem \ref{thm:oracle-selection-normality} is proved in \Cref{sec:app-ch2-pf-oracle-selection}; it
establishes the oracle asymptotic normal distribution for our proposed estimator. Notably, the targeted selection set depends on the weight function chosen, as well as the rate of $\lambda_n$. Using the product weights \eqref{eq:prd-weight}, two selection sets are possible: the true set of mediators $\model^*$ and the more conservative set $\model^\dagger$. By contrast, the ADP weights only target the outcome-only set $\model^\ddagger$. Consequently, the proposed weights are selection consistent for the true set of mediators, while the standard weights are only consistent when $\model^* = \model^{\ddagger}.$ The following corollary is a direct consequence of \Cref{thm:stage1-can,thm:oracle-selection-normality} along with standard Delta Method arguments.

\begin{corollary}
  Suppose $p < n$ with $p$ fixed and the setup of \Cref{thm:oracle-selection-normality-1} holds. Choose as pilot estimators $\breve{\balpha}=\hat\balpha$ and $\breve{\bbeta}=\hat\bbeta$.
  Then,
  \begin{align*}
    &\sqrt{n}\left( \hat\gamma^L - \nde_{\model^*} \right) \CiD \normal\left( 0, \Jcal_{1\model^*}\subj[1] \right) \mathspaceand \sqrt{n}\left( \hat\balpha^L \cdot \hat\bbeta^L - \nie_{\model^*} \right) \CiD \normal\left( 0, \Jcal_{NIE, \model^*} \right), \\
    &\text{where }~ \Jcal_{NIE, \model^*} \defined \left( 0,\balpha_{0}(\model^*)^\top \right) \Jcal_{1\model^*} \left( 0,\balpha_{0}(\model^*)^\top \right)^\top + \bbeta_{0\model^*}^\top \Jcal_{2\model^*} \bbeta_{0\model^*}.
  \end{align*}
\end{corollary}

\subsection{Bootstrap Theory}
\label{sec:theory-boot}

We provide justification for our inferential approach based on the first-order correctness of the proposed perturbation bootstrap scheme. We show that the bootstrap distribution is consistent for the unknown sampling distribution of $\ndehat$ and $\niehat$. Thus, confidence intervals and hypothesis testing may be performed on the basis of the bootstrap methodology.

\begin{theorem}
  \label{thm:pert-boot}
  Adopt the assumptions of Theorem \ref{thm:stage1-can}. 
  Let $\hat\balpha^{b}$, $\hat\gamma^{Lb}$, and $\hat\bbeta^{Lb}$ be defined as in \eqref{eq:pert-boot-est}, and define
  $\ndehat^b = \hat\gamma^{Lb}$ and $\niehat^b = \hat\balpha^{b} \cdot \hat\bbeta^{Lb}$. Suppose $n^{(\kappa-1)/2} \lambda_n^b \rightarrow \infty$ and $n^{-1/2}\lambda_n^b \rightarrow 0$, the weights \eqref{eq:prd-weight} are used, and
  $\breve{\balpha}$, $\breve{\bbeta}$ are the unpenalized pilot estimators. Then,
  \begin{theoremlist}
    \item $\Prob \big\{\hat\btheta^{Lb}\subj[j] \neq 0 \text{ for any } j \in \model^{*c}\big\} \rightarrow 0$ and $\sqrt{n}\big\{\hat\btheta^{Lb}(\model^*) - \hat\btheta_{\model^*}\big\} \CiD \normal(\bzero, \Jcal_{1 \model^*})$,
    \item $\sqrt{n}\big\{\hat\balpha^{b}(\model^*) - \hat\balpha({\model^*})\big\} \CiD \normal(\bzero, \Jcal_{2 \model^*})$,
    \item $ \sqrt{n}\left( \niehat^b - \hat\balpha \cdot \hat\bbeta^L \right) \CiD \normal\left( 0, \Jcal_{NIE, \model^*} \right) $
      and 
    $\sqrt{n}\left( \ndehat^b - \hat\gamma^L \right) \CiD \normal\left( 0, \Jcal_{1\model^*}\subj[1] \right).$
    \end{theoremlist}
  
\end{theorem}

\subsection{Selection and Estimation in a Local Asymptotic Framework}
\label{sec:local-asymp}

Next, we present a novel local asymptotic theorem that examines selection and estimation performance in a slightly different setting. This departs from the usual oracle selection analysis by allowing the true data-generating process to change with sample size. This framework was previously used in \cite{knightAsymptoticsLassotypeEstimators2000} to shed light on the differences between ridge regression and lasso regression when some coefficients are small, but nonzero. 
This local asymptotic framework allows us to examine how the proposed PRD weights perform relative to the ADP weights on the most difficult-to-estimate targets---those which are small but nonzero. This local view should allow the asymptotic theory to better approximate the finite-sample case \citep{knightAsymptoticsLassotypeEstimators2000}.

The main idea of this framework is the local data-generating process.
Suppose the parameters of mediation model \eqref{eq:med-model} vary with $n$, so that the true values of $\balpha$,$\bbeta$, and $\gamma$ are given by $\balpha_{0n} = \balpha_{0} + \bh_{1n}$, $\bbeta_{0n} = \bbeta_{0} + \bh_{2n}$, and $\gamma_{0n} = \gamma_0 + \bh_{3n}\subj[1]$, respectively, where the ``directions'' $\bh_{3n}\subj[1]$ and $\bh_{kn}$ for $\ell=1,2$ converge to zero in the sense that 
$|\bh_{\ell n}\subj[j]| \asymp r_{\ell nj}$ for $\ell=1,2,3$. We allow subsets $\tilde{\model}_{\ell} \subset \model^F$ for $\ell=1,2$ to be defined such that $|\bh_{\ell n}\subj[j]| = 0$ for $j \in \tilde{\model}_{\ell}^c$ and $r_{\ell nj} \rightarrow 0$ for $j \in \tilde{\model}_{\ell}$. The rate associated with the unpenalized coefficient is assumed to satisfy $r_{3n1} \rightarrow 0$.
Then for each $n$, we may define a true set of mediators $\model_n^*$ as those candidates for which $\balpha_{0n}\subj[j]\bbeta_{0n}\subj[j] \neq 0$. The true $\nde_{\model_n^*}$ and $\nie_{\model_n^*}$ values as listed in \eqref{eq:nde}-\eqref{eq:nie} may also be extended to use these local parameters. In other words, the true set of mediators, the effect sizes in each mediated pathway, and the mediation effect targets are defined relative to the local data-generating process. This allows the approximation of small-contribution mediators by setting, e.g. $\bbeta_0 \subj[j]=0$ for some $j = 1,\ldots,p$. Relative to the local data-generating process, the true $\bbeta_{0n}\subj[j] \asymp r_{2nj}$.

To simplify our theorem statement, we will make some further assumptions on the rates at which the directions of the local data-generating process converge to zero. These assumptions could be relaxed, although a straightforward comparison becomes complicated, especially when identifying appropriate rates for $\lambda_n$. We will assume that $r_{\ell nj} = n^{-c_{\ell j}}$ for $\ell=1,2$, $j \in \tilde{\model}_{\ell}$ and $\ell=3,~j=1$. For $\ell=1,2$, let $c_{\ell j}=0$ for $j \in \tilde{\model}_{\ell}^c$. We assume $c_{\ell j} \geq 0$ and $0 < c_{1j} + c_{2j} \leq 1/2$ for all $j \in \tilde{\model}_{1} \cup \tilde{\model}_{2} $, whereas $0 < c_{31} \leq 1/2$.
The oracle property for this local process is then generalized as 
\begin{equation}
  \label{eq:local-oracle-property}
  \begin{aligned}
    \Prob \big\{\hat\btheta^L\subj[j] \neq 0 \text{ for any } j \in \model^{c}\big\}\rightarrow 0 ~\mathspaceand~
    \sqrt{n}\big\{\hat\btheta^L(\model) - \btheta_{0n\model}\big\} \CiD \normal(\bzero, \Jcal_{1 \model}).
  \end{aligned}
\end{equation}

Let $\model_n^* = \{ j : \balpha_{0n}\subj[j] \bbeta_{0n}\subj[j] \neq 0 \}$ represent the set of true mediators relative to the local data-generating process. Let $\model_n^\dagger = \{ j: \balpha_{0n}\subj[j] \neq 0 \text{ or } \bbeta_{0n}\subj[j] \neq 0 \}$ 
and $\model_n^\ddagger = \{ j: \bbeta_0 \subj[j] \neq 0 \} \cup \{ j: \bbeta_0 \subj[j] = 0, c_{2j} < 1/2 \}$. Notice that $\model_n^\dagger \supseteq \model_n^*$, but it is not the case that 
$\model_n^\ddagger \supseteq \model_n^*$. 
That is, any procedure which selects $\model_n^\dagger$ conservatively selects the true set of mediators, while any procedure which selects $\model_n^\ddagger$ selects purely based on the outcome model and is sensitive to the rate at which $\bbeta_{0n}\subj[j]$ converges to zero. For relatively small coefficients satisfying $\bbeta_{0n}\subj[j] \asymp n^{-1/2}$, $j \notin \model_n^\ddagger$. We will verify below that the PRD weights may target $\model_n^\dagger$, while
the ADP weights target $\model_n^\ddagger$, meaning that true mediators may be removed by using the ADP weights.

\begin{theorem}
  \label{thm:local-oracle}
  Adopt the assumptions of Theorem \ref{thm:stage1-can}. Suppose Model \eqref{eq:med-model} holds with parameters described by the local data-generating process above.
  \begin{theoremlist}
    \item Suppose the PRD weights \eqref{eq:prd-weight} are chosen. If $n^{(\kappa - 1)/2} \lambda_n \rightarrow 0$ and $n^{(2\kappa - 1)/2} \lambda_n \rightarrow \infty$, then \eqref{eq:local-oracle-property} holds with $\model=\model_n^\dagger$.
    \item Suppose the ADP weights are chosen. If $n^{(\kappa-1-\delta)/2}\lambda_n \rightarrow 0$  and $n^{(\kappa - 1)/2}\lambda_n \rightarrow \infty$ for every $0<\delta<\kappa$, then \eqref{eq:local-oracle-property} holds with $\model=\model_n^\ddagger$.
  \end{theoremlist}
\end{theorem}

Under the conditions of this theorem, the proposed product weights \eqref{eq:prd-weight} select a conservative set of mediators. Variables included in the set include all variables with at most one of $\balpha_{0n}\subj[j],\bbeta_{0n}\subj[j]$ equal to zero. However, the stated conditions for the ADP weights only allow identification of $\model_n^\ddagger$. Variables with coefficients $\bbeta_{0n}\subj[j]$ that are too small will tend to be removed from the selected set, even if $\balpha_{0n}\subj[j]\bbeta_{0n}\subj[j]$ is on the same order as other variables.

It may appear unintuitive that the standard adaptive lasso cannot distinguish between coefficients roughly of size $n^{-1/2}$, while the PRD weights can. The crucial difference between the PRD and ADP weights in this theorem is the size of $n^{(\kappa - 1)/2}\lambda_n$.
For the ADP weights, the requirement $n^{(\kappa - 1)/2}\lambda_n \rightarrow \infty$ to de-select the noise variables replicates the results of \cite{zouAdaptiveLassoIts2006a}. This comes with the side-effect of also estimating $\bbeta_{0n}\subj[j] \asymp n^{-1/2}$ as exactly zero. 
In other words, this makes it possible to de-select true mediators based on how small $\bbeta_{0n}\subj[j]$ is, regardless of the magnitude of $\balpha_{0n}\subj[j]\bbeta_{0n}\subj[j]$.
The strength of the PRD weights is that the crucial requirement for removing completely irrelevant variables $\balpha_{0n}\subj[j]=\bbeta_{0n}\subj[j]=0$ is $n^{(2\kappa - 1)/2}\lambda_n \rightarrow\infty$, as verified in Theorem \ref{thm:oracle-selection-normality}, 
which we may allow without impacting performance on the coefficients with $\balpha_{0n}\subj[j]\bbeta_{0n}\subj[j] \asymp n^{-1/2}$.
The price paid for this flexibility is conservativeness: this latter rate requirement is reported in Theorem \ref{thm:oracle-selection-normality} as ensuring in the fixed-parameter case that non-mediators with only one coefficient nonzero are selected.

\section{Simulation Studies}\label{sec:simulation}

We demonstrate the performance of these methods through extensive simulation studies. We draw attention to the quality of our confounding control under various true-model scenarios and modeling techniques, and the accuracy of our estimated causal effects in comparison with the traditional adaptive lasso. Since selection may improve interpretability, we also compare this established method with our proposed weights on the basis of selection performance. Our summarized findings are that the product weights match the performance of the standard method in the most favorable scenarios while providing better selection performance in the least favorable settings.

A full description of our simulation setting is given in \Cref{sec:app-sim-results}; we present the salient details here.
Data were generated from the model \eqref{eq:med-model} with various configurations. We classify the scenarios based on four factors: the sample size $n$, the dimension $p$, the linear or nonlinear specification of the nuisance functions, and the size of the $\balpha_0$ and $\bbeta_0$. We simulated 1000 data sets of sizes $n=500,1000,2000,$ and $4000$, and two different dimensions were chosen for computational efficiency: $p=10$ and $p=60$. The different functional confounding forms are given below:
\begin{align*}
  \mu_{D0}(\bX) &= \left\lbrace\begin{array}{lr}
    \expit\left\{ 0.8(X_1 + X_2) \right\} &  Linear  \\
    \expit\left\{ 0.8(X_1 X_2 + X_2) \right\} &  Nonlinear
  \end{array}\right. \\
  \psi_M(\bX) &= \left\lbrace\begin{array}{lr}
    X_1 + X_2 - X_3 &  Linear  \\
    X_1^2 + X_2 - X_3 &  Nonlinear
  \end{array}\right. \\
  \psi_Y(\bX) &= \left\lbrace\begin{array}{lr}
    2(X_1 - 0.5) + X_2 + 2X_3 &  Linear  \\
    2(X_1 - 0.5)^2 + X_2 + 2X_3 &  Nonlinear.
  \end{array}\right.
\end{align*}
The linear or nonlinear confounding scenario is abbreviated in the order of $\mu_{D0}$, $\bm{\psi}_M$, and $\psi_Y$. 
For example, the simulation scenario using the linear form of $\mu_{D0}$ and nonlinear forms of $\bm{\psi}_M$ and $\psi_Y$ is denoted as LNN. Finally, we used a ``Large'' and ``Small'' scenario for the size of the coefficients. In the Large scenario, the coefficients are fixed with sample size, whereas in the Small scenario the local process of \Cref{sec:local-asymp} is used: each true mediator has contribution $4n^{-1/2}$ to $\nie_{\model^*}$ with some of the corresponding $\bbeta_{0n}\subj[j]$ of asymptotic size $n^{-1/2}$. Due to the symmetry of our proposed weights with respect to $\breve{\balpha}\subj[j]$ and $\breve{\bbeta}\subj[j]$, a ``Small Alpha'' scenario with $\balpha_{0n}\subj[j]=O(n^{-1/2})$ was considered supplemental, along with a the ``Large'' scenario with $p=60$.
These were run in a less-extensive LNN setting for sizes $n=1000,2000$ in order to reduce the computational load.
Only three mediators exist in each scenario with $\model^* = \{1,2,3\}$.
In these simulation settings, the mediation effects for each of $\model^*$ and $\model^F$ are numerically equivalent: i.e. $\nie_{\model^F}=\nie_{\model^*}$ and $\nde_{\model^F}=\nde_{\model^*}$. See \Cref{rem:natural-effect-submodel-equal} for more details. This allows direct comparison of our selection-based methods with $\model^F-$based methods.

Two of the compared techniques included simultaneous mediator selection and effect estimation, while three estimated the mediation effects under a fixed, rather than selected, model. The simultaneous selection and estimation methods included the PRD and ADP weighted lasso estimators. For these two estimators, the tuning parameters $\lambda_n$ and $\kappa$ were chosen by 10-fold cross-validation.
The fixed-model estimation methods used the unpenalized estimators of \Cref{thm:stage1-can} with the full model $\model^F$ and true model $\model^*$, respectively. 
All of the previous techniques used Robinson's transformation; the final method used the model $\model^*$ along with parametric linear models for $\psi_Y$ and $\bm{\psi}_M$. Cross-fitting was performed by the \texttt{SuperLearner R} package \citep{polleySuperLearnerSuperLearner2019c} 
using \texttt{SL.lm} or \texttt{SL.glm} for continuous or binary variables, respectively \citep{rlang}, \texttt{SL.earth} for Multivariate Adaptive Regression Splines \citep{milborrowEarthMultivariateAdaptive2020a}, and an implementation of Generalized Additive Models using the \texttt{mgcv} package in R \citep{woodFastStableRestricted2011}. Three separate specifications were included for the \texttt{mgcv} learner; all used thin-plate splines and $m=3$, though they differed in terms of the spline basis dimension $k=2,3,$ and $5$.
This library of learners was chosen to facilitate conducting extensive simulations while incorporating flexible methods, as recommended in \Cref{sec:conf-control}.

The variable selection performance of the PRD and ADP weights was evaluated in each of the Large and Small coefficient settings. These results are presented in Table \ref{tab:var-sel-perf}. Two criteria are displayed: the median number of non-mediators selected by the method (MN) and the proportion of the selected models containing $\model^*$ (PC). This latter quantity is important due to scientific, interpretive, or other errors that may result from ignoring true mediators. It is desirable to minimize MN and maximize PC.
In the Large setting, all methods behaved somewhat similarly with respect to both MN and PC. However, in the Small setting, the ADP weights performed poorly due to small values of $\bbeta_0(\model^*)$ which led to true mediators being missed. In contrast, the PRD weights include $\model^*$ nearly 3 times as often, despite each mediator being relatively weak. The proportion of simulations containing the true model also monotonically increased with sample size, while that of the ADP stayed roughly the same. On median, one non-mediator was selected by the ADP, compared to zero by the PRD.

\begin{table}
  \caption{
    \label{tab:var-sel-perf}  
    Performance of variable selection methods in the LLL, LNN, and NNN settings, $p=10$. PC represents the proportion of selected models containing $\model^*$, whereas MN represents the median number of non-mediators selected.
  }
  \centering
   {
\begin{tabular}[t]{ccccccccc}
  \toprule
  \multicolumn{3}{c}{ } & \multicolumn{2}{c}{LLL} & \multicolumn{2}{c}{LNN} & \multicolumn{2}{c}{NNN} \\
  \cmidrule(l{3pt}r{3pt}){4-5} \cmidrule(l{3pt}r{3pt}){6-7} \cmidrule(l{3pt}r{3pt}){8-9}
  Coefficients & n & Weight Version & PC & MN & PC & MN & PC & MN\\
  \midrule
   &  & adaptive & 1.00 & 1.00 & 1.00 & 1.00 & 1.00 & 1.00\\
  
   & \multirow{-2}{*}{\textbf{500}} & product & 1.00 & 0.00 & 1.00 & 1.00 & 1.00 & 1.00\\
  \cmidrule{2-9}
   &  & adaptive & 1.00 & 1.00 & 1.00 & 1.00 & 1.00 & 1.00\\
  
   & \multirow{-2}{*}{\textbf{1000}} & product & 1.00 & 0.00 & 1.00 & 0.00 & 1.00 & 0.00\\
  \cmidrule{2-9}
   &  & adaptive & 1.00 & 1.00 & 1.00 & 1.00 & 1.00 & 1.00\\
  
   & \multirow{-2}{*}{\textbf{2000}} & product & 1.00 & 1.00 & 1.00 & 1.00 & 1.00 & 1.00\\
  \cmidrule{2-9}
   &  & adaptive & 1.00 & 1.00 & 1.00 & 1.00 & 1.00 & 1.00\\
  
  \multirow{-8}{*}{\textbf{Large}} & \multirow{-2}{*}{\textbf{4000}} & product & 1.00 & 0.00 & 1.00 & 0.00 & 1.00 & 1.00\\
  \cmidrule{1-9}
   &  & adaptive & 0.14 & 1.00 & 0.19 & 1.00 & 0.22 & 1.00\\
  
   & \multirow{-2}{*}{\textbf{500}} & product & 0.50 & 0.00 & 0.56 & 0.00 & 0.60 & 0.00\\
  \cmidrule{2-9}
   &  & adaptive & 0.15 & 1.00 & 0.19 & 1.00 & 0.22 & 1.00\\
  
   & \multirow{-2}{*}{\textbf{1000}} & product & 0.55 & 0.00 & 0.58 & 0.00 & 0.62 & 0.00\\
  \cmidrule{2-9}
   &  & adaptive & 0.12 & 1.00 & 0.16 & 1.00 & 0.17 & 1.00\\
  
   & \multirow{-2}{*}{\textbf{2000}} & product & 0.57 & 0.00 & 0.61 & 0.00 & 0.65 & 0.00\\
  \cmidrule{2-9}
   &  & adaptive & 0.14 & 1.00 & 0.16 & 1.00 & 0.20 & 1.00\\
  
  \multirow{-8}{*}{\textbf{Small}} & \multirow{-2}{*}{\textbf{4000}} & product & 0.64 & 0.00 & 0.66 & 0.00 & 0.72 & 0.00\\
  \bottomrule
  \end{tabular}
  }
\end{table}

Coverage rates for $\nde_{\model^*}$ and $\nie_{\model^*}$ using the proposed bootstrap were calculated for both of the selection-based procedures. To provide a fixed-model control, we compared these coverage rates to the cross-fitted model using only the mediators in $\model^*$ under the standard nonparametric bootstrap inference procedure, which we label the ``oracle'' method. Remarkably, the proposed PRD weights achieved much higher coverage rates than the ADP weights.
In the Large coefficient setting, the PRD-based confidence interval coverage tended to closely match that of the $\model^*$-based estimator. The ADP-based estimator resulted in lower coverage rates in the smaller samples, which improved with increasing $n$. These differences were amplified in the Small coefficient setting: the ADP-based coverage was around 85\%, lower than the $\model^*$-based estimator, and did not substantially increase with $n$. The PRD-based coverage was slighly lower than that of the $\model^*$-based estimator, but tended to be higher than 90\%. The difference between these latter two estimators' coverage rates vanished with increased $n$, resulting in coverage near the nominal rate at the highest sample sizes. An example of this trend is presented in \Cref{fig:coverage-example} for the Large and Small coefficient settings in the LNN scenario.

\begin{figure}
  \centering
  \includegraphics[page=1,width=\textwidth]{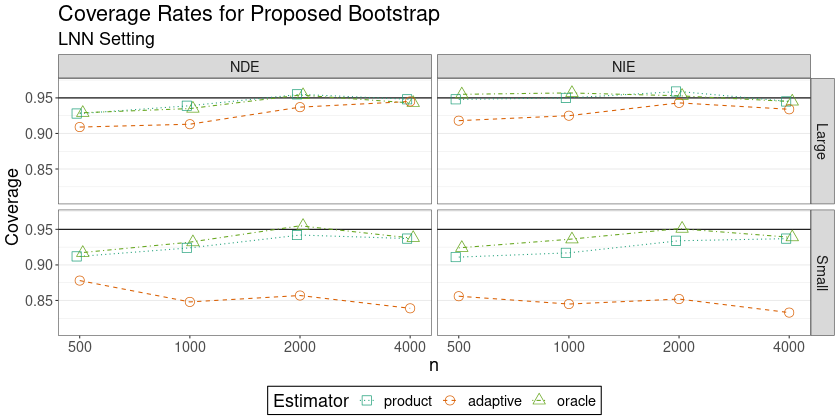}
  \caption[Natural effect bootstrap coverage rates in the Large coefficient case, $p=10$]{Coverage rates of 95\% confidence intervals for $\nde_{\model^*}$ and $\nie_{\model^*}$ in the LNN setting. Plotted are both the Large and Small coefficient settings. The proposed product estimator is compared to the standard adaptive estimator, as well as the fixed-$\model^*$ method (oracle).
  }
  \label{fig:coverage-example}
\end{figure}

A more complete set of figures and tables for these discussed settings are available in \Cref{sec:app-sim-results}, all using $p=10$, demonstrating the favorable properties of the proposal with respect to both bias and coverage rates in each of the discussed scenarios. The bias results are presented in \Cref{tab:sim-bias} and \Cref{fig:bias-large,fig:bias-small}, while the complete coverage results are presented in \Cref{fig:coverage-large,fig:coverage-small}.
The parametric linear model attained favorable bias when the parametric assumption held in scenario LLL. Conversely, the bias was substantial in settings LNN and NNN. 
The remaining methods had similar, small levels of bias. The proposed estimator most closely matched the bias performance of the fixed-$\model^*$ estimator.
Similar results for the supplemental $p=60$ and Small-$\balpha_0$ settings may be found in \Cref{tab:app-supp-results} in \Cref{sec:app-manyp-results} and align with many of the previously discussed trends.

\section{Data Application}

We illustrate our method using data from the `Student/Teacher Achievement Ratio' (STAR) study \cite{wordStudentTeacherAchievement1990}. In this study, 11,600 students were randomized to one of three class size configurations: small class size of 13-17, regular size of 22-25, and regular size with a teachers aide. Individual students maintained the randomized class size from kindergarten to grade 3. The STAR study concluded that small class sizes tended to improve student performance on standardized tests.

One interesting question that may be studied is whether small class sizes affect the developmental trajectory of a student. Small class sizes may improve later test scores via improvements in earlier skills. Identifying the pathways through which the causal effect travels may improve the understanding of how childhood development relates to standardized test scores. From an educational policy perspective, such an analysis may also help to identify particular skills to be targeted for intervention when the goal is the improvement of standardized test scores.

A previous analysis of the STAR study found no significant difference with the addition of a teacher's aide at regular class sizes \cite{kruegerExperimentalEstimatesEducation1999c}. Therefore, we group these two configurations together, yielding a binary ``treatment'' of small class size compared to ``control'' of regular class size as pursued by \cite{ertefaieDiscoveringTreatmentEffect2018}. We study the effect of class size on student performance on grade 8 standardized maths tests via several candidate mediators. We performed this analysis on two separate outcome measures:
Comprehensive Tests of Basic Skills (CTSB) standardized test scores in Math Computation, and Concepts \& Applications. Mediators were chosen to include 60 measures of academic performance and engagement measured from grades K-3. Confounding variables were gender, race, free/reduced lunch status, and school urbanicity status.

As recommended in \Cref{sec:conf-control}, the SuperLearner library contained a variety of methods: a linear model (\texttt{SL.lm}), grand mean (\texttt{SL.mean}), Generalized Additive Model (\texttt{SL.gam}), Lasso regression (\texttt{SL.glmnet}), and Multivariate Adaptive Regression Splines (\texttt{SL.earth}) with default hyperparameters. Two specifications of a \texttt{SL.randomForest} model \citep{liawClassificationRegressionRandomForest2002c} were included with minimum node sizes of 5 and 9, 1000 trees, and 3 variables randomly selected. The \texttt{SL.xgboost} model \citep{chenXGBoostScalableTree2016} was included with a maximum depth of 3 and 100 trees. All other hyperparameters were left at default values. SuperLearner was used to estimate $\mu_{Y0}$ and $\bmu_{M0}$. Since the binary treatment was randomized, the sample mean in each class size category was used as a correctly specified parametric model, thus satisfying the $\Vert\hat\mu_D - \mu_{D0} \Vert_{P_0,2}$ portion of Assumption \ref{asu:rate-nonpar}.

We used the proposed estimators, derived from equations \eqref{eq:est-med-lasso-obj} and \eqref{eq:alpha}. We compared the PRD and ADP estimators to the model using all 60 candidates (FULL), equivalently setting $\lambda_n=0$ in \eqref{eq:est-med-lasso-obj}, as well as a model without selection that also imposed a parametric linear model on the $\psi_Y$ and $\bm{\psi}_{M}$ components of \eqref{eq:med-model}. Confidence intervals based upon the proposed bootstrap inference method with 1000 bootstrap replications were compared to those based upon the oracle asymptotic distribution and the Delta Method, all at the 95\% level.
The tuning parameters $\lambda_n$ and $\kappa$ were chosen by cross-validation using a $\kappa$ grid of $(0.5,1,2,3)$, and a $\lambda_n$ grid of $n^{1/4}2^{\mathcal{G}}$, where the function $x \mapsto 2^{x}$ is applied to each element of $\mathcal{G}$, a 401-element evenly-spaced grid from $-2$ to 10. 

Point estimates, 95\% CIs, and selected model sizes are reported for all methods in Table \ref{tab:data-result}. 
Statistically significant direct effects were not found across any estimation method, inference technique, or outcome. For the Math Computation outcome, the proposed selection method found a significant natural direct effect using the proposed bootstrap methodology, while the ADP-selected model was not found to be significant. Notably, the naive Delta method confidence intervals did not cover zero for both post-selection estimators, while they did cover zero for the semiparametric full model (FULL). The proposed bootstrap applied to the ADP estimator resulted in a much larger confidence interval compared to the naive Delta method. In contrast, the proposed intervals applied to the PRD estimator resulted in a slight enlargement of the intervals. The PRD method selected a model nearly half the size of that of the ADP (20 \textit{versus} 38) while the estimates of $\nie_{\model^*}$ and $\nde_{\model^*}$ were identical up to the second decimal.

Moving to the Math Concepts \& Applications outcome, both the PRD and ADP methods found significant natural indirect effects for their respective selected models according to the proposed bootstrap technique. As before, the PRD estimator selected a smaller model than that of ADP (24 \textit{versus} 27), although the ADP estimator has a smaller $\niehat$ and larger $\ndehat$ compared to the PRD. This seems to indicate that the PRD estimator was better able to select variables in the mediating pathway, as the ADP model relegated these effects to the direct pathway. Comparing the two models, the ADP method included measures of self-concept, motivation, and absences, while the PRD method included more measures of word study skills, reading comprehension, vocabulary, and listening skills. Since the natural direct effects were found not to be significantly different from zero, this suggests that early improvements in the measures identified in \Cref{tab:mediator-cand} may completely mediate the effect of small class sizes on Grade 8 Mathematics standardized test performance.

\renewcommand{\tabcolsep}{4pt}
\begin{table}
  \caption{
    \label{tab:data-result}
    Estimates $\ndehat$ and $\niehat$ along with confidence intervals from the proposed post-selection bootstrap (Proposed CI) as well as the Delta Method (DM CI). Estimates from  the Robinson-transformed model using SuperLearner with the product estimator (PRD), standard adaptive lasso estimator (ADP), no selection (FULL) and the parametric linear model with no selection (LM) are presented.
    }
  { \footnotesize
  \begin{tabular}{ccccccccc}
    \toprule
    Outcome & Method & Size & $\ndehat$ & Proposed CI & DM CI & $\niehat$ & Proposed CI & DM CI\\
    \midrule
      & PRD & 20 & -3.14 & (-8.46, 2.60) & (-8.83, 3.01) & 5.45 & (0.10, 10.47) & (0.12, 10.33)\\
    
      Math & ADP & 38 & -3.14 & (-8.06, 2.60) & (-9.20, 2.77) & 5.45 & (-0.51, 10.72) & (0.28, 10.79)\\
    
      Computation & FULL & 60 & -3.06 & -- & (-9.27, 3.15) & 5.38 & -- & (-0.08, 10.83)\\
    
     & LM & 60 & -3.16 & -- & (-9.32, 3.00) & 5.32 & -- & (-0.06, 10.71)\\
    \cmidrule{1-9}
      & PRD & 24 & 0.26 & (-3.80, 4.18) & (-3.95, 4.66) & 5.24 & (0.91, 9.36) & (0.97, 9.30)\\
    
      Math Concepts & ADP & 27 & 0.86 & (-3.24, 4.89) & (-4.04, 4.67) & 4.63 & (0.33, 8.63) & (0.99, 9.37)\\
    
      \& Applications & FULL & 60 & 0.08 & -- & (-4.40, 4.56) & 5.42 & -- & (1.03, 9.80)\\
    
     & LM & 60 & 0.15 & -- & (-4.29, 4.59) & 5.02 & -- & (0.68, 9.36)\\
    \bottomrule
  \end{tabular}
  }
\end{table}
\renewcommand{\tabcolsep}{6pt}

\begin{table}
  \caption{ 
    \label{tab:mediator-cand} 
  Mediators selected by the proposed methodology for the two Grade 8 Mathematics test scores. Since most mediators were measured in multiple grades, the two rightmost columns list the grades which were selected as important in mediating the effect of small class size on the indicated outcome.}
  \centering
  {\footnotesize
  \begin{tabular}{ccc}
    \toprule
    \multicolumn{1}{c}{ } & \multicolumn{2}{c}{Grades Selected (K-4)} \\
    \cmidrule(l{3pt}r{3pt}){2-3}
    Mediator & Math Computation & Math Concepts \& Applications\\
    \midrule
    Days of absence & 3 & 3\\
    Total math scaled score SAT & 1,2,3 & K,1,2\\
    Total listening scale score SAT & 1,3 & K,1,3\\
    Word study skills scale score SAT & -- & K,1,2,3\\
    Vocabulary scale score SAT & 3,4 & 3\\
    \addlinespace
    Social science scale score SAT & -- & 4\\
    Math concept of numbers scale score SAT & -- & 3\\
    Math applications scale score SAT & 3 & 3\\
    Reading raw score BSF & 1,2 & 1,2\\
    Math raw score BSF & 1,2 & 2\\
    \addlinespace
    Reading percent objectives mastered BSF & 1 & 2,3\\
    Math percent objectives mastered BSF & -- & 2\\
    Total battery scale score CTBS & 4 & 4\\
    Reading comprehension scale score CTBS & 4 & 4\\
    Study skills scale score CTBS & 4 & 4\\
    \bottomrule
  \end{tabular}
  }
\end{table}

\section{Discussion}\label{sec:discussion}

In this paper, we developed causal mediation analysis in a selection setting. We defined the sparse natural effect estimands, which were motivated from both policy implementation and researcher interpretability perspectives. To target these estimands, we proposed a method for simultaneously selecting mediators among a candidate set, estimating the natural effects, and performing inference upon the selected set. Theoretical properties were established for this method in both a classical asymptotic framework, as well as in a local asymptotic framework. These theoretical properties were verified in extensive simulation studies.

In \Cref{rem:ch2-interventional-effects}, we acknowledged a possible interventional-effect interpretation of our resulting estimators, which would be valid even after removing \Cref{asu:4} from \Cref{asu:1to4-full,asu:1to4-mstar}.
However, \Cref{asu:2} still requires none of the variables indexed by $\model^{*c}$ to confound the mediator-outcome relationship. These issues are fundamental to all mediation analysis, as encoded in the causal assumptions. In terms of the observed data, this implies that certain causal or correlation patterns between $\bM(\model^{*})$, $\bM(\model^{*c})$, and $Y$ are disallowed, conditional on $\bX, D$. See \Cref{rem:app-violations-of-asu-2} in \Cref{sec:app-med-assumps} for discussion and an example of this issue.

The results presented in this paper raise additional interesting questions that merit future study.  For example,
our asymptotic framework assumed a fixed-dimensional set of mediators, as well as a binary treatment or exposure. Future work should examine establishing similar results, while allowing $p \rightarrow \infty$. For these large-$p$ settings, we might instead relax the requirement that the pilot estimators be $\sqrt{n}-$consistent in favor of a relaxed $\sqrt{n / \log p}-$consistency result, which may be achieved by L1- or L2-penalized estimators \citep[c.f.~][]{zouAdaptiveElasticnetDiverging2009,negahbanUnifiedFrameworkHighDimensional2012b}. The rates on $\lambda_n$ might be adjusted to account for this reduced rate.
This framework brings several practical and theoretical challenges, as the dimension of $\bmu_{Z0}$ increases with $p$.
It is unclear how this setting may impact bootstrap-based inference.

Theorem \ref{thm:local-oracle} presented a local asymptotic analysis to approximate the finite-sample situation of small coefficients, but did so by making assumptions on the rates at which the true parameters were perturbed.
Improvements upon these results could be made by making additional assumptions on the rates $r_{1nj},r_{2nj}$ and more finely-tuning the choice of $\lambda_n$ based on these assumptions. One implication of the existing requirement that $c_{1j} + c_{2j} \leq 1/2$ is that none of the $\nie_{\model^F}$ contributions are too small to be of smaller order than sampling error. In the absence of refinements, the presented rates for $\lambda_n$ are in some sense the most favorable to both weighting methods, in that they allow the inclusion of as many true mediators as possible. They also do not require that $\lambda_n$ be chosen in a way that makes use of the underlying size of the coefficients. 

Our proposed perturbation bootstrap leverages a first-order multiplier bootstrap distribution using a percentile approach. Studentization often can improve the quality of bootstrap approximation \citep{hallBootstrapEdgeworthExpansion1992b}, although, in the case of the $\nie_{\model^F}$, such a technique requires a variance estimate of the sum-product estimator $\hat\balpha \cdot \hat\bbeta$. Such a variance estimate might be furnished by a Delta Method approximation with a Heteroskedastic-Consistent covariance estimate, although the resulting performance is unclear. A recent proposal by \cite{dasPerturbationBootstrapAdaptive2019d} used an Edgeworth-expansion approach to arrive at a bias-correction to the \cite{minnierPerturbationMethodInference2011} approach. A higher-order study of our bootstrap approach with the Edgeworth-expansion may provide a bias correction, although it is unclear how this study might accommodate the nuisance parameter estimation via Robinson's transformation. Simulation studies did not show an improvement when applying the Das bootstrap to our setting. Developing a higher-order-correct bootstrap deserves further study.

\section{Acknowledgements}

The authors thank the Center for Integrated Research Computing (CIRC) at the University of Rochester for providing computational resources and technical support.

\bibliographystyle{rss}
\bibliography{mediation}

\begin{thebibliography}{50}
\expandafter\ifx\csname natexlab\endcsname\relax\def\natexlab#1{#1}\fi
\expandafter\ifx\csname url\endcsname\relax
  \def\url#1{\texttt{#1}}\fi
\expandafter\ifx\csname urlprefix\endcsname\relax\def\urlprefix{URL: }\fi

\bibitem[{Arcones(1998)}]{arconesAsymptoticTheoryMestimators1998}
Arcones, M.~A. (1998) Asymptotic theory for {{M-estimators}} over a convex
  kernel.
\newblock \textit{Econometric Theory}, \textbf{14}, 387--422.

\bibitem[{Berk et~al.(2013)Berk, Brown, Buja, Zhang, Zhao
  et~al.}]{berkValidPostselectionInference2013c}
Berk, R., Brown, L., Buja, A., Zhang, K., Zhao, L. et~al. (2013) Valid
  post-selection inference.
\newblock \textit{The Annals of Statistics}, \textbf{41}, 802--837.

\bibitem[{Biau(2012)}]{biauAnalysisRandomForests2012c}
Biau, G. (2012) Analysis of a random forests model.
\newblock \textit{The Journal of Machine Learning Research}, \textbf{13},
  1063--1095.

\bibitem[{Breiman(1996)}]{breimanStackedRegressions1996}
Breiman, L. (1996) Stacked regressions.
\newblock \textit{Machine Learning}, \textbf{24}, 49--64.

\bibitem[{Casella and Berger(2002)}]{casellaStatisticalInference2002}
Casella, G. and Berger, R.~L. (2002) \textit{Statistical {{Inference}}}.
\newblock Duxbury {{Advanced Series}}. {Pacific Grove, Calif}: {Duxbury/Thomson
  Learning}.

\bibitem[{Ch{\'e}n et~al.(2017)Ch{\'e}n, Crainiceanu, Ogburn, Caffo, Wager and
  Lindquist}]{chenHighdimensionalMultivariateMediation2017}
Ch{\'e}n, O.~Y., Crainiceanu, C., Ogburn, E.~L., Caffo, B.~S., Wager, T.~D. and
  Lindquist, M.~A. (2017) High-dimensional multivariate mediation with
  application to neuroimaging data.
\newblock \textit{Biostatistics}, \textbf{19}, 121--136.

\bibitem[{Chen and Guestrin(2016)}]{chenXGBoostScalableTree2016}
Chen, T. and Guestrin, C. (2016) {{XGBoost}}: {{A}} scalable tree boosting
  system.
\newblock In \textit{Proceedings of the 22nd {{ACM SIGKDD}} International
  Conference on Knowledge Discovery and Data Mining}, {{KDD}} '16, 785--794.
  {New York, NY, USA}: {ACM}.

\bibitem[{Chen(2007)}]{chenLargeSampleSieve2007}
Chen, X. (2007) Large {{Sample Sieve Estimation}} of {{Semi-Nonparametric
  Models}}.
\newblock In \textit{Handbook of {{Econometrics}}} (eds. J.~J. Heckman and
  E.~E. Leamer), vol.~6, 5549--5632. {Elsevier}.

\bibitem[{Chernozhukov et~al.(2018)Chernozhukov, Chetverikov, Demirer, Duflo,
  Hansen, Newey and Robins}]{chernozhukovDoubleDebiasedMachine2018f}
Chernozhukov, V., Chetverikov, D., Demirer, M., Duflo, E., Hansen, C., Newey,
  W. and Robins, J. (2018) Double/debiased machine learning for treatment and
  structural parameters.
\newblock \textit{The Econometrics Journal}, \textbf{21}, C1--C68.

\bibitem[{Das et~al.(2019)Das, Gregory and
  Lahiri}]{dasPerturbationBootstrapAdaptive2019d}
Das, D., Gregory, K. and Lahiri, S.~N. (2019) Perturbation bootstrap in
  adaptive lasso.
\newblock \textit{The Annals of Statistics}, \textbf{47}, 2080--2116.

\bibitem[{Dudoit and van {der
  Laan}(2005)}]{dudoitAsymptoticsCrossvalidatedRisk2005}
Dudoit, S. and van {der Laan}, M.~J. (2005) Asymptotics of cross-validated risk
  estimation in estimator selection and performance assessment.
\newblock \textit{Statistical Methodology}, \textbf{2}, 131--154.

\bibitem[{Ertefaie et~al.(2018)Ertefaie, Hsu, Page and
  Small}]{ertefaieDiscoveringTreatmentEffect2018}
Ertefaie, A., Hsu, J.~Y., Page, L.~C. and Small, D.~S. (2018) Discovering
  treatment effect heterogeneity through post-treatment variables with
  application to the effect of class size on mathematics scores.
\newblock \textit{Journal of the Royal Statistical Society: Series C (Applied
  Statistics)}, \textbf{67}, 917--938.

\bibitem[{Ertefaie et~al.(2021)Ertefaie, McKay, Oslin and
  Strawderman}]{ertefaieRobustQLearning2021}
Ertefaie, A., McKay, J.~R., Oslin, D. and Strawderman, R.~L. (2021) Robust
  {{Q-Learning}}.
\newblock \textit{Journal of the American Statistical Association},
  \textbf{116}, 368--381.

\bibitem[{Fan and Li(2001)}]{fanVariableSelectionNonconcave2001}
Fan, J. and Li, R. (2001) Variable selection via nonconcave penalized
  likelihood and its oracle properties.
\newblock \textit{Journal of the American Statistical Association},
  \textbf{96}, 1348--1360.

\bibitem[{Hall(1992)}]{hallBootstrapEdgeworthExpansion1992b}
Hall, P. (1992) \textit{The Bootstrap and {{Edgeworth}} Expansion}.
\newblock {Springer Science \& Business Media}.

\bibitem[{Huang and Zhang(2012)}]{huangEstimationSelectionAbsolute2012c}
Huang, J. and Zhang, C.-H. (2012) Estimation and selection via absolute
  penalized convex minimization and its multistage adaptive applications.
\newblock \textit{The Journal of Machine Learning Research}, \textbf{13},
  1839--1864.

\bibitem[{Huang and Pan(2016)}]{huangHypothesisTestMediation2016}
Huang, Y.-T. and Pan, W.-C. (2016) Hypothesis test of mediation effect in
  causal mediation model with high-dimensional continuous mediators.
\newblock \textit{Biometrics}, \textbf{72}, 402--413.

\bibitem[{Klaassen(1987)}]{klaassenConsistentEstimationInfluence1987}
Klaassen, C. A.~J. (1987) Consistent {{Estimation}} of the {{Influence
  Function}} of {{Locally Asymptotically Linear Estimators}}.
\newblock \textit{The Annals of Statistics}, \textbf{15}, 1548--1562.

\bibitem[{Knight and Fu(2000)}]{knightAsymptoticsLassotypeEstimators2000}
Knight, K. and Fu, W. (2000) Asymptotics for lasso-type estimators.
\newblock \textit{The Annals of Statistics}, \textbf{28}, 1356--1378.

\bibitem[{Krueger(1999)}]{kruegerExperimentalEstimatesEducation1999c}
Krueger, A.~B. (1999) Experimental {{Estimates}} of {{Education Production
  Functions}}.
\newblock \textit{The Quarterly Journal of Economics}, \textbf{114}, 497--532.

\bibitem[{Kuchibhotla et~al.(2020)Kuchibhotla, Brown, Buja, Cai, George and
  Zhao}]{kuchibhotlaValidPostselectionInference2020c}
Kuchibhotla, A.~K., Brown, L.~D., Buja, A., Cai, J., George, E.~I. and Zhao,
  L.~H. (2020) Valid post-selection inference in model-free linear regression.
\newblock \textit{The Annals of Statistics}, \textbf{48}, 2953--2981.

\bibitem[{Leeb and P{\"o}tscher(2006)}]{leebCanOneEstimate2006}
Leeb, H. and P{\"o}tscher, B.~M. (2006) Can one estimate the conditional
  distribution of post-model-selection estimators?
\newblock \textit{The Annals of Statistics}, \textbf{34}, 2554--2591.

\bibitem[{Leeb and P{\"o}tscher(2008)}]{leebSparseEstimatorsOracle2008}
--- (2008) Sparse estimators and the oracle property, or the return of
  {{Hodges}}' estimator.
\newblock \textit{Journal of Econometrics}, \textbf{142}, 201--211.

\bibitem[{Liaw et~al.(2002)Liaw, Wiener
  et~al.}]{liawClassificationRegressionRandomForest2002c}
Liaw, A., Wiener, M. et~al. (2002) Classification and regression by
  {{randomForest}}.
\newblock \textit{R news}, \textbf{2}, 18--22.

\bibitem[{Milborrow(2020)}]{milborrowEarthMultivariateAdaptive2020a}
Milborrow, S. (2020) Earth: {{Multivariate Adaptive Regression Splines}}.

\bibitem[{Minnier et~al.(2011)Minnier, Tian and
  Cai}]{minnierPerturbationMethodInference2011}
Minnier, J., Tian, L. and Cai, T. (2011) A perturbation method for inference on
  regularized regression estimates.
\newblock \textit{Journal of the American Statistical Association},
  \textbf{106}, 1371--1382.

\bibitem[{Negahban et~al.(2012)Negahban, Ravikumar, Wainwright and
  Yu}]{negahbanUnifiedFrameworkHighDimensional2012b}
Negahban, S.~N., Ravikumar, P., Wainwright, M.~J. and Yu, B. (2012) A unified
  framework for high-dimensional analysis of {{M-estimators}} with decomposable
  regularizers.
\newblock \textit{Statistical Science}, \textbf{27}, 538--557.

\bibitem[{Pearl(2001)}]{pearlDirectIndirectEffects2001}
Pearl, J. (2001) Direct and indirect effects.
\newblock In \textit{Proceedings of the Seventeenth Conference on Uncertainty
  in Artificial Intelligence}, {{UAI}}'01, 411--420. {San Francisco, CA, USA}:
  {Morgan Kaufmann Publishers Inc.}

\bibitem[{Polley et~al.(2019)Polley, LeDell, Kennedy and van~der
  Laan}]{polleySuperLearnerSuperLearner2019c}
Polley, E., LeDell, E., Kennedy, C. and van~der Laan, M. (2019)
  \textit{{{SuperLearner}}: {{Super Learner Prediction}}}.

\bibitem[{{R Core Team}(2020)}]{rlang}
{R Core Team} (2020) \textit{R: {{A Language}} and {{Environment}} for
  {{Statistical Computing}}}.
\newblock {Vienna, Austria}: {R Foundation for Statistical Computing}.

\bibitem[{Robins and
  Greenland(1992)}]{robinsIdentifiabilityExchangeabilityDirect1992}
Robins, J.~M. and Greenland, S. (1992) Identifiability and exchangeability for
  direct and indirect effects.
\newblock \textit{Epidemiology}, \textbf{3}, 143--155.

\bibitem[{Robinson(1988)}]{robinsonRootnconsistentSemiparametricRegression1988}
Robinson, P.~M. (1988) Root-n-consistent semiparametric regression.
\newblock \textit{Econometrica}, \textbf{56}, 931--954.

\bibitem[{Schaid and Sinnwell(2020)}]{schaidPenalizedModelsAnalysis2020}
Schaid, D.~J. and Sinnwell, J.~P. (2020) Penalized models for analysis of
  multiple mediators.
\newblock \textit{Genetic Epidemiology}, \textbf{44}, 408--424.

\bibitem[{Song et~al.(2020)Song, Zhou, Zhang, Zhao, Liu, Kardia, Roux, Needham,
  Smith and Mukherjee}]{songBayesianShrinkageEstimation2020}
Song, Y., Zhou, X., Zhang, M., Zhao, W., Liu, Y., Kardia, S. L.~R., Roux, A.
  V.~D., Needham, B.~L., Smith, J.~A. and Mukherjee, B. (2020) Bayesian
  shrinkage estimation of high dimensional causal mediation effects in omics
  studies.
\newblock \textit{Biometrics}, \textbf{76}, 700--710.

\bibitem[{van~der Vaart et~al.(2006)van~der Vaart, Dudoit and {van der
  Laan}}]{vaartOracleInequalitiesMultifold2006}
van~der Vaart, A.~W., Dudoit, S. and {van der Laan}, M.~J. (2006) Oracle
  inequalities for multi-fold cross validation.
\newblock \textit{Statistics \& Decisions}, \textbf{24}, 351--371.

\bibitem[{{van der Laan} and Petersen(2008)}]{vanderlaanDirectEffectModels2008}
{van der Laan}, M.~J. and Petersen, M.~L. (2008) Direct effect models.
\newblock \textit{The International Journal of Biostatistics}, \textbf{4},
  Article 23.

\bibitem[{{van der Laan} et~al.(2007){van der Laan}, Polley and
  Hubbard}]{vanderlaanSuperLearner2007}
{van der Laan}, M.~J., Polley, E.~C. and Hubbard, A.~E. (2007) Super
  {{Learner}}.
\newblock \textit{Statistical Applications in Genetics and Molecular Biology},
  \textbf{6}, Article 25.

\bibitem[{VanderWeele and
  Vansteelandt(2014)}]{vanderweeleMediationAnalysisMultiple2014}
VanderWeele, T. and Vansteelandt, S. (2014) Mediation analysis with multiple
  mediators.
\newblock \textit{Epidemiologic Methods}, \textbf{2}, 95--115.

\bibitem[{VanderWeele and
  Vansteelandt(2009)}]{vanderweeleConceptualIssuesConcerning2009}
VanderWeele, T.~J. and Vansteelandt, S. (2009) Conceptual issues concerning
  mediation, interventions and composition.
\newblock \textit{Statistics and its Interface}, \textbf{2}, 457--468.

\bibitem[{VanderWeele and
  Vansteelandt(2010)}]{vanderweeleOddsRatiosMediation2010}
--- (2010) Odds ratios for mediation analysis for a dichotomous outcome.
\newblock \textit{American Journal of Epidemiology}, \textbf{172}, 1339--1348.

\bibitem[{Vansteelandt and
  Daniel(2017)}]{vansteelandtInterventionalEffectsMediation2017}
Vansteelandt, S. and Daniel, R.~M. (2017) Interventional effects for mediation
  analysis with multiple mediators.
\newblock \textit{Epidemiology (Cambridge, Mass.)}, \textbf{28}, 258--265.

\bibitem[{Wood(2011)}]{woodFastStableRestricted2011}
Wood, S.~N. (2011) Fast stable restricted maximum likelihood and marginal
  likelihood estimation of semiparametric generalized linear models.
\newblock \textit{Journal of the Royal Statistical Society. Series B
  (Methodological)}, \textbf{73}, 3--36.

\bibitem[{Word et~al.(1990)}]{wordStudentTeacherAchievement1990}
Word, E. et~al. (1990) Student/{{Teacher Achievement Ratio}} ({{STAR}}),
  {{Tennessee}}'s {{K-3 Class Size Study}}. {{Final Summary Report}} 1985-1990.

\bibitem[{{Xiaohong Chen} and
  White(1999)}]{xiaohongchenImprovedRatesAsymptotic1999}
{Xiaohong Chen} and White, H. (1999) Improved rates and asymptotic normality
  for nonparametric neural network estimators.
\newblock \textit{IEEE Transactions on Information Theory}, \textbf{45},
  682--691.

\bibitem[{Zhang(2010)}]{zhangNearlyUnbiasedVariable2010}
Zhang, C.-H. (2010) Nearly unbiased variable selection under minimax concave
  penalty.
\newblock \textit{The Annals of Statistics}, \textbf{38}, 894--942.

\bibitem[{Zhang et~al.(2016)Zhang, Zheng, Zhang, Gao, Joyce, Yoon, Zhang,
  Schwartz, Just, Colicino et~al.}]{zhangEstimatingTestingHighdimensional2016}
Zhang, H., Zheng, Y., Zhang, Z., Gao, T., Joyce, B., Yoon, G., Zhang, W.,
  Schwartz, J., Just, A., Colicino, E. et~al. (2016) Estimating and testing
  high-dimensional mediation effects in epigenetic studies.
\newblock \textit{Bioinformatics}, \textbf{32}, 3150--3154.

\bibitem[{Zhao and Luo(2016)}]{zhaoPathwayLassoEstimate2016}
Zhao, Y. and Luo, X. (2016) Pathway {{Lasso}}: {{Estimate}} and select sparse
  mediation pathways with high dimensional mediators.
\newblock \textit{arXiv:1603.07749 [stat]}.

\bibitem[{Zheng and {van der
  Laan}(2011)}]{zhengCrossValidatedTargetedMinimumLossBased2011}
Zheng, W. and {van der Laan}, M.~J. (2011) Cross-{{Validated Targeted
  Minimum-Loss-Based Estimation}}.
\newblock In \textit{Targeted {{Learning}}: {{Causal Inference}} for
  {{Observational}} and {{Experimental Data}}}, 459--474. {New York, NY}:
  {Springer New York}.

\bibitem[{Zou(2006)}]{zouAdaptiveLassoIts2006a}
Zou, H. (2006) The adaptive lasso and its oracle properties.
\newblock \textit{Journal of the American Statistical Association},
  \textbf{101}, 1418--1429.

\bibitem[{Zou and Zhang(2009)}]{zouAdaptiveElasticnetDiverging2009}
Zou, H. and Zhang, H.~H. (2009) On the adaptive elastic-net with a diverging
  number of parameters.
\newblock \textit{The Annals of Statistics}, \textbf{37}, 1733.

\end{thebibliography}

\appendix

\beginsuppplement[A]

\section{Simulation Results}
\label{sec:app-sim-results}

For simplicity, we set $\bm{\psi}_{M}(\cdot) \equiv \bm{1} \psi_M(\cdot)$ for $\bm{1} \in \R^p$ (i.e., the vector-valued function $\bm{\psi}_M$ took on the same value in each index), so that the model was characterized by the parameters $\balpha_0, \btheta_0$, and three functions.
The dimension of the confounders $\bX$ was set to $q=3$. Unless otherwise specified, the dimension of the candidate mediators was set to $p=10$; a setting with $p=60$ is reported in \Cref{sec:app-manyp-results}. Each entry of $\bX$, $\bm{\eta}$, and $\epsilon$ was generated independently as a Normal random variable with mean zero. The variances of $X_l,\ l=1,\ldots,q$ were chosen as $1/4$, those of $\eta_j,\ j=1,\ldots,10$ as $1$, and that of $\epsilon$ as $1$. The small variance of the confounders was chosen to minimize the chance of inadvertent positivity violations, while those of the mediators and outcome were chosen to increase the noise in the outcome model.

We examined our methods under two specifications of the structural parameters $\balpha_0, \bbeta_0$. In both cases, $\gamma_0=2$ and only the first 3 elements of $\bM$ represent true mediators. The Large setting represents a classical simulation for variable selection. In this case, $\balpha_0 = (1, 2, 2, \bzero^\top)^\top$, $\bbeta_0=(.8,.4,.4, \bzero^\top)^\top$, so that $\nie=2.4$. The Small setting represents a set of relatively weak mediators, which add up to a moderate indirect effect. We set $\balpha_0=4(n^{-1/4}, 1, 1, \bzero^\top)^\top$, $\bbeta_0=(n^{-1/4}, n^{-1/2}, n^{-1/2}, \bzero^\top)^\top$. In both $\balpha_0$ and $\bbeta_0$, the vector $\bzero$ is of dimension $p-3$.
Each of the three true mediators' contributions to the $\nie_{\model^*}$ (and $\nie_{\model^F}$) are equal, at a value of $4n^{-1/2}$, with $\nie_{\model^*}=12n^{-1/2}$.

Cross-fitting was used with $K=10$ folds. 
The proposed perturbation bootstrap was also evaluated for each of the estimators using the product (PRD) and standard adaptive lasso (ADP) weights. We used the exponential distribution with unity rate parameter as the distribution of the random weights $G$. The 0.025 and .975 quantiles of $\ndehat^b$ and $\niehat^b$ were estimated using the empirical quantiles of 1000 Monte-Carlo draws from the distribution of $G_1, \ldots, G_n$ conditioned upon the data. These coverage rates were compared to that of the oracle estimator, which leverages cross-fitting under the true $\model^*$, using the nonparametric bootstrap.

\begin{table}
  
  \caption{ 
    \label{tab:sim-bias} Absolute bias of each method across simulation settings. 
  }
  \centering
  {\footnotesize 
\begin{tabular}[t]{ccccccccc}
  \toprule
  \multicolumn{3}{c}{ } & \multicolumn{6}{c}{Abs. Monte-Carlo Bias} \\
  \cmidrule(l{3pt}r{3pt}){4-9}
  \multicolumn{3}{c}{ } & \multicolumn{2}{c}{LLL} & \multicolumn{2}{c}{LNN} & \multicolumn{2}{c}{NNN} \\
  \cmidrule(l{3pt}r{3pt}){4-5} \cmidrule(l{3pt}r{3pt}){6-7} \cmidrule(l{3pt}r{3pt}){8-9}
  Coefficients & n & Model & NDE & NIE & NDE & NIE & NDE & NIE\\
  \midrule
   &  & product & 0.008 & 0.007 & 0.022 & 0.011 & 0.038 & 0.019\\
  
   &  & adaptive & 0.001 & 0.016 & 0.013 & 0.003 & 0.030 & 0.011\\
  
   &  & full & 0.011 & 0.004 & 0.025 & 0.014 & 0.042 & 0.023\\
  
   &  & oracle & 0.011 & 0.004 & 0.025 & 0.014 & 0.041 & 0.022\\
  
   & \multirow{-5}{*}{\textbf{500}} & parametric-linear & 0.005 & 0.004 & 0.893 & 0.899 & 0.892 & 0.896\\
  \cmidrule{2-9}
   &  & product & 0.004 & 0.006 & 0.014 & 0.007 & 0.020 & 0.014\\
  
   &  & adaptive & 0.001 & 0.011 & 0.007 & 0.001 & 0.013 & 0.008\\
  
   &  & full & 0.005 & 0.005 & 0.015 & 0.008 & 0.021 & 0.015\\
  
   &  & oracle & 0.005 & 0.004 & 0.015 & 0.009 & 0.021 & 0.015\\
  
   & \multirow{-5}{*}{\textbf{1000}} & parametric-linear & 0.003 & 0.005 & 0.901 & 0.900 & 0.900 & 0.907\\
  \cmidrule{2-9}
   &  & product & 0.002 & 0.004 & 0.004 & 0.005 & 0.011 & 0.005\\
  
   &  & adaptive & 0.005 & 0.006 & 0.002 & 0.002 & 0.008 & 0.002\\
  
   &  & full & 0.001 & 0.003 & 0.005 & 0.005 & 0.012 & 0.005\\
  
   &  & oracle & 0.001 & 0.003 & 0.005 & 0.005 & 0.012 & 0.005\\
  
   & \multirow{-5}{*}{\textbf{2000}} & parametric-linear & 0.005 & 0.003 & 0.902 & 0.904 & 0.903 & 0.899\\
  \cmidrule{2-9}
   &  & product & 0.003 & 0.000 & 0.007 & 0.005 & 0.011 & 0.007\\
  
   &  & adaptive & 0.001 & 0.002 & 0.005 & 0.003 & 0.009 & 0.005\\
  
   &  & full & 0.003 & 0.000 & 0.007 & 0.005 & 0.012 & 0.007\\
  
   &  & oracle & 0.003 & 0.000 & 0.007 & 0.005 & 0.011 & 0.007\\
  
  \multirow{-20}{*}{\textbf{Large}} & \multirow{-5}{*}{\textbf{4000}} & parametric-linear & 0.001 & 0.001 & 0.907 & 0.906 & 0.908 & 0.906\\
  \cmidrule{1-9}
   &  & product & 0.005 & 0.013 & 0.048 & 0.044 & 0.088 & 0.080\\
  
   &  & adaptive & 0.140 & 0.147 & 0.095 & 0.100 & 0.055 & 0.063\\
  
   &  & full & 0.052 & 0.044 & 0.104 & 0.100 & 0.146 & 0.137\\
  
   &  & oracle & 0.051 & 0.043 & 0.105 & 0.100 & 0.145 & 0.137\\
  
   & \multirow{-5}{*}{\textbf{500}} & parametric-linear & 0.006 & 0.004 & 1.584 & 1.589 & 1.583 & 1.587\\
  \cmidrule{2-9}
   &  & product & 0.013 & 0.018 & 0.017 & 0.013 & 0.047 & 0.043\\
  
   &  & adaptive & 0.108 & 0.113 & 0.078 & 0.082 & 0.055 & 0.059\\
  
   &  & full & 0.023 & 0.017 & 0.056 & 0.052 & 0.081 & 0.077\\
  
   &  & oracle & 0.023 & 0.017 & 0.056 & 0.052 & 0.081 & 0.077\\
  
   & \multirow{-5}{*}{\textbf{1000}} & parametric-linear & 0.007 & 0.008 & 1.568 & 1.568 & 1.568 & 1.571\\
  \cmidrule{2-9}
   &  & product & 0.013 & 0.014 & 0.009 & 0.010 & 0.030 & 0.029\\
  
   &  & adaptive & 0.089 & 0.090 & 0.070 & 0.069 & 0.055 & 0.057\\
  
   &  & full & 0.009 & 0.009 & 0.030 & 0.031 & 0.052 & 0.050\\
  
   &  & oracle & 0.009 & 0.009 & 0.031 & 0.031 & 0.052 & 0.050\\
  
   & \multirow{-5}{*}{\textbf{2000}} & parametric-linear & 0.006 & 0.005 & 1.552 & 1.554 & 1.553 & 1.553\\
  \cmidrule{2-9}
   &  & product & 0.002 & 0.003 & 0.011 & 0.010 & 0.031 & 0.029\\
  
   &  & adaptive & 0.059 & 0.061 & 0.048 & 0.049 & 0.032 & 0.034\\
  
   &  & full & 0.011 & 0.010 & 0.024 & 0.023 & 0.042 & 0.040\\
  
   &  & oracle & 0.011 & 0.009 & 0.024 & 0.023 & 0.042 & 0.040\\
  
  \multirow{-20}{*}{\textbf{Small}} & \multirow{-5}{*}{\textbf{4000}} & parametric-linear & 0.001 & 0.001 & 1.544 & 1.543 & 1.545 & 1.543\\
  \bottomrule
  \end{tabular}
  }
\end{table}

In \Cref{fig:bias-large,fig:bias-small,fig:coverage-large,fig:coverage-small}, the performance of the different estimators is plotted in different simulation scenarios. These estimators are: product, adaptive, full, oracle, and parametric-linear. The product and adaptive estimators perform simulatenous selection and estimation using the proposed and standard adaptive lasso weights, respectively, in \eqref{eq:est-med-lasso-obj}. The full and oracle estimators use least-squares and cross-fitting without selection, using $\model^F$ and $\model^*$, respectively. In other words, these estimators use \eqref{eq:est-med-lasso-obj} with different specifications of the weights: for the full estimator, they are set to zero; for the oracle estimator, those corresponding to $\model_Z^{*c}$ are set to $+\infty$ and the remaining weights are set to zero. The parametric-linear model makes direct use of \eqref{eq:med-model} with least-squares and a parametric linear model for $\psi_Y(\bX)$ and $\psi_M(\bX)$, under knowledge of $\model^*$. The coverage rates are presented in \Cref{fig:coverage-large,fig:coverage-small}, with bias results in \Cref{fig:bias-large,fig:bias-small}.

\vspace{-5em}

\begin{figure}
  \centering
  \includegraphics[page=1,width=\textwidth]{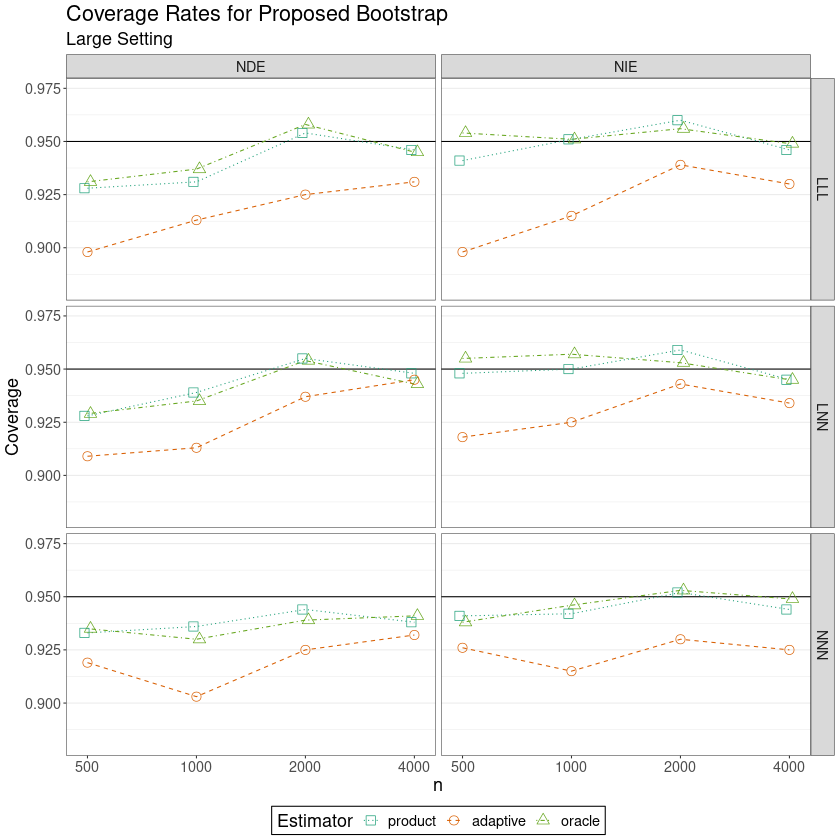}
  \caption[Natural effect bootstrap coverage rates in the Large coefficient case, $p=10$]{Coverage rates of 95\% confidence intervals in the Large setting. 
  }
  \label{fig:coverage-large}
\end{figure}

\begin{figure}
  \centering
  \includegraphics[page=1,width=\textwidth]{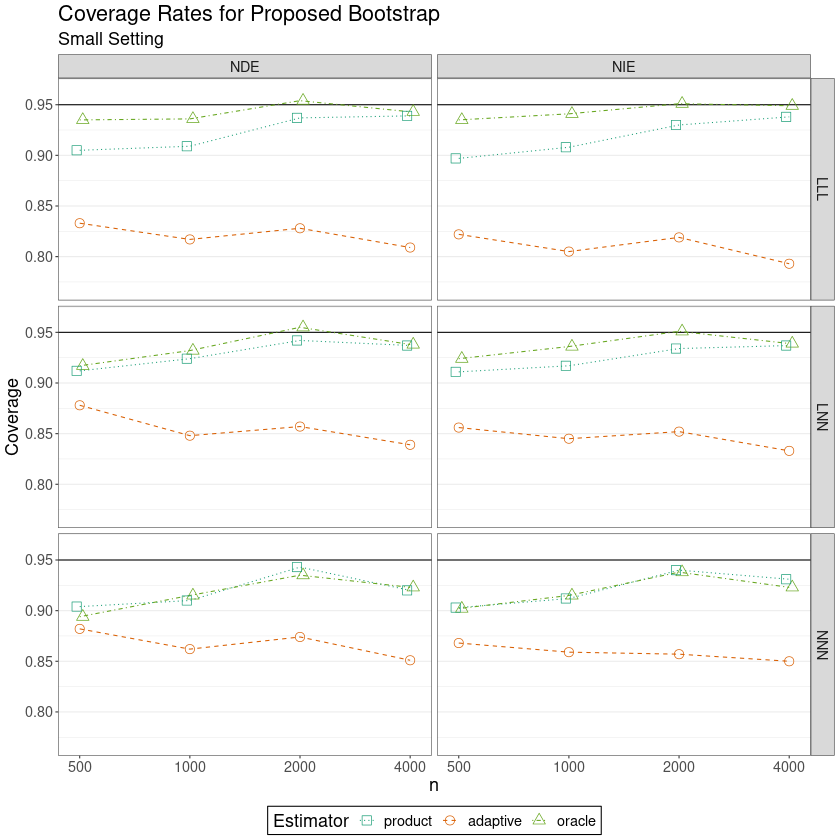}
  \caption[Natural effect bootstrap coverage rates in the Small coefficient case]{Coverage rates of 95\% confidence intervals in the Small setting. }
  \label{fig:coverage-small}
\end{figure}

\begin{figure}
  \centering
  \includegraphics[page=1,width=0.9\textwidth]{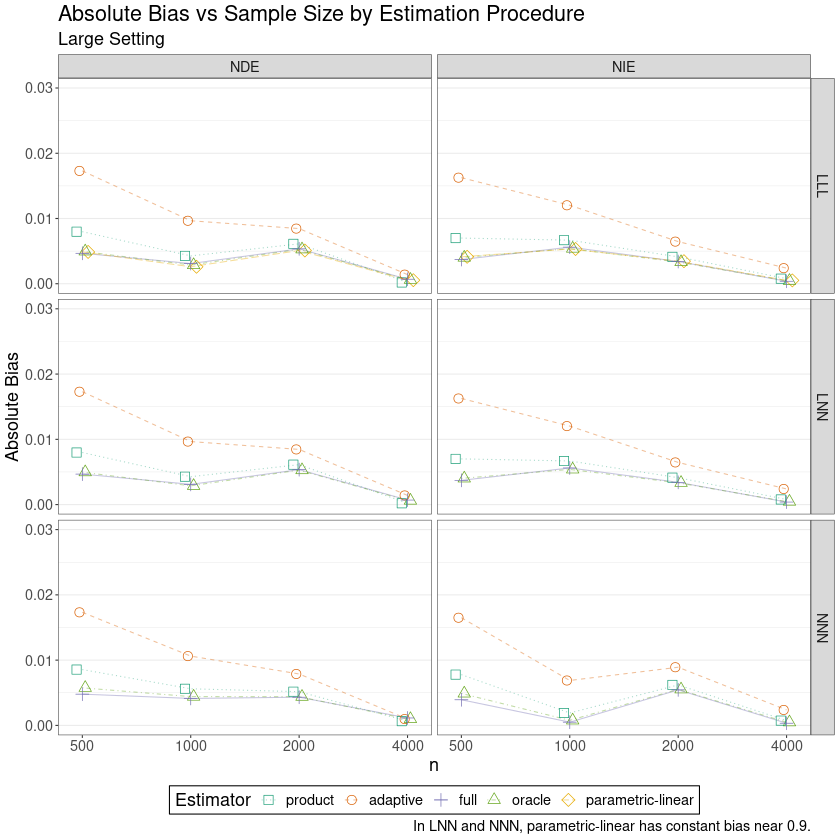}
  \caption[Natural effect bias in the Large coefficient case, $p=10$]{Bias of estimators in the Large coefficient case.
  In the LNN and NNN scenarios, the parametric-linear model does not appear on the graph as the absolute bias is of approximate size 0.5.
  }
  \label{fig:bias-large}
\end{figure}

\begin{figure}
  \centering
  \includegraphics[page=1,width=0.9\textwidth]{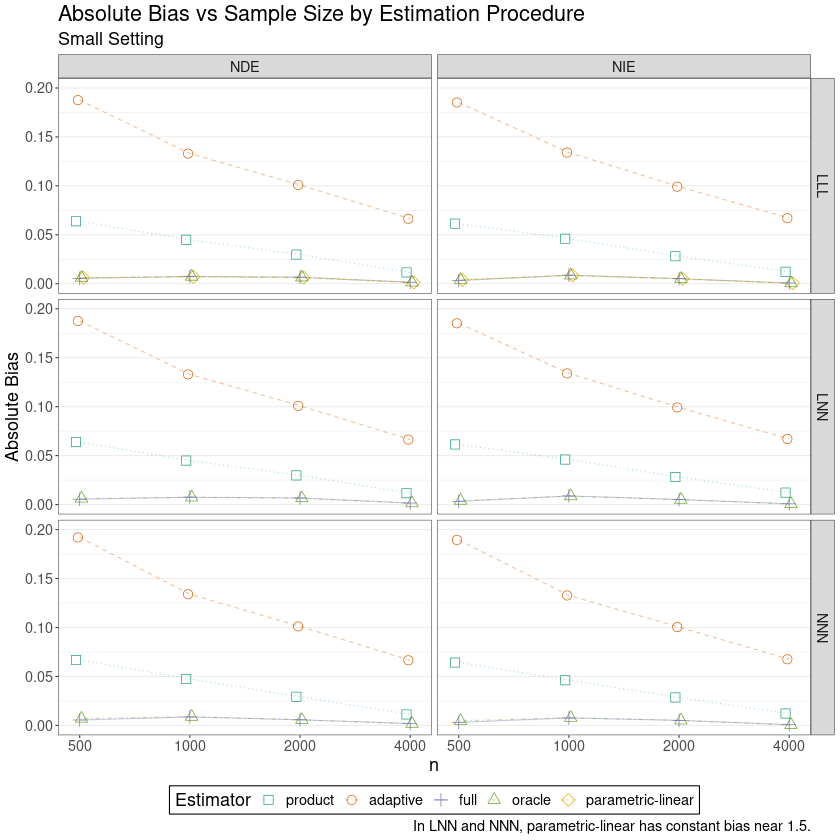}
  \caption[Natural effect bias in the Small coefficient case]{Bias of estimators in the Small coefficient case. 
  }
  \label{fig:bias-small}
\end{figure}

\FloatBarrier
\newpage
\section{Supplemental Simulation Results}
\label{sec:app-manyp-results}

\beginsuppplement[B]

\vspace{-.5em}

For these supplemental results, only the LNN scenario was run on a reduced number of sample size configurations in order to reduce the computational load.
In the Small Alpha setting, we set $\balpha_0 = 4(n^{-1/2}, n^{-1/2}, n^{-1/4}, \bzero^\top)^\top$, $\bbeta_0=(1, 1, n^{-1/4}, \bzero^\top)^\top$. These simulations used $p=10$, whereas the Large, $p=60$ simulations below used the fixed-coefficients described in \Cref{sec:app-sim-results}.

In the table below, we see that the move up to $p=60$ does not substantially change the coverage properties of either estimator. However, the additional irrelevant variables cause higher numbers of noise variables to be included on median than in the $p=10$ scenario. Notably, the product estimator tends to include at least 3 fewer noise variables on median. Both methods contain the true model in all 1000 simulations. Desirable properties in terms of coverage and bias are also demonstrated.

In the Small Alpha setting, both estimators again perform well. Due to the setup, most of the $\bbeta_0$ coefficients are fixed with $n$, and consequently are expected to behave as in the Large coefficient setting. Consistent with the local asymptotics, the adaptive lasso is able to distinguish the $n^{-1/4}$ coefficient from exact zero in almost all simulations. The product estimator also performs well in this setting, consistent with \Cref{thm:local-oracle}.

\begin{table}
  \caption{
    \label{tab:app-supp-results}
    Estimator performance in supplementary analysis, LNN scenario.}

  \begin{tabular}{ccccccccc}
    \toprule
    \multicolumn{3}{c}{ } & \multicolumn{2}{c}{Selection} & \multicolumn{2}{c}{Coverage} & \multicolumn{2}{c}{Bias} \\
    \cmidrule(l{3pt}r{3pt}){4-5} \cmidrule(l{3pt}r{3pt}){6-7} \cmidrule(l{3pt}r{3pt}){8-9}
    Coefficients & n & Weight Version & PC & MN & NDE & NIE & NDE & NIE\\
    \midrule
    &  & adaptive & 1.000 & 8.500 & 0.952 & 0.948 & 0.006 & 0.009\\
    
    & \multirow{-2}{*}{\textbf{1000}} & product & 1.000 & 5.000 & 0.953 & 0.948 & 0.016 & 0.012\\
    \cmidrule{2-9}
    &  & adaptive & 1.000 & 8.000 & 0.930 & 0.936 & 0.001 & 0.003\\
    
    \multirow{-4}{*}{\textbf{Large, p=60}} & \multirow{-2}{*}{\textbf{2000}} & product & 1.000 & 5.000 & 0.929 & 0.941 & 0.014 & 0.009\\
    \cmidrule{1-9}
    &  & adaptive & 0.998 & 1.000 & 0.934 & 0.951 & 0.001 & 0.008\\
    
    & \multirow{-2}{*}{\textbf{500}} & product & 1.000 & 1.000 & 0.937 & 0.949 & 0.005 & 0.004\\
    \cmidrule{2-9}
    &  & adaptive & 1.000 & 1.000 & 0.937 & 0.948 & 0.001 & 0.005\\
    
    \multirow{-4}{*}{\textbf{Small Alpha}} & \multirow{-2}{*}{\textbf{1000}} & product & 1.000 & 1.000 & 0.936 & 0.949 & 0.003 & 0.003\\
    \bottomrule
  \end{tabular}
\end{table}

\newpage

\section{Appendix: Proof of Theorems}
\label{sec:app-proof}

\subsection{Useful Lemmas and a Theorem}

Throughout these proofs, we will make use of the following lemmas, which are similar to those in \cite{ertefaieRobustQLearning2021} with the addition of allowing random multipliers.
\begin{lemma}
  \label{lem:mean-zero}
  Suppose $\bO_1,\ldots, \bO_n$ are i.i.d. random vectors, and $(R_i, \bX_i)$ are $\bO_i$-measurable random variables. Use the cross-fitting notation defined in \Cref{sec:conf-control} so that $\cup_{k=1,\ldots,K} \crossfitData = (\bO_1,\ldots, \bO_n)$. 
  Recall the notation of $\crossfitDataC$ for $(\bO_i$ for $i \notin \Ik)$. 
   Let $\Eop( R_i \vert \bX_i ) = \Eop(R_i \vert \bX_i, \crossfitDataC) = 0$ and $0 < \Var(R_i \vert \bX_i) = \Var(R_i \vert \bX_i, \crossfitDataC) \leq C_1 < \infty$ for all $i=1,\ldots,n$. Then
  \begin{equation}
    \label{eq:lem-mean-zero-1}
    \bigg\vert n^{-1} \sum_{i \in \Ik} R_i h(\bX_i; \crossfitDataC) \bigg\vert = O_p(n^{-1/2} \Ltwo{h} ),
  \end{equation}
  where $\Ltwo[P]{\cdot}$ for $k=1,\ldots,K$ is defined in \eqref{eq:ltwo-crossfit}.
  In particular, if $\Ltwo{h} = O_p((n-n_k)^{-a}) = O_p(n^{-a})$ for $k=1,\ldots,K$, then
  \begin{equation}
    \label{eq:lem-mean-zero-2}
    \bigg\vert n^{-1} \sum_{i \in \Ik} R_i h(\bX_i; \crossfitDataC) \bigg\vert = O_p(n^{-(1+a)/2} ).
  \end{equation}
  Finally, if we let $R_i = G_i H_i$, where $G_1,\ldots,G_n$ are i.i.d. random variables with constant variance $\Eop G_1^2 < \infty$, which are independent of $(H_i, \bX_i)$, and the random variables $H_i$ otherwise satisfy the conditions on $R_i$ above, then \eqref{eq:lem-mean-zero-1} and \eqref{eq:lem-mean-zero-2} hold.
\end{lemma}
\begin{proof}
  Let
  \[
    S_n = n^{-1} \sum_{i \in \Ik} R_i h(\bX_i; \crossfitDataC)
  \]
  and using iterated expectations, 
  \[
    \Eop (S_n \vert \crossfitDataC) = \frac{n_k}{n}\Eop\Big\{ h(\bX_1; \crossfitDataC) \Eop \big(R_1 \big\vert \bX_1, \crossfitDataC \big) \Big\vert \crossfitDataC \Big\} = 0.
  \]
  Then we can similarly calculate the conditional variance as
  \begin{align*}
    \Var(S_n \vert \crossfitDataC) =& \Eop\Big( S_n^2 \big\vert \crossfitDataC \Big) \\
    =& n^{-2} \sum_{i \in \Ik} \Eop\Big\{ h(\bX_i; \crossfitDataC)^2 \Var(R_i \vert \bX_i) \big\vert \crossfitDataC \Big\}  \\
    \leq& C_1 n^{-1} \Eop\Big\{ n^{-1} \sum_{i \in \Ik} h(\bX_i; \crossfitDataC)^2 \Big\vert \crossfitDataC \Big\}
  \end{align*}
  where the second line follows from $R_i \independent R_j$ and hence $\Eop( R_i R_j \vert \bX_i, \bX_j, \crossfitDataC)=0$ for $i \neq j \in \Ik$. The final statement is true by assumption. 
  By applying Chebyshev's Inequality, using the definition of $\Ltwo{h}$, and noting $n_k/n \rightarrow K^{-1} < \infty$, we have
  \begin{align*}
    \Prob( \vert S_n \vert > \sqrt{n} \varepsilon ~|~ \crossfitDataC) \leq& \varepsilon^{-1} C_1 n_k n^{-1} \Ltwo{h}^2 \\
    =& O_p( \Ltwo{h}^2 )
  \end{align*}
  which establishes \eqref{eq:lem-mean-zero-1}. \Cref{eq:lem-mean-zero-2} follows directly. 

  The final statement of the Lemma follows by noting that if $R_i = G_i H_i$ with $G_i$ independent of the remaining random variables, then we can verify that the hypotheses of the Lemma hold with respect to $R_i$. These are the zero-mean assumption $\Eop(R_i \vert \bX_i) = \Eop(R_i \vert \bX_i, \crossfitDataC)=0$ and the finite-variance assumption $\Var(R_i \vert \bX_i) = \Var(R_i \vert \bX_i, \crossfitDataC) < C_2$. The zero-mean assumption follows from \[ 0= \Eop(R_i \vert \bX_i) = \Eop G_i \Eop(H_i \vert \bX_i) = \Eop G_i \Eop(H_i \vert \bX_i, \crossfitDataC) = \Eop(R_i \vert \bX_i, \crossfitDataC) \]
  where the string of equalities hold due to independence of $G_i$ from the remaining variables. Finally, if
  \( \Var(H_i \vert \bX_i) = \Var(H_i \vert \bX_i, \crossfitDataC) < C_1 \)
  then the finite variance assumption holds with $C_2 = C_1 \Eop G_1^2$.

\end{proof}

\begin{lemma}
  \label{lem:crossprod-terms}
  Assume the notation and setup in the statement of \ref{lem:mean-zero}. Additionally, let $h_1, h_2$ be measurable functions of $\bX$ and $\crossfitDataC$. Also, for fixed $\gamma \in \{0, 1\}$, let $V(\gamma)= (1-\gamma) + \gamma\Eop G_1^2$. Then
  \[
    \bigg\vert n^{-1} \sum_{i \in \Ik} G_i^{\gamma} h_1(\bX_i; \crossfitDataC) h_2(\bX_i; \crossfitDataC) \bigg\vert = O_p( \Ltwo{h_1} \Ltwo{h_2} ).
  \]
  In particular, if $\Ltwo{h_1} = o_p((n-n_k)^{-a_1/2})$ and $\Ltwo{h_2} = o_p((n-n_k)^{-a_2/2})$, then
  \[
    \bigg\vert n^{-1} \sum_{i \in \Ik} G_i^\gamma h_1(\bX_i; \crossfitDataC) h_2(\bX_i; \crossfitDataC) \bigg\vert = o_p((n-n_k)^{-(a_1 + a_2)/2} ) = o_p(n^{-(a_1 + a_2)/2} ).
  \]
\end{lemma}
\begin{proof}
  This is a rather straightforward application of the Cauchy-Schwarz inequality. Letting $S_n = n^{-1} \sum_{i \in \Ik} G_i^{\gamma} h_1(\bX_i; \crossfitDataC) h_2(\bX_i; \crossfitDataC) $, apply the triangle inequality and the Cauchy-Schwarz inequality to $\vert S_n \vert$ to obtain
  \begin{align*}
    \vert S_n \vert \leq& \bigg\{n^{-1} \sum_{i \in \Ik} G_i^{2\gamma} h_1(\bX_i; \crossfitDataC)^2 \bigg\}^{1/2} \bigg\{n^{-1} \sum_{i \in \Ik} h_2(\bX_i; \crossfitDataC)^2 \bigg\}^{1/2} \\
    =& (n_k/n) \bigg\{n_k^{-1} \sum_{i \in \Ik} G_i^{2\gamma} h_1(\bX_i; \crossfitDataC)^2 \bigg\}^{1/2} \bigg\{n_k^{-1} \sum_{i \in \Ik} h_2(\bX_i; \crossfitDataC)^2 \bigg\}^{1/2}
  \end{align*}
  Now take expectations of both sides and apply Jensen's Inequality to the concave function $x \mapsto x^{1/2}$ to obtain the following sequence of inequalities:
  \begin{align*}
    \Eop\left(  \vert S_n \vert ~\big|~ \crossfitDataC \right) \leq& (n_k/n) \Eop\Bigg[ \bigg\{n_k^{-1} \sum_{i \in \Ik} G_i^{2\gamma} h_1(\bX_i; \crossfitDataC)^2 \bigg\}^{1/2} \bigg\{n_k^{-1} \sum_{i \in \Ik} h_2(\bX_i; \crossfitDataC)^2 \bigg\}^{1/2} ~\Bigg|~ \crossfitDataC \Bigg] \\
    \leq& (n_k/n) \Bigg[ n_k^{-1} \sum_{i \in \Ik} \Eop\Big\{  G_i^{2\gamma} h_1(\bX_i; \crossfitDataC)^2 ~\Big|~ \crossfitDataC \Big\} \Bigg]^{1/2} \Bigg[ n_k^{-1} \sum_{i \in \Ik} \Eop\Big\{ h_2(\bX_i; \crossfitDataC)^2 ~\Big|~ \crossfitDataC  \Big\} \Bigg]^{1/2}.
  \end{align*}
  The independence between $G_i$ and $\bO_i$ ensures that this final line equals
  \[
    (n_k/n) V(\gamma)^{1/2} \Ltwo{h_1} \Ltwo{h_2}.
  \]
  By assumption, $(n_k/n) V(\gamma)^{1/2} = O(1)$.
  By applying Markov's inequality, we see
  \[
    \Prob(\vert S_n \vert > \varepsilon ~|~ \crossfitDataC) \leq \varepsilon^{-1} \Eop\vert S_n \vert
  \]
  and by combining with the previous bound on $\Eop\vert S_n \vert$, we obtain the result.
\end{proof}

Several of the main proofs make heavy use of Theorem 3 of \cite{arconesAsymptoticTheoryMestimators1998}, which we restate here in slightly less generality.
\begin{theorem}
  \label{thm:arcones}
  Let $\Theta$ be a subset of $\R^d$. Let $\{ F_n(\btheta) : \btheta \in \Theta \}$ be a sequence of stochastic processes. Let $\btheta_0$ be a point in the interior of $\Theta$. Let $\{ V_n \}$ be a sequence of nonsingular symmetric $d\times d$ matrices. Let $\{ G_n \}$ be a sequence of $\R^d$-valued r.v.'s. Suppose the following conditions hold.
  \begin{enumerate}
    \item $F_n(\btheta)$ is a convex function in $\btheta$.
    \item $\hat\btheta$ satisfies $F_n(\hat\btheta) \leq \inf_{\btheta \in \Theta} F_n(\btheta) + o_p(1)$.
    \item {For each $\btheta \in \R^d$, \[ 
      F_n(\btheta_0 + n^{-1/2}\btheta) - F_n(\btheta_0) - 2\btheta^\top G_n + \btheta^\top V_n \btheta \CiP 0 
    \]
    }
    \item $G_n = O_p(1)$.
    \item $\liminf_{n\rightarrow \infty}\inf_{|\btheta|=1} \btheta^\top V_n \btheta > 0$ and $\limsup_{n\rightarrow \infty} \sup_{|\btheta|=1} \btheta^\top V_n \btheta  < \infty$.
  \end{enumerate}
  Then,
  \[ \left\Vert\sqrt{n}\left( \hat\btheta - \btheta_0 \right) - V_n^{-1} G_n \right\Vert_{\infty} \CiP 0 \]
\end{theorem}

\subsection{Approximation Results for Least-Squares Robinson-Transformed Loss}

Momentarily, let $\modelspace$ refer to the set of all subsets of $\model^F\defined\{1,\ldots,p+1\}$. That is, $\model \in \modelspace$ will index elements of $Z$ to simplify some notation. For any $\model \in \modelspace$, define the submodel-specific loss function for any $\bmu \in [L_2(P_0)(\Xcal)]^{p+2}$ and any $\btheta_{\model} \in \R^{|\model|}$:
\begin{equation*}
  L_{\model}(\btheta_{\model}; \bmu, \bO) = \left[ Y - \mu_{Y}(\bX) - \btheta_{\model}^\top \{ \bZ - \bmu_{Z}(\bX) \}(\model) \right]^2.
\end{equation*}
This loss function implies gradient and Hessian:
\begin{align*}
  \nabla L_{\model}(\btheta_{\model}; \bmu, \bO) =& -2 \{ \bZ - \bmu_{Z}(\bX) \}(\model) \left[ Y - \mu_{Y}(\bX) - \btheta_{\model}^\top \{ \bZ - \bmu_{Z}(\bX) \}(\model) \right]  \\
  \nabla^2 L_{\model}(\btheta_{\model}; \bmu, \bO) =& 2 \{ \bZ - \bmu_{Z}(\bX) \}(\model)^{\otimes 2}.
\end{align*}
Importantly, this allows the expression for $\Lbar_{\model}(\btheta_{\model}; \bmu, \bO) = L_{\model}(\btheta_{\model}; \bmu, \bO) - L_{\model}(\btheta_{0\model}; \bmu, \bO)$:
\begin{equation}
  \label{eq:sub-loss-quad}
  \Lbar_{\model}(\btheta_{\model}; \bmu, \bO) = (\btheta_{\model} - \btheta_{0\model})^\top \nabla L_{\model}(\btheta_{0\model}; \bmu, \bO) + \frac{1}{2} (\btheta_{\model} - \btheta_{0\model})^\top \nabla^2 L_{\model}(\btheta_{0\model}; \bmu, \bO) (\btheta_{\model} - \btheta_{0\model}),
\end{equation}
a quadratic equation in $\btheta_{\model} - \btheta_{0\model}$. This follows from the Taylor expansion of $L_{\model}$ about $\btheta_{0\model}$, and is absent of remainder since the loss itself is quadratic in $\btheta_{\model}$. We will also make use of the associated ``oracle'' empirical risks
\[
  R_{n\model}(\btheta_{\model}) = n^{-1} \sum_{i=1}^n L_{\model}(\btheta_{\model}; \bmu_0, \bO_i)
\]
and the cross-fitted empirical risks which leverage the estimators $\hat\bmu_k$ of $\bmu_0$ for $k=1,\ldots,K$:
\[
  \hat R_{n\model}(\btheta_{\model}) = n^{-1} \sum_{k=1}^K \sum_{i \in \Ik} L_{\model}(\btheta_{\model}; \hat\bmu_k, \bO_i).
\]
Then defining $\bG_0 = \Eop\left[ \{\bZ - \bmu_{Z0}(\bX)\} \{Y - \mu_{Y0}(\bX)\}\right]$, we may state the following proposition.
\begin{proposition}
  \label{prop:ch1-plm-submodel}
  Under \Cref{asu:ch2-bounded-variance}, the parameter
  \begin{equation}
    \btheta_{0\model} = \bH_0(\model)^{-1} \bG_0(\model)
    \label{eq:ch1-submodel-param}
  \end{equation}
  is defined for all $\model\in\modelspace$ and the residuals
  \[
    \epsilon_{\model} = Y - \mu_{Y0}(\bX) - \btheta_{0\model}^\top \{\bZ - \bmu_{Z0}(\bX)\}(\model)
  \]
  have conditional mean zero:
  $\Eop(\epsilon_{\model} ~|~ \bX)=0$. Finally,
  \begin{equation}
    \Eop \left[ \epsilon_{\model} \left\{ \bZ(\model) - \bmu_{Z0}(\bX)(\model) \right\} \right] = \bzero.
    \label{eq:ch1-submodel-error-uncorrelated}
  \end{equation}
\end{proposition}
\begin{proof}
  To begin, we will remark that the Cauchy Interlace Theorem for principal submatrices ensures that $\lambda_{max}(\bH_0) < c < \infty$ implies $\lambda_{max}(\bH_0(\model)) < c$, ensuring the existence of $\bH_0(\model)^{-1}$.
  The first display is a direct consequence of this fact. The second display follows by linearity of the expectation operator.
  To see \eqref{eq:ch1-submodel-error-uncorrelated}, we note that the LHS is equivalent to 
  \[
    \Eop\left[ \nabla L_{\model}(\btheta_{0\model}; \bmu, \bO) \right] = \bzero.
  \] 
  Notice that \eqref{eq:ch1-submodel-param} is the solution to
  \begin{equation}
    \bG_0(\model) = \bH_0(\model) \btheta_{0\model},
    \label{eq:app-ch1-expected-normal}
  \end{equation}
  which results from the first-order condition
  \begin{equation}
    \bzero = \Eop\left[ \nabla R_{n\model}(\btheta_{0\model}) \right] = \Eop\left[ \nabla L_{\model}(\btheta_{0\model}; \bmu, \bO) \right].
    \label{eq:app-ch1-nabla-inside}
  \end{equation}
  Exchanging derivatives and expectation yields the conditions for the submodel population minimizer
  \begin{equation}
    \bzero = \nabla \Eop\left[ L_{\model}(\btheta_{0\model}; \bmu, \bO) \right].
    \label{eq:app-ch1-minimizer}
  \end{equation}
  Since $\bH_0 \succ 0$ implies that $\Eop R_{n\model}(\cdot)$ is a strictly convex function with a unique minimum, \eqref{eq:app-ch1-minimizer}
  is both necessary and sufficient for $\btheta_{0\model}$ to minimize the risk. 
  The equivalence between \eqref{eq:app-ch1-nabla-inside} and \eqref{eq:app-ch1-minimizer} follows from standard results 
  \citep[e.g.~Cor. 2.4.4][]{casellaStatisticalInference2002}. 
  This follows from the bounded derivative of the submodel losses, which we establish by the Taylor's Theorem arguments which led to \eqref{eq:sub-loss-quad}, along with the condition $\lambda_{max}(\bH_0) < \infty$:
  \begin{align*}
    \nabla L_{\model}(\btheta_{\model}; \bmu_0, \bO) = \nabla L_{\model}(\btheta_{0\model}; \bmu_0, \bO) + \nabla^2 L_{\model}(\btheta_{0\model}; \bmu_0, \bO)(\btheta_{\model} - \btheta_{0\model})  \\
    \Vert \nabla^2 L_{\model}(\btheta_{0\model}; \bmu_0, \bO)(\btheta_{\model} - \btheta_{0\model}) \Vert_{\infty}
    \leq 2 \Vert \{ \bZ(\model) - \bmu_{Z0}(\bX)(\model)\}^{\otimes 2} \Vert_{\infty} \Vert \btheta_{\model} - \btheta_{0\model} \Vert_{1} \\
    \leq 2 \Vert \bZ(\model) - \bmu_{Z0}(\bX)(\model) \Vert_{2}^2 \Vert \btheta_{\model} - \btheta_{0\model} \Vert_{1} \\
    \text{ where } \Eop \Vert \bZ(\model) - \bmu_{Z0}(\bX)(\model) \Vert_{2}^2 \leq |\model| \lambda_{max}(\bH_0)
  \end{align*}
  Note that $\bG_{0}(\model) = \Eop \left\{ \nabla L_{\model}(\btheta_{0\model}; \bmu_0, \bO) \right\} $ satisfies $\Vert \bG_{0}(\model) \Vert_{\infty} < \infty$ by assumption. Then for every $\model \in \modelspace$, and all $\Vert \btheta_{\model} - \btheta_{0\model} \Vert_{1} < \delta$,
  \[
    \Vert \nabla L_{\model}(\btheta_{\model}; \bmu_0, \bO) \Vert_{\infty} \leq g(\btheta_{0\model}; \bO) \equiv \Vert \nabla L_{\model}(\btheta_{0\model}; \bmu_0, \bO) \Vert_{\infty} + 2 \Vert\bZ(\model) - \bmu_{Z0}(\bX)(\model) \Vert_{2}^2 \delta
  \]
  which satisfies $\Eop g(\btheta_{0\model}; \bO) < \infty$, completing the application of the stated Corollary.
\end{proof}

At this point, we are ready to examine the pointwise difference between the estimated empirical risk and the empirical risk with all nuisance parameters known. To examine this, we re-state some simplified assumptions which allow us to establish some results. The following assumptions are a subset of those made in \Cref{sec:theory}.

\begin{assumption}[Bounded Parameter]
  \label{asu:appendix-bounded-theta0}
  The parameter $\btheta_0$ satisfies $\Vert\btheta_0\Vert_{\infty} < \infty$.
\end{assumption}

\begin{assumption}[Learning Rates]
  Under cross-fitting setup, the following rates hold for $k=1,\ldots,K$ and $j=1,\ldots,p$:
  \begin{align*}
    \Ltwo{\hat\mu_{Y} - \mu_{Y0}}  = o_p(1)
    \\ 
    \Ltwo{\hat\mu_{Zj}- \mu_{Z0j}} = o_p((n-n_k)^{-1/4}).
    \\
    \Ltwo{\hat\mu_{Y} - \mu_{Y0}} \cdot  \Ltwo{\hat\mu_{Zj}- \mu_{Z0j}}  = o_p((n-n_k)^{-1/2})
  \end{align*}
  
  \label{asu:appendix-rates}
\end{assumption}

\begin{assumption}[Non-degenerate Gram Matrix]
  The matrix $\bH_0$ satisfies \[0 < \lambda_{min}(\bH_0) < \lambda_{max}(\bH_0) < \infty.\]
  \label{asu:ch1-bounded-gram}
  \correctThmDisplayMath
\end{assumption}
\noindent
With these assumptions, we may examine the approximation of $R_{n\model}(\btheta_{\model})$ by $\hat R_{n\model}(\btheta_{\model})$. We will make use of the following quantities:
\begin{align}
  \label{eq:loss-diff-repr}
  \frac{1}{n} \sum_{k=1}^K \sum_{i \in \Ik} \big\{ \Lbar_{\model}\big(\btheta_{\model}; \hat\bmu_k, \bO_i \big) - \Lbar_{\model}\big(\btheta_{\model}; \bmu_0, \bO_i \big) \big\} = (\btheta_{\model} - \btheta_{0\model})^\top \oneDerivDiffmodel  \\
  \nonumber
  + (\btheta_{\model} - \btheta_{0\model})^\top \twoDerivDiffmodel (\btheta_{\model} - \btheta_{0\model}) \\
  \label{eq:grad-loss-repr}
  \frac{1}{n} \sum_{k=1}^K \sum_{i \in \Ik} \big\{ \nabla L_{\model}\big(\btheta; \hat\bmu_k, \bO_i \big) - \nabla L_{\model}\big(\btheta; \bmu_0, \bO_i \big) \big\} = \oneDerivDiffmodel + 2 \twoDerivDiffmodel (\btheta_{\model} - \btheta_{0\model}),
\end{align}
where we define the quantities
\begin{align*}
  \oneDerivDiffmodel =&  \frac{1}{n} \sum_{k=1}^K \sum_{i \in \Ik} \big\{ \nabla L_{\model}\big(\btheta_{0\model}; \hat\bmu_k, \bO_i \big) - \nabla L_{\model}\big(\btheta_{0\model}; \bmu_0, \bO_i \big) \big\} \\
  \twoDerivDiffmodel =& \frac{1}{2n} \sum_{k=1}^K \sum_{i \in \Ik} \Big\{ \nabla^2 L_{\model}(\btheta_{0\model}; \hat\bmu_k, \bO_i) - \nabla^2 L_{\model}(\btheta_{0\model}; \bmu_0, \bO_i) \Big\}.
\end{align*}

\begin{lemma}
  \label{lem:uniform-conv-loss}
  Suppose \Cref{asu:appendix-bounded-theta0,asu:appendix-rates} hold. Then $\Vert \oneDerivDiffmodel \Vert_\infty = \Vert \twoDerivDiffmodel \Vert_\infty=o_p(n^{-1/2})$, where $\Vert A \Vert_\infty$ represents the element-wise maximum of the matrix $A$. Furthermore for compacts $\Theta_{\model}$ of $\R^{|\model|}$ and $C < \infty$,
  \begin{equation}
    \label{eq:lem-unif-conv}
    \begin{aligned}
      \sup_{\btheta_{\model} \in \Theta_{\model}} \bigg\vert \frac{1}{n} \sum_{k=1}^K \sum_{i \in \Ik} \big\{ \Lbar_{\model}\big(\btheta_{\model}; \hat\bmu_k, \bO \big) - \Lbar_{\model}\big(\btheta_{\model}; \bmu_0, \bO \big) \big\} \bigg\vert = o_p(n^{-1/2}) \\
      \sup_{\btheta_{\model} \in \Theta_{\model}} \bigg\Vert \frac{1}{n} \sum_{k=1}^K \sum_{i \in \Ik} \big\{ \nabla L_{\model}\big(\btheta_{\model}; \hat\bmu_k, \bO \big) - \nabla L_{\model}\big(\btheta_{\model}; \bmu_0, \bO \big) \big\} \bigg\Vert_\infty = o_p(n^{-1/2}) \\
      \sup_{\Vert \bu_{\model} \Vert_\infty \leq C} \bigg\vert \frac{1}{n} \sum_{k=1}^K \sum_{i \in \Ik} \big\{ \Lbar_{\model}\big(\btheta_{0\model} + n^{-1/2}\bu_{\model}; \hat\bmu_k, \bO \big) - \Lbar_{\model}\big(\btheta_{0\model} + n^{-1/2}\bu_{\model}; \bmu_0, \bO \big) \big\} \bigg\vert = o_p(n^{-1}).
    \end{aligned}
  \end{equation}
\end{lemma}
\begin{proof}
  Starting with $\oneDerivDiffmodel$, recall $\epsilon_{i\model} = \left\{ Y_i - \mu_{Y0i} - \btheta_{0\model}^\top (\bZ_i - \bmu_{Z0i})(\model) \right\}$, 
  and let \[
    \hat\epsilon_{i\model} = \left\{ Y_i - \hat\mu_{Yi} - \btheta_{0\model}^\top (\bZ_i - \hat\bmu_{Zi})(\model) \right\}.
  \] Then
  \begin{align*}
    \oneDerivDiffmodel =& -2 n^{-1} \sum_{k=1}^K \sum_{i \in \Ik} \Big\{ \hat\epsilon_{i\model} ( \bZ_i - \hat\bmu_{Zi})(\model) - \epsilon_{i\model} ( \bZ_i - \bmu_{Z0i} )(\model) \Big\} \\
    =& -2 n^{-1} \sum_{k=1}^K \sum_{i \in \Ik} \Big\{ (\hat\epsilon_{i\model} - \epsilon_{i\model}) ( \bZ_i - \bmu_{Z0i})(\model) - \epsilon_{i\model} ( \hat\bmu_{Zi} - \bmu_{Z0i} )(\model) + (\epsilon_{i\model} - \hat\epsilon_{i\model}) ( \hat\bmu_{Zi} - \bmu_{Z0i} )(\model) \Big\},
  \end{align*}
  which is a $|\model| \times 1$ vector. 
  Notice that $\hat\epsilon_{i\model} - \epsilon_{i\model} = \mu_{Y0i} - \hat\mu_{Yi} + \btheta_{0\model}^\top( \bmu_{Z0i} - \hat\bmu_{Zi} )(\model)$ is a function of the data only through $\crossfitDataC$ and $\bX_i$ for $i \in \Ik$. Write $ \hat\epsilon_{i\model} - \epsilon_{i\model} \equiv (\hat\epsilon_{i\model} - 
  \epsilon_{i\model})(\bX_i; \crossfitDataC)$. For each $j \in \model$, there is a corresponding entry in $\oneDerivDiffmodel$ given by
  \begin{align}
    \label{eq:An-terms-1}
    -2 \sum_{k=1}^K \frac{n_k}{n} n_k^{-1} \sum_{i \in \Ik} \Big\{ &
      (\hat\epsilon_{i\model} - 
    \epsilon_{i\model})(\bX_i; \crossfitDataC) ( \bZ_i - \bmu_{Z0i})\subj[j] \\
    \label{eq:An-terms-2}
      &- \epsilon_{i\model} ( \hat\bmu_{Zi} - \bmu_{Z0i} )\subj[j] \\
    \label{eq:An-terms-3}
      &+ (\hat\epsilon_{i\model} - 
      \epsilon_{i\model})(\bX_i; \crossfitDataC) ( \hat\bmu_{Zi} - \bmu_{Z0i} )\subj[j] \Big\}
  \end{align}
  where we remind the reader of the notation $\bZ'\subj[j]$ representing the $j^{th}$ element of the vector $\bZ'$.
  Since $n_k / n \rightarrow K^{-1}$ for each $k=1,\ldots,K$ with $K$ finite, we may focus on the inner sum normalized by $n_k$. We consider each term individually.

  \Cref{lem:mean-zero} applies to both \eqref{eq:An-terms-1} and \eqref{eq:An-terms-2}, since $\Eop(\epsilon_{\model} ~|~ \bX) = \Eop[ \{\bZ - \bmu_{Z0}(\bX)\}\subj[j] ~|~ \bX ] = 0$. This lemma ensures these terms are $o_p(n^{-1/2})$ if $\LtwoPzero{\hat\epsilon_{\model} - 
  \epsilon_{\model}} = \LtwoPzero{(\hat\bmu_{Z} - 
  \bmu_{Z0})\subj[j]} = o_p(1)$. Since $\LtwoPzero{(\hat\bmu_{Z} - 
  \bmu_{Z0})\subj[j]} = o_p(n^{-1/4})$ under \Cref{asu:appendix-rates}, this handles \eqref{eq:An-terms-2}.  
  Note that the norm $\LtwoPzero{\cdot}$ satisfies the triangle inequality:
  \[
    \LtwoPzero{\hat\epsilon_{\model} - 
    \epsilon_{\model}} \leq \LtwoPzero{\mu_{Y0} - \hat\mu_{Y}} + \sum_{j=1}^{p+1} \theta_{0\model}(\{j\}) \LtwoPzero{\mu_{Z0j} - \hat\mu_{Zj}},
  \]
  which is inherited from the $L_2$ norm on which it is based. Consequently, the above display is $o_p(1)$ under \Cref{asu:appendix-bounded-theta0,asu:appendix-rates}, handling \eqref{eq:An-terms-1}.

  Finally, \Cref{lem:crossprod-terms} applies to \eqref{eq:An-terms-3}. Expand the final term to obtain
  \begin{align}
    \label{eq:An-terms-4}
    -2 \sum_{k=1}^K \frac{n_k}{n} n_k^{-1} \sum_{i \in \Ik} \Big\{&
      (\mu_{Y0i} - \hat\mu_{Yi}) ( \bmu_{Z0i} - \hat\bmu_{Zi} )(\{j\}) 
      \\
      &+ \label{eq:An-terms-5}
      \sum_{\ell=1}^{p+1} \btheta_{0\model}(\{\ell\})^\top( \bmu_{Z0i} - \hat\bmu_{Zi} )(\{\ell\}) ( \bmu_{Z0i} - \hat\bmu_{Zi} )(\{j\})
    \Big\}
  \end{align}
  \Cref{asu:appendix-rates} ensures \eqref{eq:An-terms-4} is $o_p(n^{-1/2})$. The final line \eqref{eq:An-terms-5} results directly from \Cref{asu:appendix-bounded-theta0,asu:appendix-rates} and \Cref{lem:crossprod-terms}.

  Next we examine the $\twoDerivDiffmodel$ matrix:
  \begin{equation*}
    \twoDerivDiffmodel= n^{-1} \sum_{k=1}^K \sum_{i \in \Ik} \Big\{ (\bZ_i - \bmu_{Z0i} + \bmu_{Z0i} - \hat\bmu_{Zi})(\model)^{\otimes 2} - (\bZ_i - \bmu_{Z0i})(\model)^{\otimes 2} \Big\} ,
  \end{equation*}
  which is a $|\model| \times |\model|$ matrix. For $j, l \in \model$, there is a corresponding entry:
  \begin{multline}
    \label{eq:Bn-terms}
    \sum_{k=1}^K \frac{n_k}{n} n_k^{-1} \sum_{i \in \Ik} \Big\{ (\bZ_i - \bmu_{Z0i})(\{j\}) (\bmu_{Z0i} - \hat\bmu_{Zi})(\{l\}) + (\bmu_{Z0i} - \hat\bmu_{Zi})(\{j\}) (\bZ_i - \bmu_{Z0i})(\{l\})
    \\+ (\bmu_{Z0i} - \hat\bmu_{Zi})(\{j\}) (\bmu_{Z0i} - \hat\bmu_{Zi})(\{l\}) \Big\}.
  \end{multline}
  As before, $n_k/n \rightarrow K^{-1} < \infty$ for each $k=1,\ldots,K < \infty$, so that we focus on the inner sum. \Cref{lem:mean-zero} applies to the first two terms of this inner sum, and \Cref{lem:crossprod-terms} applies to the final term. Thus under \Cref{asu:appendix-rates}, $\Vert \twoDerivDiffmodel \Vert_\infty = o_p((n-n_k)^{-1/2}) = o_p(n^{-1/2})$.

  The first two equations in the display \eqref{eq:lem-unif-conv} follow from representations \eqref{eq:loss-diff-repr} and \eqref{eq:grad-loss-repr}, respectively. Since $\oneDerivDiffmodel$ and $\twoDerivDiffmodel$ do not depend on the minimand $\btheta_{\model}$, the $o_p(n^{-1/2})$ result is uniform in $\btheta_{\model}$. From compactness, there exists $M>0$ such that $\Vert \btheta_{\model} - \btheta_{0\model} \Vert_\infty < M < \infty$ for all $\btheta_{\model} \in \Theta_{\model}$; then we can bound \eqref{eq:loss-diff-repr} by $ |\model| M \Vert \oneDerivDiffmodel \Vert_\infty + (|\model| M)^2 \Vert \twoDerivDiffmodel \Vert_\infty $, which establishes the first equation of the display. 
  The second equation of the display in the Lemma statement can be deduced similarly.

  By taking the change of variables $\btheta_{\model} - \btheta_{0\model} = n^{-1/2}\bu_{\model}$, for $\Vert \bu_{\model} \Vert_\infty \leq C$, we can let $M = n^{-1/2}C$ in the above derivation, so that the LHS of the third equation \eqref{eq:lem-unif-conv} is bounded by
  \[
    |\model| C n^{-1/2} \Vert \oneDerivDiffmodel \Vert_\infty + (|\model| C n^{-1/2})^2 \Vert \twoDerivDiffmodel \Vert_\infty = o_p(n^{-1}) + o_p(n^{-3/2})
  \]
  which yields the final result.
  
\end{proof}

The outcome of Lemma \ref{lem:uniform-conv-loss} is that, heuristically, 
the empirical risk function and its gradient with nuisance parameters known is approximated using cross-fitting on compacts of $\btheta_{\model}$. Furthermore, the rate of uniform approximation improves to $o_p(n^{-1})$ for the first quantity in any $n^{-1/2}$-neighborhood of $\btheta_{0\model}$. This proves useful, as examining the behavior of the estimated risk on these shrinking neighborhoods will allow us to establish asymptotic properties of our proposed adaptive lasso. 
The following Lemma establishes similar results for the perturbation bootstrap when applied to the squared-error risk.

\begin{lemma}
  \label{lem:perturbed-uniform-conv}
  Let $G_1,\ldots,G_n \iid P_G$ satisfying $\Eop G_1^2 < \infty$ be drawn independently of the data. Then under the setup of \Cref{lem:uniform-conv-loss},
  \begin{align*}
    \sup_{\btheta_{\model} \in \Theta_{\model}} \bigg\vert \frac{1}{n} \sum_{k=1}^K \sum_{i \in \Ik} G_i \big\{ \Lbar_{\model}\big(\btheta_{\model}; \hat\bmu_k, \bO \big) - \Lbar_{\model}\big(\btheta_{\model}; \bmu_0, \bO \big) \big\} \bigg\vert = o_{p^*}(n^{-1/2}) \\
    \sup_{\btheta_{\model} \in \Theta_{\model}} \bigg\Vert \frac{1}{n} \sum_{k=1}^K \sum_{i \in \Ik} G_i \big\{ \nabla L_{\model}\big(\btheta_{\model}; \hat\bmu_k, \bO \big) - \nabla L_{\model}\big(\btheta_{\model}; \bmu_0, \bO \big) \big\} \bigg\Vert_\infty = o_{p^*}(n^{-1/2}) \\
    \sup_{\Vert \bu_{\model} \Vert_\infty \leq C} \bigg\vert \frac{1}{n} \sum_{k=1}^K \sum_{i \in \Ik} G_i \big\{ \Lbar_{\model}\big(\btheta_{0\model} + n^{-1/2}\bu_{\model}; \hat\bmu_k, \bO \big) - \Lbar_{\model}\big(\btheta_{0\model} + n^{-1/2}\bu_{\model}; \bmu_0, \bO \big) \big\} \bigg\vert = o_{p^*}(n^{-1}).
  \end{align*}
  where $o_{p^*}(1)$ is to be understood as measuring probability under the product measure $P_0 \times P_G$.
\end{lemma}
\begin{proof}
  Since \Cref{lem:mean-zero,lem:crossprod-terms} also apply when using perturbations $G$ in the sum, nearly-identical arguments to those in the proof of \Cref{lem:uniform-conv-loss} establish the desired results.
\end{proof}

\subsection{Proof of Theorem \ref{thm:stage1-can}}
\label{sec:prf-thm-stage1-can}

\begin{proof}[Proof of Theorem \ref{thm:stage1-can}]\label{pf:thm-1}
  
  The estimators $\hat\btheta_{\model}$ and $\tilde\btheta_{\model}$ minimize $\hat R_{n \model}(\cdot)$ and $R_{n \model}(\cdot)$, respectively. 
  Both $\hat R_{n \model}(\cdot)$ and $R_{n \model}(\cdot)$ are convex in their arguments, which satisfy (1)-(2) of \Cref{thm:arcones}. To examine (3), we make the change of variables $\btheta_{\model} = \btheta_{0 \model} + n^{-1/2}\bu_{\model}$ and study:
  \begin{align*}
    n \{ R_{n \model}(\btheta_{0 \model} + n^{-1/2}\bu_{\model}) - R_{n \model}(\btheta_{0 \model}) \} =& \bu_{\model}^\top \Big\{ -2 n^{-1/2} \sum_{i=1}^n \epsilon_{\model i} \big(\bZ_i - \bmu_{Z0i} \big)(\model) \Big\} + \bu_{\model}^\top \bH_n(\model) \bu_{\model} \\
    =& -2 \bu_{\model}^\top \bar{\bG}_{n \model} + \bu_{\model}^\top \bH_n(\model) \bu_{\model} \\
    =& -2 \bu_{\model}^\top \bar{\bG}_{n \model} + \bu_{\model}^\top \bH_0(\model) \bu_{\model} + o_p(1),
  \end{align*}
  where the first line follows from \eqref{eq:sub-loss-quad} along with the representation of the gradient and Hessian of $L$, the second results from the substitutions $\bar{\bG}_{n \model} = n^{-1/2} \sum_{i=1}^n \epsilon_i (\bZ_i - \bmu_{Z0i})(\model)$ and $\bH_n$, and the final from $\Vert \bH_n - \bH_0 \Vert_\infty = o_p(1)$ via the LLN. This establishes (3).
  Since $\bar{\bG}_{n \model}$ is a normalized sum with mean $\bzero$ by \eqref{eq:ch1-submodel-error-uncorrelated} of \Cref{prop:ch1-plm-submodel}, CLT arguments imply $\bar{\bG}_{n \model} \CiD \mathbb{G}_1 \sim N(\bzero, \bV_{1 \model} )$ for
  \begin{equation}
    \label{eq:submodel-meat-matrix}
    \bV_{1 \model} = \Eop\left[ \epsilon_{1\model}^2 \left\{ \bZ_1(\model) - \bmu_{Z01}(\model) \right\}^{\otimes 2} \right],
  \end{equation}
  and hence (4) follows from $\bar{\bG}_{n \model}=O_p(1)$. Finally, (5) follows from the eigenvalue restrictions of \Cref{asu:ch1-bounded-gram}. Then \Cref{thm:arcones} implies $\Vert\sqrt{n}(\tilde\btheta_{\model} - \btheta_{0 \model}) - \bH_0^{-1} \bar{\bG}_{n \model}\Vert_{\infty} = o_p(1)$. We may infer from the Continuous Mapping Theorem and Slutsky's Theorem
  \[
    \sqrt{n}(\tilde\btheta_{\model} - \btheta_{0 \model}) \CiD \mathbb{G}_2 \sim N\left( \bm{0}, \left\{ \bH_0(\model) \right\}^{-1} \bV_{1 \model} \left\{ \bH_0(\model) \right\}^{-1} \right).
  \]  

  Next, we use a similar argument to establish the distribution of $\sqrt{n}(\hat\btheta_{\model}-\btheta_{0\model})$. Conditions (1)-(2) of \Cref{thm:arcones} again follow from convexity arguments. We examine the function $n \{ \hat R_{n \model}(\btheta_{0 \model} + n^{-1/2}\bu_{\model}) - \hat R_{n \model}(\btheta_{0 \model}) \}$ for condition (3) by adding and subtracting the terms $n \{ R_{n \model}(\btheta_{0 \model} + n^{-1/2}\bu_{\model}) - R_{n \model}(\btheta_{0 \model}) \}$. For each $\bu_{\model} \in \R^{|\model|}$,
  \begin{multline}
    n \{ \hat R_{n \model}(\btheta_{0 \model} + n^{-1/2}\bu_{\model}) - \hat R_{n \model}(\btheta_{0 \model}) \} = -2 \bu_{\model}^\top \bar{\bG}_{n \model} + \bu_{\model}^\top \bH_0(\model) \bu_{\model} + o_p(1)\\
    + \underbrace{n\{ (\hat R_{n \model} - R_{n \model})(\btheta_{0 \model} + n^{-1/2}\bu_{\model}) - (\hat R_{n \model} - R_{n \model})(\btheta_{0 \model}) \} }_{o_p(1)}
    .
    \label{eq:app-submodel-risk-quadratic}
  \end{multline}
  The first line on the RHS represents the component examined in the known-$\bmu_0$ case while the second line examines the error induced from estimation. This error term is equivalent to the expression \eqref{eq:loss-diff-repr} and is $o_p(1)$ by \Cref{lem:uniform-conv-loss}. This establishes (3); the remaining conditions follow as before. Then we may conclude by this theorem that $\Vert\sqrt{n}(\hat\btheta_{\model} - \btheta_{0 \model}) - \bH_0^{-1} \bar{\bG}_{n \model}\Vert_{\infty} = o_p(1)$. Recalling the previous known-$\bmu_0$ case, this implies both (i) $\sqrt{n}(\hat\btheta_{\model} - \tilde\btheta_{\model}) = o_p(1)$ and (ii) $\sqrt{n}(\hat \btheta_{\model} - \btheta_{0 \model}) \CiD \mathbb{G}_2$.

  For examining $\hat\balpha$, we notice that a risk similar to $\hat R_n(\cdot)$ is involved.  We examine the suitably normalized risk for $M_j$:
  \[
    H_{nj}(\balpha_0\subj[j] + n^{-1/2} u_j) - H_{nj}(\balpha_0\subj[j]) = -2 u_j \left\{ n^{-1/2} \sum_{i=1}^n \eta_{ij} (D_i - \mu_{D0i}) \right\} + u_j^2 \left\{ n^{-1}\sum_{i=1}^n( D_i - \mu_{D0i})^2 \right\}.
  \]
  The total risk corresponding to all $p$ variables in $\bM$ may be derived as the sum of these risks. Define the $p \times 1$ vector $\bar\bG_{nM}\defined \left\{ n^{-1/2} \sum_{i=1}^n \bm{\eta}_{i} (D_i - \mu_{D0i}) \right\}$ and write
  \begin{align*}
    H_{n}(\balpha_{0} + n^{-1/2} \bu) - H_{n}(\balpha_{0}) = -2 \bu^\top \bar\bG_{nM} + \Vert\bu\Vert_2^2 \bH_{0M} + o_p(1).
  \end{align*}
  From here, a similar argument establishes the result.
  Notably, the eigenvalue restrictions on $\bV_{2}$ ensure $\Vert\bar\bG_{nM}\Vert_{\infty}=O_p(1)$ and the positivity assumption ensures $0<\bH_{0M}<\infty$.

\end{proof}

\subsection{Proof of Lemma \ref{lem:rates-determine-behavior}}

Next we move on to establish properties of the adaptive lasso estimator. The main lemma of this section will simplify establishing the properties of various weight functions under different scenarios. Below, we establish some notation.

Represent the known-$\bmu_0$ version of \eqref{eq:est-med-lasso-obj} by
\begin{equation}
  \Lcal_n(\btheta; \hat\bw, \lambda_n) := \frac{1}{n}\sum_{i=1}^n{\Big\{Y_i - {\mu}_{Y0i} - \big( \bZ_i - {\bmu}_{Z0i} \big)^\top \btheta \Big\}^2} + \frac{\lambda_n}{n} \sum_{j=1}^{p+1}{\hat\bw\subj[j] \vert \theta_j \vert},
  \label{eq:med-lasso-obj}
\end{equation}
with associated minimizer $\tilde\btheta^L$.
In the sequel, we assume the weights $\hat\bw = (0, \hat{w}_2,\ldots,\hat{w}_{p+1})$ depend on pilot estimators $\hat{\balpha}, \hat{\bbeta}$ which are $\sqrt{n}$-consistent, as described in \Cref{sec:proposed-loss}. Let us define two sets $\Scal$ and $\Scal^c$, depending on the weights $\hat\bw$, which satisfy
\begin{equation}
  \label{eq:scal-def}
  \begin{aligned}
    \max_{j \in \Scal} n^{-1/2}\lambda_n \hat\bw\subj[j] \CiP 0 \\
    \min_{j \in \Scal^c} n^{-1/2}\lambda_n \hat\bw\subj[j] \CiP \infty.
  \end{aligned}
\end{equation}
We shall assume that these sets partition $\{1,\ldots,p+1\}$. That is, $\Scal$ represents the subset of $\bZ_1$ for which the scaled variable-specific penalty will become asympotically negligible, and $\Scal^c$ represents the remaining subset for which the corresponding penalties diverge. We draw attention to our definition of $\hat w_1 = 0$, so that we always have $1 \in \Scal$.

Let us also define the random matrix $C_{n \Scal} = \Pn(\bZ_{\Scal} - \bmu_{Z0\Scal})^{\otimes 2}$, along with its nonrandom target $C_{0\Scal}=\Eop(\bZ_{\Scal} - \bmu_{Z0\Scal})^{\otimes 2}$. Let $G_{n\Scal} = \sqrt{n} \Pn \big( \epsilon_\Scal (\bZ_{\Scal} - \bmu_{Z0\Scal}) \big)$ where
\begin{align*}
  \epsilon_\Scal =& Y - \mu_{Y0} - \btheta_{0\Scal}^\top (\bZ_{\Scal} - \bmu_{Z0\Scal}) \\
  =& \epsilon + \btheta_0^\top (\bZ - \bmu_{Z0}) - \btheta_{0\Scal}^\top (\bZ_{\Scal} - \bmu_{Z0\Scal}).
\end{align*} 
The vector $\bZ_{\Scal}$ represents the sub-vector of $\bZ$ corresponding to $\Scal$, and $\bmu_{Z0\Scal}$ is defined similarly. 
Recall that $\btheta_{0\Scal} = C_{0\Scal}^{-1} \Eop(\bZ_{\Scal} - \bmu_{Z0\Scal})(Y - \mu_{Y0}) $, ensuring that $\Eop(\epsilon_{\Scal} \vert \bZ_{\Scal}, \bX) = 0$. Based on this relationship, we can see that $G_{n\Scal} \CiD \mathbb{G}_{1\Scal} \sim N(\bzero, V_{1\Scal})$, where $V_{1\Scal} = \Eop\big(\epsilon_{\Scal}^2 (\bZ_{\Scal} - \bmu_{Z0\Scal})^{\otimes 2} \big)$. 

Finally, we will state a useful Proposition before stating the Lemma.
\begin{proposition}
  \label{prop:bound-and-diverge}
  Let $U_n$, $V_n$ be real-valued random variables defined on the same probability space for each $n\geq 1$. Suppose $U_n=O_p(1)$ and $V_n \CiP \infty$. Then $\Prob\left( U_n \geq V_n \right) \rightarrow 0$.
\end{proposition}
\begin{proof}
  By assumption, for any $\delta_1>0$, there exists $M_1$, $N_1$ such that for all $n \geq N_1$, $\Prob(U_n \geq M_1) \leq \delta_1$, and for any $\delta_2>0$, $0 < M_2 < \infty$, there exists $N_2$ such that for all $n \geq N_2$, $\Prob( V_{n} \leq M_2 ) \leq \delta_2$. Decompose the original probability statement and apply inequalities to obtain
  \begin{multline*}
    \Prob\Big( U_n \geq V_n \Big) =  \Prob\Big( U_n \geq V_n, V_n \leq M_2 \Big) + \Prob\Big( U_n \geq V_n, V_n > M_2 \Big) \\
    \leq \Prob\Big( V_n \leq M_2 \Big) + \Prob\Big( U_n \geq M_2, V_n > M_2 \Big) \\
    \leq \delta_2 + \min\Big( \Prob( U_n \geq M_2), 1-\delta_2 \Big),
  \end{multline*}
  where the last line follows from Frechet's inequality.
  Since $M_2,\  \delta_2, $ and $\delta_1$ are arbitrary, we may choose $\delta_1$ and subsequently set $M_2$ such that this last expression is bounded by the arbitrary $\delta = \delta_2 + \delta_1$ for all $n\geq \max(N_1, N_2)$.
\end{proof}

\begin{lemma}
  \label{lem:rates-determine-behavior}

  Assume the mediation model \eqref{eq:med-model} holds, and use the previously defined notation. Let $\tilde\btheta^L$ and $\hat\btheta^L$ minimize \eqref{eq:med-lasso-obj} and \eqref{eq:est-med-lasso-obj}, respectively, with weights $\hat\bw$. Let the sets $\Scal$ and $\Scal^c$ satisfy the asymptotic conditions of \eqref{eq:scal-def}.  
  Under the assumptions of \Cref{thm:stage1-can},
  \begin{align}
    \label{eq:lemma-sel-cons}
    \Prob \big\{ \tilde\btheta^L\subj[j] \neq 0 \text{ for any } j \in \Scal^c  \big\}\longrightarrow 0 \\
    \label{eq:lemma-sel-cons-hat}
    \Prob \big\{ \hat\btheta^L\subj[j] \neq 0 \text{ for any } j \in \Scal^c \big\} \longrightarrow 0
  \end{align}
  and the coefficients $\hat\btheta^L(\Scal)$ corresponding to $\Scal$ follow the oracle asymptotic normal distribution
  \begin{equation}
    \label{eq:lemma-oracle-dist}
    \begin{aligned}
      \big\Vert \sqrt{n}\big\{\hat\btheta^L(\Scal) - \btheta_{0 \Scal}\big\} - \sqrt{n}\big\{\tilde\btheta(\Scal) - \btheta_{0 \Scal}\big\}\big\Vert_\infty = o_p(1)\\
      \sqrt{n}\big\{\tilde\btheta(\Scal) - \btheta_{0 \Scal}\big\} \CiD \mathbb{G}_{2\Scal} \sim N(\bm{0}, C_{0\Scal}^{-1} V_{1\Scal} C_{0\Scal}^{-1}),
    \end{aligned}
  \end{equation}
  with $C_{0\Scal} = \Eop(\bZ_{\Scal} - \bmu_{Z0\Scal})^{\otimes 2}$, and $V_{1\Scal} = \Eop\big(\epsilon_{\Scal}^2 (\bZ_{\Scal} - \bmu_{Z0\Scal})^{\otimes 2} \big)$. 
  Finally, we have
  \begin{equation}
    \label{eq:lemma-mu-approx}
    \begin{aligned}
      \Prob\big\{\tilde\btheta^L(\Scal^c) = \hat\btheta^L(\Scal^c) = \bm{0} \big\}\longrightarrow 1\\
      \Vert\hat\btheta^L - \tilde\btheta^L\Vert_\infty = o_p(n^{-1/2}).
    \end{aligned}
  \end{equation}
\end{lemma}
\begin{proof}[Proof of \Cref{lem:rates-determine-behavior}]
  We follow the general method of proof adopted in \cite{zouAdaptiveElasticnetDiverging2009} and \cite{dasPerturbationBootstrapAdaptive2019d}. First, we will prove that $\tilde\btheta^L(\Scal^c) = \hat\btheta^L(\Scal^c) = \bm{0}$ in probability for large $n$, demonstrating the variable selection consistency results \eqref{eq:lemma-sel-cons}-\eqref{eq:lemma-sel-cons-hat} from the KKT conditions. Then we will show the oracle normal distribution result for the nonzero parameters \eqref{eq:lemma-oracle-dist}. The final statement then follows.

  First, we show variable selection consistency. As discussed in \cite{huangEstimationSelectionAbsolute2012c}, since the functions $\hat R_n(\btheta)$ and $R_n(\btheta)$ are convex in $\btheta$, the KKT conditions are both necessary and sufficient. These authors also give the form of these conditions, which we present for the estimator with $\bmu_0$ known.
  These KKT conditions for $\tilde\btheta^L$, the minimizer of $\Lcal_n$, are given by
  \[
    -\Pn \big\{ \nabla L(\tilde\btheta^L; \bmu_0, \bO) \big\} = \mathbf{g},
  \]
  where $\Pn f = n^{-1} \sum_{i=1}^n f(\bO_i)$ is the empirical measure. We can write the LHS as
  \begin{align*}
    -\Pn\big\{ \nabla L(\tilde\btheta^L; \bmu_0, \bO) \big\} = 2\Pn\Big\{ \big( \epsilon - (\tilde\btheta^L - \btheta_0)^\top (\bZ - \bmu_{Z0}) \big) (\bZ - \bmu_{Z0})  \Big\} \\
    = 2 \Big \{ n^{-1/2} \bar{\bG}_n - \bH_n(\tilde\btheta^L - \btheta_0) \Big\}
  \end{align*}
  which results from the representation of the gradient in the previous section, and where $\bar{\bG}_n \equiv \bar{\bG}_{n \model^F}$ as defined in the proof of \Cref{thm:stage1-can}, and $\bH_n=n^{-1}\sum_{i=1}^n (\bZ_i - \bmu_{Z0i})^{\otimes 2}$ .
  Of particular interest is the form of the vector $\mathbf{g}$. When $\tilde\btheta^L\subj[j] \neq 0$, $\bg\subj[j] = n^{-1} \lambda_n \hat\bw\subj[j] sgn\{\tilde\btheta^L\subj[j]\}$. By multiplying $sgn(\tilde\btheta^L\subj[j])\sqrt{n}$ through both sides of KKT equality, we may examine the probability of non-sparsity in $\model_0^c$ by comparing bounded variables with unbounded ones. Let $\bm{e}_j$ be a conformable column vector of 0 with 1 in the $j^{th}$ position, $U_{nj}=\bm{e}_j^\top\{\bar{\bG}_n - \sqrt{n}\bH_{n}(\tilde\btheta^L - \btheta_0)\}$, $v_{nj}=n^{-1/2} \lambda_n \hat\bw\subj[j]$, and $\bar{v}_n=\min_{j \in \model_0^c} v_{nj}$. Then,
  \begin{equation}
    \label{eq:prob-nonzero}
    \begin{aligned}
      \Prob \big\{ \tilde\btheta^L\subj[j] \neq 0 \text{ for any } j \in \model_0^c \big\} \leq& \Prob \Big( \cup_{j \in \model_0^c} \left\{ 2 sgn(\tilde\btheta^L\subj[j]) U_{nj} \geq v_{nj}  \right\} \Big) \\
      \leq& \sum_{j \in \model_0^c} \Prob\left\{ 2 sgn(\tilde\btheta^L\subj[j]) U_{nj} \geq v_{nj}  \right\} \\
      \leq& \sum_{j \in \model_0^c} \Prob \left( 2 | U_{nj} | \geq v_{nj} \right) \\
      \leq& \sum_{j \in \model_0^c} \Prob \left( 2 U_n \geq v_{nj} \right) \\
      \leq& \vert \model_0^c \vert \Prob\left( 2U_n \geq \bar{v}_n \right),
    \end{aligned}
  \end{equation}
  where the first inequality follows from the necessary KKT equality, the
  second inequality follows from subadditivity, the third from the trivial inequality $A_n \leq |A_n|$ for any random variable $A_n$, the fourth from defining the stochastically larger random variable
  $U_n = \Vert \bar{\bG}_n \Vert_{\infty} + \Vert \sqrt{n}\bH_n \Vert_\infty \Vert \tilde\btheta^L - \btheta_0 \Vert_\infty$, and the final inequality follows from the definition of $\bar{v}_n$.
  First, notice that $\Vert \tilde\btheta^L - \btheta_0 \Vert_\infty = O_p(1)$ by a similar argument as in the proof of Theorem 1 in \cite{knightAsymptoticsLassotypeEstimators2000}. Next, notice that under our modeling assumptions, $\Vert \bar{\bG}_n \Vert_{\infty} $ and $ \Vert \sqrt{n} \bH_n \Vert_\infty$ are both $O_p(1)$. Then $U_n = O_p(1)$, and since $ \bar{v}_{n} \CiP \infty$, we apply Proposition \ref{prop:bound-and-diverge} along with $\vert \model_0^c \vert < \infty$ to arrive at \eqref{eq:lemma-sel-cons}.

  The result for $\hat\btheta^L$, the minimizer of $\hat\Lcal_n$ follows similarly. First, note that an identical argument as before establishes that $\Vert\hat\btheta^L - \btheta_0\Vert_\infty=O_p(1)$. The LHS of the KKT conditions related to the estimator $\hat\btheta^L$ may be written as
  \[
    - \frac{1}{n} \sum_{k=1}^K \sum_{i \in \Ik} \nabla L(\hat\btheta^L; \hat\bmu_k, \bO_i)
  \]
  which by \Cref{lem:uniform-conv-loss} is well-approximated by 
  \[
    - \frac{1}{n} \sum_{k=1}^K \sum_{i \in \Ik} \nabla L(\hat\btheta^L; \bmu_0, \bO_i),
  \]
  up to $o_p(n^{-1/2})$. This last display is equal to $2\{n^{-1/2} \bar{\bG}_n - (\hat\btheta^L - \btheta_0) \bH_n\}$ which only differs from the known-$\bmu_0$ case by the presence of $\hat\btheta^L$ rather than $\tilde\btheta^L$. The arguments from here follow similarly with $\hat U_n=\Vert \bar{\bG}_n  \Vert_{\infty} + \Vert \sqrt{n} \bH_n \Vert_\infty \Vert \hat\btheta^L - \btheta_0 \Vert_\infty + o_p(1)$ to establish \eqref{eq:lemma-sel-cons-hat}. This completes the proof of selection consistency.

  Now we move on to the asymptotic normality.
  We have shown that with arbitrarily high probability and for sufficiently large $n$, the solutions to $\Lcal_n$ and $\hat\Lcal_n$ satisfy $\tilde\btheta^L(\model_0^c) = \bzero $  and $\hat\btheta^L(\model_0^c) = \bzero $, respectively. 
  Consequently, $\tilde\btheta(\model_0)^L$ and $\hat\btheta(\model_0)^L \in \R^{\vert \model_0 \vert}$ are the respective minimizers, in probability, of
  \begin{align}
    \label{eq:appendix-regularized-subloss}
    \Lcal_{n1}(\btheta_{\model_0}; \lambda_n, \hat\bw) = R_{n \model_0}(\btheta_{\model_0}) + n^{-1} \lambda_n \sum_{j \in \model_0} \hat\bw\subj[j] \vert \btheta_{\model_0}\subj[j] \vert \\
    \label{eq:appendix-regularized-subloss-est}
    \hat{\Lcal}_{n1}(\btheta_{\model_0}; \lambda_n, \hat\bw) = \hat{R}_{n \model_0}(\btheta_{\model_0}) + n^{-1} \lambda_n \sum_{j \in \model_0} \hat\bw\subj[j] \vert \btheta_{\model_0}\subj[j] \vert, 
  \end{align}
  where the submodel risks $R_{n \model}(\btheta_{\model_0})$ and $\hat{R}_{n \model}(\btheta_{\model_0})$ are defined as before.

  From this point, we follow the proof of \Cref{thm:stage1-can}. We present an argument for $\hat\btheta^L_{\model_0}$. The argument establishing the result for $\tilde\btheta^L_{\model_0}$ is nearly identical. Note that both \eqref{eq:appendix-regularized-subloss} and \eqref{eq:appendix-regularized-subloss-est} are convex functions of $\btheta_{\model_0} \in \R^{|\model_0|}$, which ensures conditions (1)-(2) of \Cref{thm:arcones}. To establish (3), we examine
  \begin{align}
    \label{eq:ada-loss-subproblem-1}
    n \{\hat\Lcal_{n1}(\btheta_{0\model_0} + n^{-1/2} \bu_{\model_0}; \hat\bw) - \hat\Lcal_{n1}(\btheta_{0\model_0}; \hat\bw) \} = 
    n \left\{ \hat{R}_{n \model}(\btheta_{0\model_0} + n^{-1/2} \bu_{\model_0}) - \hat{R}_{n \model}(\btheta_{\model_0}) \right\} \\
    \label{eq:ada-loss-subproblem-2}
    + \sum_{j \in \model_0} v_{nj} \Big[n^{1/2} \big\{ \vert \btheta_{\model_0}\subj[j] + n^{-1/2} u_{\model_0}\subj[j] \vert - \vert \btheta_{\model_0}\subj[j] \vert \big\} \Big],
  \end{align}
  where we recall $v_{nj}= n^{-1/2} \lambda_n \hat\bw\subj[j]$.
  The RHS of \eqref{eq:ada-loss-subproblem-1} equals $\bu_{\model_0}^\top \bH_{0}(\model_0) \bu_{\model_0} - 2 \bu_{\model_0}^\top \bar{\bG}_{n\model_0} + o_p(1)$ by \eqref{eq:app-submodel-risk-quadratic}.
  The term in square brackets within \eqref{eq:ada-loss-subproblem-2} is upper- and lower-bounded by $\vert u_{\model_0}\subj[j] \vert$ and $-\vert u_{\model_0}\subj[j] \vert$, respectively, for each $u_{\model_0}\subj[j]$. This follows by taking the absolute value and applying the reverse triangle inequality. Now, since $v_{nj} \CiP 0$ for all $j \in \model_0$, \eqref{eq:ada-loss-subproblem-2} converges to zero in probability for each $\bu_{\model_0} \in \R^{|\model_0|}$. This establishes condition (3) of \Cref{thm:arcones}:
  \[
    n \{\hat\Lcal_{n1}(\btheta_{0\model_0} + n^{-1/2} \bu_{\model_0}; \hat\bw) - \hat\Lcal_{n1}(\btheta_{0\model_0}; \hat\bw) \} = \bu_{\model_0}^\top \bH_{0}(\model_0) \bu_{\model_0} - 2 \bu_{\model_0}^\top \bar{\bG}_{n\model_0} + o_p(1).
  \]
  Condition (4) follows from $\bar{\bG}_{n\model_0} \CiD \mathbb{G}_{1\model_0} \sim N(\bzero, \bV_{1\model_0})$, where $\bV_{1 \model_0}$ is defined in \eqref{eq:submodel-meat-matrix}. The final condition (5) follows from the eigenvalue restrictions on $\bH_{0}(\model_0)$. Then we may conclude
  \[
    \Vert \sqrt{n}\{\hat\btheta(\model_0)^L - \btheta_{0\model_0}\} - [\bH_{0}(\model_0)]^{-1} \bar{\bG}_{n\model_0} \Vert_\infty = o_p(1).
  \]
  We combine this result with the Triangle inequality and \Cref{thm:stage1-can} to establish the first part of \eqref{eq:lemma-oracle-dist}. 
  The second part of \eqref{eq:lemma-oracle-dist} follows from the Continuous Mapping Theorem.
  
  The argument of the preceding paragraph follows similarly in the case where $\bmu_0$ is known. The main change is made to line \eqref{eq:ada-loss-subproblem-1}, instead written as
  \[
    n \{\Lcal_{n1}(\btheta_{0\model_0} + n^{-1/2} \bu_{\model_0}; \hat\bw) - \Lcal_{n1}(\btheta_{0\model_0}; \hat\bw) \} = 
    n \left\{ {R}_{n \model}(\btheta_{0\model_0} + n^{-1/2} \bu_{\model_0}) - {R}_{n \model}(\btheta_{\model_0}) \right\} + p_n(\bu_{\model_0})
  \]
  with $p_n(\bu_{\model_0})$ defined by \eqref{eq:ada-loss-subproblem-2}. The remaining arguments follow as before.

  Finally, we show \eqref{eq:lemma-mu-approx}. For the first statement, apply the bound $\Prob(A \cap B) \leq \Prob(A) + \Prob(B)$ for events $A$, $B$ defined in \eqref{eq:lemma-sel-cons} and \eqref{eq:lemma-sel-cons-hat}, respectively. Thus
  $\tilde\btheta^L_{\model_0^c} = \hat\btheta^L_{\model_0^c} = \bm{0}$ with probability converging to one. Using the Triangle Inequality argument above, we have $\Vert \sqrt{n}(\tilde\btheta(\model_0)^L - \hat\btheta(\model_0)^L) \Vert_{\infty} = o_p(1)$, which in combination with the previous probability statement completes the proof of \eqref{eq:lemma-mu-approx}.
\end{proof}

\subsection{Proof of Theorem \ref{thm:oracle-selection-normality}}
\label{sec:app-ch2-pf-oracle-selection}

Lemma \ref{lem:rates-determine-behavior} is flexible and will prove useful in demonstrating asymptotic results in several interesting cases. These demonstrations essentially require straightforward checks on the behavior of $v_{nj}=n^{-1/2} \lambda_n \hat\bw\subj[j]$ for $j=2,\ldots,p+1$. As long as these variables either converge to zero or diverge in probability, we may construct $\Scal$ and $\Scal^c$ and satisfy the above conditions. For the remainder of this section, let $j=1,\ldots,p$ index the elements of $\bbeta$ rather than components of $\btheta$. This also extends to the weights $\hat\bw\subj[j]$, which range from 1 to $p$ rather than to $p+1$ as previously defined. Since we previously defined $\hat w_1 \equiv 0$, implying $1 \in \Scal_Z$ always holds, this re-indexing does not induce any fundamental changes to the proceeding arguments.

\begin{proof}[Proof of Theorem \ref{thm:oracle-selection-normality}]
  Recall $v_{nj}=n^{-1/2}\lambda_n \hat\bw\subj[j]$.
  Based on Lemma \ref{lem:rates-determine-behavior}, we see that it is sufficient to determine appropriate rates for $\lambda_n$ to ensure the $v_{nj}$ converge to either 0 or $\infty$, ensuring that each variable either belongs to $\Scal$ or $\Scal^c$, respectively.

  Suppose $\hat\bw\subj[j] = \vert \hat\balpha\subj[j] \hat\bbeta\subj[j] \vert^{-\kappa}$ for fixed $\kappa > 0$.
  \begin{itemize}
    \item \textbf{Case 1:} $\balpha_0\subj[j] \bbeta_0\subj[j] \neq 0$.
    Since $\hat\bw\subj[j] \CiP \vert \balpha_0\subj[j] \bbeta_0\subj[j] \vert^{-\kappa} < \infty$, it follows that $v_{nj} = n^{-1/2} \lambda_n  O_p(1)$.
    \item \textbf{Case 2: } $\balpha_0\subj[j] = 0$, $\bbeta_0\subj[j] \neq 0$. \[
      n^{-\kappa/2} \hat\bw\subj[j] = \Big[\big\vert \sqrt{n} \hat\balpha\subj[j] \big\{\bbeta_0\subj[j] + o_p(1) \big\} \big\vert \Big]^{-\kappa} = O_p(1)
    \]
    which implies $v_{nj} = n^{(\kappa-1)/2}\lambda_n O_p(1)$.
    \item \textbf{Case 3: } $\balpha_0\subj[j] \neq 0$, $\bbeta_0\subj[j] = 0$. By the symmetry of the weights, a similar argument as above implies $v_{nj} = n^{(\kappa-1)/2} \lambda_n O_p(1)$.
    \item \textbf{Case 4: } $\balpha_0\subj[j] = 0$, $\bbeta_0\subj[j] = 0$. By the above arguments, we have
    \[
      v_{nj} = n^{(2\kappa - 1)/2}\lambda_n / \Big\{ \vert n \hat\balpha\subj[j] \hat\bbeta\subj[j] \vert \Big\}^\kappa = n^{(2\kappa - 1)/2}\lambda_n O_p(1).
    \]
  \end{itemize}
  Recall $\model^* = \{ j : \balpha_0\subj[j] \bbeta_0\subj[j] \neq 0 \}$ represent the set of true mediators. To ensure asymptotic normality for $\model^*$, we require $n^{-1/2}\lambda_n \rightarrow 0$. To remove the complete noise variables in Case 4, we require $n^{(2\kappa - 1)/2}\lambda_n \rightarrow \infty$. Whether the middle two cases result in asymptotic normality or elimination with high probability is controlled by the behavior of $n^{(\kappa-1)/2}\lambda_n \rightarrow 0$ or $\infty$, respectively.

  Now, consider the ADP weights $\hat\bw\subj[j] = \vert \hat\bbeta\subj[j] \vert^{-\kappa}$ for fixed $\kappa > 0$.
  \begin{itemize}
    \item \textbf{Cases 1\&2:} $\bbeta_0\subj[j] \neq 0$. Since $\hat\bw\subj[j] \CiP \vert \bbeta_0\subj[j] \vert^{-\kappa} < \infty$, it follows that $v_{nj} = n^{-1/2} \lambda_n  O_p(1)$.
    \item \textbf{Cases 3\&4: } $\bbeta_0\subj[j] = 0$.
    \[
      n^{-\kappa/2}\hat\bw\subj[j]= \big\vert \sqrt{n} \hat \beta_{j} \big\vert^{-\kappa} = O_p(1),
    \]
    which implies $v_{nj} = n^{(\kappa - 1)/2} \lambda_n O_p(1)$.
  \end{itemize}
  Note that there are only two classes of rates available to distinguish between the cases. To ensure asymptotic normality for $\model^*$, we require $n^{-1/2}\lambda_n \rightarrow 0$. On the other hand, we would require $n^{(\kappa - 1)/2}\lambda_n \rightarrow \infty$ to ensure that the Case 4 variables are deselected with high probability. 
  This implies the ADP weights asymptotically select the mediators $\model^{\ddagger}$ under these rate assumptions.
\end{proof}

\subsection{Proof of Theorem \ref{thm:local-oracle}}

Next we will prove the local asymptotic version of the previous Theorem. We will use Lemma \ref{lem:rates-determine-behavior}, as the local parameters do not affect its development. Recall that the true parameters vary with $n$, with $\balpha_{0n} = \balpha_{0} + \bh_{1n}$, $\bbeta_{0n} = \bbeta_{0} + \bh_{2n}$, and $\gamma_{0n} = \gamma_0 + \bh_{3n}$. For $\ell=1,2$, we have that $\bh_{\ell n}\subj[j] = 0$ if $j \in \tilde\model_{\ell}^c$.
For the purposes of this proof, we use the convention that $c_{\ell j} = r_{\ell n j} = 0$ if $j \in \tilde\model_{\ell}^c$, $\ell=1,2$.
We will appeal to decompositions of the pilot estimates:
  \begin{align*}
    \hat\balpha\subj[j] \hat\bbeta\subj[j] =& {\balpha_{0n}\subj[j] \bbeta_{0n}\subj[j]} + (\hat\balpha\subj[j] - \balpha_{0n}\subj[j]) (\hat\bbeta\subj[j] - \bbeta_{0n}\subj[j]) \\
    &+ \big\{(\hat\balpha\subj[j] - \balpha_{0n}\subj[j]) \bbeta_{0n}\subj[j] + (\hat\bbeta\subj[j] - \bbeta_{0n}\subj[j]) \balpha_{0n}\subj[j]\big\} \\
    =& {\balpha_0\subj[j] \bbeta_0\subj[j] + \bh_{1n}\subj[j] \bbeta_0\subj[j] + \bh_{2n}\subj[j]\balpha_0\subj[j] + \bh_{2n}\subj[j] \bh_{1n}\subj[j]} \\
    &+ \big\{O_p(n^{-1/2})[ \balpha_{0n}\subj[j] + \bbeta_{0n}\subj[j] ] \big\} + O_p(n^{-1}) \\
    =& \balpha_0\subj[j] \bbeta_0\subj[j] \\ 
      &+ O_p\Big( \max\Big\{ \Ind(\bbeta_0\subj[j] \neq 0) r_{1nj}, \Ind(\balpha_0\subj[j] \neq 0) r_{2nj}, r_{1nj}r_{2nj}, \\
    & \hphantom{+ O_p\Big( \max\Big\{ \Ind } \Ind(\balpha_{0n}\subj[j],
    \bbeta_{0n}\subj[j] \neq 0) n^{-1/2}, n^{-1} \Big\}  \Big)
  \end{align*}
  for the product of the pilot estimates, as well as
  \begin{align*}
    \hat\bbeta\subj[j] =& \bbeta_0\subj[j] + \bh_{2n}\subj[j] + (\hat\bbeta\subj[j] - \bbeta_{0n}\subj[j]) \\
    =& \bbeta_0\subj[j] + O_p\Big( \max\Big\{ r_{2nj}, n^{-1/2} \Big\} \Big),
  \end{align*}
  for the $\hat\bbeta$ pilot estimate.

  Since both $r_{1nj}, r_{2nj} \rightarrow 0$, the product $r_{1nj}r_{2nj}$ will always be smaller than either rate individually. Hence, as long as one of $\balpha_0\subj[j],\bbeta_0\subj[j]$ are nonzero, we may safely remove this product term from the maximum above.

\begin{proof}[Proof of Theorem \ref{thm:local-oracle}]
  As before, we will apply Lemma \ref{lem:rates-determine-behavior}. Hence, we simply check the behavior of the variable-specific penalties.

  Suppose $\hat\bw\subj[j] = \vert \hat\balpha\subj[j] \hat\bbeta\subj[j] \vert^{-\kappa}$ for fixed $\kappa > 0$.
  \begin{itemize}
    \item \textbf{Case 1:} $\balpha_0\subj[j] \bbeta_0\subj[j] \neq 0$.
    Little changes from the fixed-coefficient case. $v_{nj} = n^{-1/2}\lambda_n O_p(1)$.
    \item \textbf{Case 2: } $\balpha_0\subj[j] = 0$, $\bbeta_0\subj[j] \neq 0$. Let $R^{(2)}_{nj} = \max\Big\{ r_{1nj}, n^{-1/2} \Big\}^{-1}$. Since $R^{(2)}_{nj} \hat\balpha\subj[j] \hat\bbeta\subj[j] = O_p(1)$, we have
    $v_{nj} = (R^{(2)}_{nj})^{\kappa} n^{-1/2} \lambda_n O_p(1)$.
    \item \textbf{Case 3: } $\balpha_0\subj[j] \neq 0$, $\bbeta_0\subj[j] = 0$. Similarly define $R^{(3)}_{nj} = \max\Big\{ r_{2nj}, n^{-1/2} \Big\}^{-1}$.
    In this case $v_{nj} = (R^{(3)}_{nj})^{\kappa} n^{-1/2} \lambda_n O_p(1)$.
    \item \textbf{Case 4: } $\balpha_0\subj[j] = 0$, $\bbeta_0\subj[j] = 0$. Define $R^{(4)}_{nj} = \max\Big\{ r_{1nj}r_{2nj}, n^{-1/2} \Big\}^{-1}$. Then
    \[
      v_{nj} = (R^{(4)}_{nj})^{\kappa} n^{-1/2}\lambda_n O_p(1).
    \]
  \end{itemize}
  By construction, we have $R^{(2)}_{nj}=n^{c_{1j}}, \ R^{(3)}_{nj}=n^{c_{2j}}, \ \text{ and }R^{(4)}_{nj}=\sqrt{n}$ for $j=1,\ldots,p$. Each of these rates is no larger than $\sqrt{n}$, hence allowing $n^{(\kappa-1)/2}\lambda_n \rightarrow 0$ ensures that the variables in all 4 of the cases above will be included in the model asymptotically.

  For the ADP weights, we can restrict to two cases:
  \begin{itemize}
    \item \textbf{Cases 1\&2:} $\bbeta_0\subj[j] \neq 0$. $v_{nj} = n^{-1/2}\lambda_n O_p(1)$.
    \item \textbf{Cases 3\&4:} $\bbeta_0\subj[j] = 0$. Use $R^{(3)}_{nj}$ as defined above and note $R^{(3)}_{nj}\hat\bbeta\subj[j] = O_p(1)$. Then $v_{nj}=(R^{(3)}_{nj})^{\kappa} n^{-1/2}\lambda_n O_p(1)$. When $\beta_{nj}=0$, this corresponds to $r_{2nj} = 0$ and hence $R^{(3)}_{nj}=\sqrt{n}$. Otherwise, if $c_{2j} = 1/2$, we also have $R^{(3)}_{nj}=\sqrt{n}$.
  \end{itemize}
  
\end{proof}

\subsection{Proof of Theorem \ref{thm:pert-boot}}\label{sec:app-pert-boot}

\begin{proof}[Proof of Theorem \ref{thm:pert-boot}]
  The proof of this Theorem follows nearly identically to the proof in \cite{minnierPerturbationMethodInference2011}. In Section A.1 of that reference, the authors essentially study the behavior of the following quantity in order to establish the perturbation-bootstrap distribution:
  \[
    n \{\Lcal_{n}^b(\btheta_0 + n^{-1/2}\bu; \hat\bw^b, \lambda_n) - \Lcal_{n}^b(\btheta_0; \hat\bw^b, \lambda_n)\},
  \]
  where $\Lcal_{n}^b$ is defined similarly to $\Lcal_n$ with multipliers in the risk. Define $\tilde\btheta^{Lb}$ and $\tilde\balpha^b$ as in \eqref{eq:pert-boot-est} but using the true $\bmu_0$ rather than the estimates.
  Following along these arguments, we would establish
  \begin{itemize}
    \item[(i') ] $\Prob\big\{\tilde\btheta^{Lb}\subj[j] \neq 0 \text{ for any } j \in \model^{*c}\big\} \rightarrow 0$ and $\sqrt{n}\big\{\tilde\btheta^{Lb}({\model^*}) - \tilde\btheta_{\model^*} \big\} \CiD N(\bzero, \Jcal_{1\model^*})$.
    \item[(ii') ] $\sqrt{n}(\tilde{\balpha}^{b} - \tilde{\balpha})({\model^*}) \CiD N(\bzero, \Jcal_{2\model^*})$. 
  \end{itemize}
  One implication of these results is that $\tilde\btheta_{\model^*}^{L*}$ and $\tilde{\balpha}^{b}$ are both  $O_{p^*}(1)$, where we use the subscript of $p^*$ to indicate that the implied probability statements should be taken with respect to the product measure of the observed data and the multipliers.
  By Lemma \ref{lem:perturbed-uniform-conv}, and under our assumptions, 
  \[
    \sup_{\Vert \bu \Vert_\infty \leq C} \Big| \Pn^*\left\{ \bar{L}(\btheta_0 + n^{-1/2} \bu; \hat{\bmu}, \bO) - \bar{L}(\btheta_0 + n^{-1/2} \bu; \bmu_0, \bO) \right\} \Big| = o_{p^*}(n^{-1}).
  \]
  This implies that $n [\hat\Lcal_{n}^b(\btheta_0 + n^{-1/2}\bu; \hat\bw^b, \lambda_n) - \hat\Lcal_{n}^b(\btheta_0; \hat\bw^b, \lambda_n)]$ is a uniformly $o_{p^*}(1)$ approximation to $n [\Lcal_{n}^b(\btheta_0 + n^{-1/2}\bu; \hat\bw^b, \lambda_n) - \Lcal_{n}^b(\btheta_0; \hat\bw^b, \lambda_n)]$ over compacts of $\bu$.
  Hence the result (i) in the Theorem statement can be inferred from (i'). We can also apply this same logic to the un-penalized risks derived from the loss functions $L_j(\alpha_j; \bmu, \bO_i) = \Big( \bM_{i}\subj[j] - \mu_{Mj}(\bX_i) - \balpha\subj[j] \{ D_i - \mu_{D}(\bX_i) \} \Big)^2$ for $j=1,\ldots,p$ in order to establish (ii) from (ii').
  The final result (iii) follows from the previous results along with the Delta Method.
\end{proof}

\newpage

\section{Appendix: Natural Effects for Submodels}
\label{sec:app-nde-nie}

Here, we prove the natural effect representations for submodels $\model \subseteq \model^F$. 
Under the causal \cref{asu:1,asu:2,asu:3,asu:4} for submodel $\model$, we may follow the development of \cite{vanderweeleMediationAnalysisMultiple2014} and use Pearl's mediation formula \citep{pearlDirectIndirectEffects2001}:
\begin{multline}
  \Eop(Y_{d, \bM_{\model, d^*}} ~|~ \bX) = \int_{\bmm_{\model}}  \Eop\left( Y ~|~ \bX, D=d, \bM_{\model} = \bmm_{\model} \right) dP_0(\bmm_{\model} ~|~ \bX, D=d^*) \\
  = \mu_Y(\bX) + d \gamma_{0\model} + \left[ \int_{\bmm_{\model}} \{\bM_{\model} - \bmu_{M \model}(\bX)\}^\top  dP_0(\bmm_{\model} ~|~ \bX, D=d^*) \right] \bbeta_{0 \model} \\
  = \mu_Y(\bX) + d \gamma_{0\model} + \Eop \left[  \{\bM_{\model} - \bmu_{M \model}(\bX)\}^\top~|~ \bX, D=d^* \right] \bbeta_{0 \model} \\
  = \mu_Y(\bX)  + d \gamma_{0\model} + d^* \balpha_{0\model}^\top \bbeta_{0\model}.
\end{multline}
The second line follows from \eqref{eq:med-submodel} along with the requirement that $\Eop(\epsilon_{\model^*} ~|~ \bX, D, \bM_{\model^*})=0$. The NDE fixes $d^*$ and contrasts $d=1$ vs $d=0$. The NIE fixes $d$ and varies $d^*=1$ vs $d^*=0$.
The representations \eqref{eq:nde}-\eqref{eq:nie} follow directly.

\newpage

\section{Remarks on Assumptions for Mediator Selection}
\label{sec:app-med-assumps}

\begin{remark}
  \label{rem:ch2-assump-known-mstar}
  \Cref{asu:med-submodel-correct,asu:1to4-mstar} are sufficient to achieve the second and third goals in the case where $\model^*$ is known exactly. However, we are using the data in the full model $\model^F$ to search for associations and establish an appropriate set of mediators. Without \Cref{asu:1to4-full}, we have not ruled out the possibility that some of the collected variables indexed by $\model^{*c}$ have associations attributable to confounding. This may lead to spurious members of $\model^*$.
\end{remark}
\begin{remark}
  \label{rem:ch2-discussion-causal-assumptions}
  \Cref{asu:1,asu:2,asu:3,asu:4} for $\model^*$ generally follow if \Cref{asu:1,asu:2,asu:3,asu:4} hold for $\model^F$ and none of $\bM(\model^{*c})$ confound the relationship between $\bM(\model^*)$ and either of $D$ or $Y$, given pre-treatment covariates $\bX$. 
\end{remark}
\begin{remark}
  \label{rem:ch2-correctly-specified-med-model}
  We might also relax the requirement in \Cref{asu:1to4-full} that \eqref{eq:med-model} is correctly-specified. Since we take the view that $\nie_{\model^F}$ is not directly interesting, there is no need to provide a valid causal estimate using the observed data for this quantity. However, we still require \Cref{asu:med-submodel-correct} in order to estimate and define the $\model^*$-based estimands. 
  
  As an example, suppose $\alpha_{0p}=0$ and the outcome-model term $\beta_{0p} M_p$ was replaced by some possibly-nonlinear function $\phi(M_p) \in L_2(P_0)(\R)$. In this case, we would require $\phi(M_p)$ to effectively be included in the submodel errors in order to satisfy \Cref{asu:med-submodel-correct}. Such a situation might occur if \( M_p \independent \{ \bM(\model^*), D \} ~|~ \bX\).
\end{remark}

\begin{remark}
  \label{rem:app-violations-of-asu-2}
  We remarked in \Cref{sec:ch2-notation-estimand} that a slightly different estimand may obviate \Cref{asu:4}. Even in this case, \Cref{asu:2} is not trivially satisfied for the submodel $\model^*$ even if all four assumptions hold for the full model $\model^F$.

  To see this, consider the following two-mediator causal model:
  \begin{align*}
    Y_{d, m_1, m_2} =& \beta_1 m_1 + \beta_2 m_2 + \epsilon \\
    M_{1d} =& \alpha_1 d + \eta_1 \\
    M_{2d} =& \eta_2.
  \end{align*}
  We assume that the observed treatment $D$ is randomly assigned, and that the set of baseline confounders $\bX$ is empty.
  In the model above, $M_1$ is affected by treatment, although there is a strong lack of treatment effect for $M_2$, as both potential outcomes are equivalent for every subject. Consequently, $\model^F = \{1,2\}$ and $\model^*=\{1\}$.
  The observed data is $Y = Y_{D, M_1, M_2}$, $M_1 = M_{1D}$, and $M_2=M_{2D}$.
  Under \Cref{asu:1to4-full}, we must have $\epsilon \independent \{ \eta_1, \eta_2 \}$. 
  
  Now let's examine if \Cref{asu:2} holds in $\model^*$. To do this, we must look at the potential outcomes resulting from control only of $d$ and $m_1$. This results in the potential outcomes
  \[
    Y_{d, m_1} = Y_{d, m_1, M_2} = \beta_1 m_1 + \beta_2 \eta_2 + \epsilon.
  \]
  Due to the randomization of $D$ and empty $\bX$, we can simplify the requirement of the relevant assumption to
  \[
    Y_{d, m_1} \independent M_1,
  \]
  which upon substitution of the previous display becomes
  \[
    \beta_1 m_1 + \beta_2 \eta_2 + \epsilon \independent \alpha_1 D + \eta_1.
  \]
  Since $D$ is randomized, we may effectively remove it from the RHS. The remaining requirement is then that
  $\eta_1 \independent \beta_2 \eta_2 + \epsilon$.
  Since \Cref{asu:1to4-full} ensures that $\eta_1 \independent \epsilon$, the last piece is $\eta_1 \independent \beta_2 \eta_2$.

  Therefore, in order for the causal assumptions to hold for $\model^*$, we must have that either $\beta_2=0$ or $\eta_2 \independent \eta_1$. Correlations between the post-treatment variables may violate the unmeasured confounding \Cref{asu:2} if the excluded variables are predictive of the outcome.

\end{remark}

\section{Numerical Equivalence Between Submodel and Full-model Effects}
\label{sec:natural-effect-submodel-equal}

Using the Model \eqref{eq:med-model}, we have that 
\[
  (\bZ - \bmu_{Z0})(\model_Z^*) = (D - \mu_{D0}) \begin{pmatrix} 1 \\ \balpha_{0} \end{pmatrix}(\model_Z^*) + \begin{pmatrix} 0 \\ \bm{\eta} \end{pmatrix}(\model_Z^*)
\]
with $\Eop(\bm{\eta} | \bX, D) = 0$ and $\Eop(D - \mu_{D0} | \bX) = 0$.

Recall $\bH_0 = \Eop(\bZ - \bmu_{Z0})^{\otimes 2}$. Write $\btheta_{0 \model^*} = [\bH_0(\model_Z^*)]^{-1} \bH_0(\model_Z^*,\model_Z^F) \btheta_0$, where $(\model_Z^*,\model_Z^F)$ represents subsetting the rows by $\model_Z^*$, but including all columns in $\model_Z^F$. 
Using the conditional expectations above and some matrix algebra, we can show
\[
  \bH_0(\model_Z^*,\model_Z^F) = 
  \Eop \left\{  (D - \mu_{D0})^2 \right\}
   \begin{pmatrix} 1 \\ \balpha_{0}(\model^*) \end{pmatrix}  
   \begin{pmatrix} 1 & \balpha_{0}^\top \end{pmatrix}
  + \begin{pmatrix} 0 \\ \bm{\eta}(\model^*) \end{pmatrix} \begin{pmatrix} 0 & \bm{\eta}^\top \end{pmatrix}
\]
A similar decomposition may be shown to hold for $\bH_0(\model^*)$. This allows us to decompose $\btheta_{0 \model^*}$ into separate terms depending on $\btheta_0(\model_Z^*)$ and $\btheta_0(\model_Z^{*c})$:
\begin{multline*}
  \btheta_{0 \model^*} = \btheta_0(\model^*) \\ + [\bH_0(\model_Z^*)]^{-1} \bigg\{ 
    \Eop \left\{  (D - \mu_{D0})^2 \right\}
  \begin{pmatrix} 1 \\ \balpha_{0}(\model^*)  \end{pmatrix} \balpha_{0}(\model^{*c})^\top +
  \Eop\left[ \begin{pmatrix} \bzero^\top \\ \bm{\eta}(\model^*) \bm{\eta}(\model^{*c})^\top \end{pmatrix} \right]  \bigg\}\bbeta_0(\model^{*c}) \\
  = \btheta_0(\model^*) + [\bH_0(\model_Z^*)]^{-1} ~
  \Eop\left[ \begin{pmatrix} \bzero^\top \\ \bm{\eta}(\model^*) \bm{\eta}(\model^{*c})^\top \end{pmatrix} \right] \bbeta_0(\model^{*c}).
\end{multline*}
The final line follows because $\balpha_{0}(\model^{*c}) \cdot \bbeta_0(\model^{*c}) = 0$. This leads to the stated conditions in \Cref{rem:natural-effect-submodel-equal}

\end{document}